\documentclass[12pt,journal,onecolumn,draftcls]{IEEEtran}
\usepackage[cmex10]{amsmath}
\usepackage{amssymb}
\usepackage{amsfonts,bm}
 \usepackage{bbm}

\usepackage{hyperref}
 
\usepackage{xcolor}
\usepackage{soul}
\usepackage{algorithm}
\usepackage{algpseudocode}
\usepackage{thmtools, thm-restate}

\newtheorem{lemma}{Lemma}
\newtheorem{theorem}{Theorem}
\newtheorem{proposition}{Proposition}
\newtheorem{corollary}{Corrollary}
\newtheorem{definition}{Definition}

\newenvironment{proof}[1][]{\textit{Proof#1: }}{\hfill$\square$\\}
\newcounter{remarkcnt}
\newenvironment{boldremark}[1][]{\refstepcounter{remarkcnt}\par\smallskip
   \noindent \textbf{Remark~\theremarkcnt. #1} \rmfamily}{\smallskip}
\newcounter{interpcnt}

\newcommand{\indep}{\perp \!\!\! \perp}

\def\floor#1{\lfloor #1 \rfloor}
\def\1{\bm{1}}

\def\vzero{{\bm{0}}}
\def\vone{{\bm{1}}}

\def\va{{\bm{a}}}
\def\vb{{\bm{b}}}
\def\vc{{\bm{c}}}

\def\ve{{\bm{e}}}

\def\vx{{\bm{x}}}
\def\vy{{\bm{y}}}

\def\mA{{\bm{A}}}
\def\mB{{\bm{B}}}
\def\mC{{\bm{C}}}

\def\mF{{\bm{F}}}
\def\mG{{\bm{G}}}

\def\mI{{\bm{I}}}

\def\mK{{\bm{K}}}

\def\mM{{\bm{M}}}
\def\mN{{\bm{N}}}

\def\mP{{\bm{P}}}

\def\mT{{\bm{T}}}
\def\mU{{\bm{U}}}
\def\mV{{\bm{V}}}

\def\mLambda{{\bm{\Lambda}}}
\def\mSigma{{\bm{\Sigma}}}

\DeclareMathAlphabet{\mathsfit}{\encodingdefault}{\sfdefault}{m}{sl}
\SetMathAlphabet{\mathsfit}{bold}{\encodingdefault}{\sfdefault}{bx}{n}

\def\sA{{\mathcal{A}}}
\def\sB{{\mathcal{B}}}
\def\sC{{\mathcal{C}}}

\def\sR{{\mathcal{R}}}
\def\sS{{\mathcal{S}}}

\def\sW{{\mathcal{W}}}

\DeclareMathOperator*{\argmax}{arg\,max}

\usepackage{graphicx}
\usepackage{subfigure}
\usepackage{algorithm} 
\usepackage{pifont}% http://ctan.org/pkg/pifont

\usepackage{xcolor}

\usepackage{algorithm}
\usepackage{booktabs}         % fixes poor spacing in tables and
\allowdisplaybreaks

\begin{document}

\title{Common Information Dimension} 

\author{\IEEEauthorblockN{Xinlin Li, Osama Hanna, Suhas Diggavi and Christina Fragouli\\ 
University of California, Los Angeles\\
Email: \{xinlinli, ohanna, suhasdiggavi, christina.fragouli\}@ucla.edu}\thanks{This paper was presented in part at the 2023 IEEE International Symposium on Information Theory and the 2024 IEEE International Symposium on Information Theory.}
}

\maketitle
\begin{abstract}
Quantifying the common information between random variables is a fundamental problem with a long history in information theory. Traditionally,  common information is measured in number of bits and thus such measures are mostly informative when the common information is finite. However, the common information between continuous variables can be infinite; in such cases, a real-valued random vector $W$ may be needed to represent the common information, and to be used for instance for distributed simulation. In this paper, we propose the concept of Common Information Dimension (CID) with respect to a given class of functions $\mathcal{F}$, defined as the minimum dimension of a random vector $W$ required to distributively simulate a set of random vectors $X_1, \cdots, X_n$, such that $W$ can be expressed as a function of $X_1, \cdots, X_n$ using a member of $\mathcal{F}$. We compute the common information dimension for jointly Gaussian random vectors in a closed form, with $\mathcal{F}$ being the linear functions class. We also analytically prove, under three different formulations, that the growth rate of common information in the nearly infinite regime is determined by the common information dimension, for the case of two Gaussian vectors. 
\end{abstract}

\section{Introduction}
Quantifying the common information between random variables is a fundamental problem with a long history in information theory~\cite{watanabe1960information, gacs1973common, wyner1975common, bell2003co, yu2017generalized}, and has found application in diverse areas including source coding \cite{wyner1974source,xu2015lossy,erixhen2022gaussian},  cryptography \cite{ahlswede1993common,ahlswede1998common,maurer1993secret,csiszar2000common} and multimodal learning \cite{li2018survey,sutter2021generalized,ramesh2022hierarchical,kleinman2023gacskorner}. Multiple information theoretical notions have been developed to measure the common randomness, for instance,  G{\'a}cs-K{\"o}rner's common information \cite{gacs1973common}, Wyner's common information\cite{wyner1975common}, common entropy \cite{kumar2014exact}, and coordination capacity \cite{cuff2010coordination} (see also the monograph \cite{yu2022common}); but as far as we know, all of them measure common information in terms of bits. In this paper, we introduce the notion of common information dimension, that uses dimensionality instead of bits to quantify the common randomness for continuous random variables, and establish a connection with the common information.
%A popular operational meaning comes from distributed simulation, where the common information captures the amount of shared randomness needed to simulate a joint target distribution \cite{wyner1975common}. In this paper, we promote our understanding in the regime where the common information between random variables can be (or can approach) infinity and introduce the notion of common information dimension, that uses dimensionality instead of bits to characterize common randomness for continuous random variables.

%The common randomness between dependent random variables is a fundamental problem in information theory \cite{ watanabe1960information, gacs1973common, wyner1975common, bell2003co} and has ubiquitous applications in a number of areas, such as key generation in cryptography\cite{ahlswede1993common,ahlswede1998common,maurer1993secret,csiszar2000common}, hypothesis testing in statistical inference \cite{hotelling1992relations, csiszar2011information, shayevitz2011renyi, yu2022common}, and multi-modal representation learning in machine learning \cite{li2018survey, wu2018multimodal,shi2019variational,sutter2021generalized,zhang2019cpm}. %xu2013survey, sun2013survey, suzuki2016joint, sutter2020multimodal, ramesh2022hierarchical, tian2018cr
% have operational meaning 

We will illustrate the necessity of a new notion of common randomness through an example. Let $X=[X_1,V]^\top, Y=[Y_1,V]^\top$ be two Gaussian random vectors with $X_1, Y_1, V$ being independent scalar variables.  The common information between $X$ and $Y$ is captured by $V$, which is a continuous scalar variable with infinite entropy - and thus the common information between $X$ and $Y$, as calculated for instance in \cite{satpathy2015gaussian}, is also infinite. For general Gaussian random vectors, we distinguish two cases: if $\text{rank}(\mSigma_X) + \text{rank}(\mSigma_Y) = \text{rank}(\mSigma)$, where  $\mSigma_X =\mathbb{E}(XX^\top)$,  $\mSigma_Y=\mathbb{E}(YY^\top)$ and $\mSigma$ is the joint covariance matrix of the vector $[X\; Y]$, then the common information can be described using a finite number of bits; while if $\text{rank}(\mSigma_X) + \text{rank}(\mSigma_Y) > \text{rank}(\mSigma)$, the common information measured in bits becomes infinite (in this second case, we say that $X$ and $Y$ are {\em jointly singular}). To the best of our knowledge, there is no proposed metric that distinguishes between different amounts of common information in the infinite bits regime.
%To address this, in our recent work \cite{hanna2023common} we introduced the notion of common information dimension $d(X,Y)$, and showed that, for Gaussian variables, $d(X,Y)$ can be calculated as  \[d(X,Y) = \text{rank}(\mSigma_X) + \text{rank}(\mSigma_Y) - \text{rank}(\mSigma),\] where  $\mSigma_X =\mathbb{E}(XX^\top)$,  $\mSigma_Y=\mathbb{E}(YY^\top)$ and  $\mSigma$ is the joint covariance matrix of the vector $[X\; Y]$.

%Why do we need a new notion of common randomness? All of the above discussed notions of common randomness are expressed in terms of ``bits", and thus they are only meaningful, for example in distributed simulation, when finite common randomness is sufficient. Nevertheless, in the general case of continuous variables, an infinite amount of randomness may be required. For instance, \cite{xu2015lossy} calculated the Wyner's common information of bivariate Gaussian random variables $X, Y$ as $C_{\text{Wyner}}(X,Y) = \frac{1}{2}\log\frac{1+\rho}{1-\rho}$, where $\rho$ is the correlation coefficient, and $C_{\text{Wyner}}(X,Y)$ goes to infinity as $\rho$ goes to~$1$. To the best of our knowledge, there is no proposed metric that distinguishes between different problem complexities when an infinite amount of randomness is needed.

Our observation is that,
for the latter case, a real-valued random vector $W$ may be needed to represent the common randomness, and its dimension could provide guidance for practical applications. This is akin to compressed sensing \cite{donoho2006compressed,candes2006stable}, where we seek a low dimensional representation in a high dimensional space. 

Note that if we do not impose structural assumptions (on how $W$ depends on variables $X_1,\ldots,X_n$), the minimum dimension cannot exceed one. This is due to the existence of measurable bijections between $\mathbb{R}$ and $\mathbb{R}^d$ for any $d\geq 1$. However, these functions are not implementable and unstable, as noted in \cite{wu2010renyi}, and hence, are not useful for applications. Therefore, regularity constraints on the common variable need to be considered. In particular, in our definition, we allow for a restriction of the form $W=g(X_1,\cdots,X_n)$ for some $g\in \mathcal{F}$, where $\mathcal{F}$ is a given class of functions. For example, $\mathcal{F}$ may be chosen to be the set of linear, or smooth functions. If $\mathcal{F}$ is chosen to contain all possible functions, then the minimum dimension of $W$ will be upper bounded by one, as previously explained.

In addition to exploring what is the common information dimension of random variables, 
%noticing the discontinuity in the space of common information as described above, we also promote our understanding in the regime where the common information between random variables can be (or can approach) infinity 
we consider the regime where the common information between random variables can be (or can approach) infinity and ask three questions: (1) How fast does the common information grow, from a finite to an infinite number of bits, as the dependency between variables increases? 
 (2) How well can we ``approximately" simulate a pair of random variables $(X,Y)$ using a finite number of shared bits, even though their common randomness is infinite? and
 (3) If the target continuous random variables are quantized into discrete values, how large is the resulting common information between the quantized variables for a certain quantization resolution?

{\bf Contributions.} Our main contributions in this paper are as follows: 
\begin{itemize}
    \item We propose the concept of Common Information Dimension (CID), defined as the minimum dimension of a random variable, with respect to a class of functions, required to distributively simulate a set of random vectors  $X_1,...,X_n$. We define the R{\'e}nyi common information dimension (RCID) as the minimum R{\'e}nyi dimension of a random variable, with respect to a class of functions, required to distributively simulate $X_1,...,X_n$. We define the G{\'a}cs-K{\"o}rner's common information dimension (GKCID) as the maximum R{\'e}nyi dimension of a common function that can be extracted from each random variable individually.
    \item  We prove that for jointly Gaussian random vectors and $\cal{F}$ being the class of linear functions, CID and RCID coincide, and GKCID is upper bounded by CID. Moreover, we give closed-form solutions for the CID, RCID, and GKCID in such case and {an efficient} method to construct $W$ with the minimum dimension that enables the distributed simulation of $X_1,\cdots,X_n$. These can be computed by examining ranks of covariance matrices.
    \item We characterize the common information in the nearly infinite regime by considering a sequence of nearly singular Gaussian pairs with decreasing distances to a jointly singular target random variables $X,Y$. We prove in closed form that the common information of such sequences grows proportionally to $d(X,Y)$, which establishes a link between common information and the common information dimension.
    % \item We characterize the common information of Gaussian vectors in the nearly infinite regime using three different formulations, and establish a link between common information and the common information dimension $d(X,Y)$. This offers a new interpretation for $d(X,Y)$.
    \item We define the approximate common information as the minimum amount of common information between random variables that approximate a target distribution within a given error. We prove that, in this case as well, the growth rate of the approximate common information is proportional to the common information dimension. Our result quantifies the number of shared bits needed to distributedly simulate jointly singular distributions within a desired accuracy. We illustrate this through numerical evaluations in Section~\ref{sec:numeric}.
    % \item Our formulation of approximate common information also helps understand the quantity of shared bits needed for a distributed simulation to achieve a desired level of accuracy. We illustrate this through numerical evaluations in Section~\ref{sec:numeric}, where we show for instance that to simulate a target distribution with $d(X,Y)=5$ (as described in Section~\ref{sec:numeric} Example 2), with accuracy $2^{-5}$ we need to share $19$ bits, while we can achieve a (very high) accuracy of $2^{-20}$ using around $53$ bits.
    \item Moreover, we show that the common information between a uniformly quantized pair of Gaussian random variables also grows as a function of $d(X,Y)$. This result provides  guidance on the number of shared bits needed for the distributed simulation under different quantization precisions. %This formulation captures the decimal-digit precision in practical representations (i.e., computer systems store and process data in digital form, which inherently quantize data into discrete digital values). Our result therefore provides a guidance on the number of shared bits needed for the distributed simulation under different system precision. Details are presented in Example 2 in Section~\ref{sec:numeric}  as well.}
    %\item Our formulation of quantized common information captures the decimal-digit precision in practical representations (i.e., computer systems store and process data in digital form, which inherently quantize data into discrete digital values). Our result therefore provides a guidance on the number of shared bits needed for the distributed simulation under different system precision (details are presented in Section~\ref{sec:numeric} Example 2 as well).
\end{itemize}

{\bf Paper organization.} We review the related work and results on the common information of Gaussian vectors in Section \ref{sec:related}. We introduce the definitions of common information dimension in Section \ref{sec:def}. We present our results in calculating the common information dimension in Section \ref{sec:calc_CID}. We state our problem formulation and results in the asymptotic behavior of common information in Section \ref{sec:approx}. We present our numerical evaluation in Section \ref{sec:numeric}
and conclude the paper in Section~\ref{sec:concl}. We include proof outlines in the paper and detailed proofs in the Appendix.
\section{Related Work and Preliminaries}\label{sec:related}
\subsection{Common Information}
G{\'a}cs-K{\"o}rner's common information \cite{gacs1973common} and Wyner's common information\cite{wyner1975common} are perhaps the most classical notions of common information. 
G{\'a}cs-K{\"o}rner's common information  is defined as the maximum number of bits per symbol of the information that can be individually extracted from two dependent discrete variables $X, Y$, namely
\begin{equation}
    C_{\text{GK}}(X,Y):= \max_{f,g: f(X)=g(Y)} \quad H(f(X)),
\end{equation}
where $f$ and $g$ are deterministic functions. However, it is known from \cite{gacs1973common,witsenhausen1975sequences} that $C_{\text{GK}}(X,Y)$ equals zero except for one special case where $X=(X',V)$, $Y=(Y',V)$ and $X',Y',V$ are independent variables.

Wyner's common information was originally defined for a pair of discrete sources $(X, Y) \sim \pi_{XY}$ as
\begin{equation}\label{def:wyner}
    C_{\text{Wyner}}(X,Y) := \min_{P_W P_{X|W} P_{Y|W}:P_{XY}=\pi_{XY}} \quad I(X\;Y;W).
\end{equation}
Wyner \cite{wyner1975common} provided two operational interpretations. One is for the \textit{source coding} problem: the minimum common rate for the lossless source coding problem over the Gray-Wyner network, subject to a sum rate constraint. The other is for the \textit{distributed simulation} problem: the minimum amount of shared randomness to simulate a given joint distribution $\pi_{XY}$. Recently, the works in \cite{liu2010common} and \cite{cuff2010coordination} generalized Wyner's common information to $n$ discrete random variables settings;  and the works \cite{viswanatha2014lossy, xu2015lossy} and \cite{li2017distributed,yu2020exact}  generalized its interpretations to continuous sources in lossy source coding  and distributed simulation, respectively. 

 Wyner's distributed simulation  assumes codes of large block length (i.e., multi-shot) and approximate generation: the relative entropy between the generated distribution and the target distribution goes to zero as the block length goes to infinity. Kumar, Li and El Gamal \cite{kumar2014exact} extended Wyner's work to define the exact common information (also called common entropy) which requires a single-shot (i.e., block length 1) and exact generation of $\pi_{XY}$. The common entropy is defined as
\begin{equation}
    G(X,Y):= \min_{P_W P_{X|W} P_{Y|W}:P_{XY}=\pi_{XY}} H(W).
\end{equation}
To generalize this to the multi-shot (asymptotic) setting, they also defined the exact common information rate as
\begin{equation}
    C_{\text{Exact}(X,Y)}:= \lim_{n \to \infty} \frac{1}{n}G(X^n,Y^n).
\end{equation}
The exact common information was extended to $n$ continuous variables in \cite{li2017distributed}, and was shown to provide an upper bound on Wyner's common information in \cite{kumar2014exact}.

The calculation of Wyner's common information and its variants are challenging in general since it involves optimizing a concave function over a non-convex set. Therefore, closed-form solutions are available only for special cases \cite{wyner1975common,witsenhausen1976values}. In particular, for continuous sources, closed-form solution is known only for Gaussian sources. \cite{xu2015lossy} calculated it for bivariate Gaussian and the multivariate one with certain correlation structure, while \cite{satpathy2015gaussian} and \cite{erixhen2022gaussian} extended it to a pair of Gaussian vectors. %And \cite{li2017distributed} established a one-shot upper bound on the common entropy of $n$ continuous variables with log-concave pdf. 
The general formula of Wyner's common information between Gaussian vectors $X \in \mathbb{R}^n$ and $Y \in \mathbb{R}^n$ is
\begin{equation}\label{eq:wyner}
    \begin{aligned}
        C_{\text{Wyner}}(X, Y) = \frac{1}{2} \sum_{i=1}^{n} \log\frac{1+\rho_i}{1-\rho_i},
    \end{aligned}
\end{equation}
where $\rho_i$'s are the singular values of  the normalized cross-covariance matrix $\mSigma_X^{-1/2} \mSigma_{XY} \mSigma_Y^{-1/2}$, and $\mSigma_X^{-1/2}, \mSigma_Y^{-1/2}$ are defined using pseudo-inverse when needed.
Observe that when $X$ and $Y$ are jointly singular (i.e., \eqref{eq:singular} holds and thus $\rho_i = 1$ for some $i$ in \eqref{eq:wyner}), the Wyner's common information $C_{\text{Wyner}}(X,Y)$ is infinite.

Wyner also describes two natural relaxations in \cite{wyner1975common}: (i) one replaces the conditional independence with a bounded conditional mutual information; (ii) the other allows a small distance between the generated and the target distributions, measured by Kullback–Leibler (KL) divergence. However, these were only analyzed in discrete settings. %Although Wyner’s common information has been generalized to continuous sources, the relaxed formulations have not been generalized. 
Recently, \cite{erixhen2022gaussian} studies the first relaxation in the case of Gaussian random variables. However, this version of the relaxed common information is still infinite when singular distributions are involved. In a separate study, \cite{hanna2022can} explores a  related, but different, problem of exchanging a small number of bits to break/reduce the dependency between distributed source.
On the other hand, this paper considers relaxation (ii) (with a different distance\footnote{Note that the KL divergence between any singular and non-singular distributions is always infinite, it is not suitable for the task of approximating a singular distribution with a non-singular one.}) which allows an approximate generation when the sources can be continuous and the distributions may be singular.

\subsection{Dimension}
In our work, we consider two notions of dimension: the number of elements of a vector and the information dimension (also called R{\'e}nyi dimension). The R{\'e}nyi dimension is an information measure for random vectors in Euclidean space that was proposed by R{\'e}nyi in 1959 \cite{renyi1959dimension}. It characterizes the growth rate of the entropy of successively finer discretizations of random variables. The R{\'e}nyi dimension of a random vector $W \in \mathbb{R}^{d_W}$ is defined as (when the limit exists)
\begin{equation}
    d^{R}(W) = \lim_{m\to \infty} \frac{H(\langle W \rangle _m)}{\log m},
\label{eq:renyi}
\end{equation}
where $\langle W \rangle _m$ is the element-wise discretization of $W$ defined as $\langle W^{(i)} \rangle _m = \frac{\floor{mW^{(i)}}}{m}$ and $H(V)$ is the entropy of $V$. Wu and Verd{\'u} \cite{wu2010renyi} interpreted the R{\'e}nyi dimension as the fundamental limit of almost lossless data compression for analog sources under regularity constraints that include linearity of the compressor and Lipschitz continuity of the decompressor.

\section{Notation and Definitions} \label{sec:def}
{\bf Notation.} We use boldface capital letters to represent matrices and use capital letters to represent (vectors of) random variables, while $P$ is specifically used to denote a probability distribution. For a random (vector) variable $X$, we use $d_X$ to denote the number of its dimensions. For a matrix $\mM$, we use $r(\mM)$ to denote the number of its rows. {We use $X=[X_1,\ldots,X_n]$ when $X_i$'s are column vectors to refer to $X=[X_1^\top\ \dots\ X_n^\top]^{\top}$.} 
% We use $f_X$ to denote the probability mass function (pmf) if the random variable is discrete or the distribution function (pdf) if it is continuous. 
We use $[X_1 \indep ... \indep X_n | W]$ to abbreviate that $X_1,\cdots,X_n$ are conditionally independent given $W$. The entropy of a random variable $V$ is denoted as $H(V)$. For a discrete random variable  $H(V)$ is defined as $\sum_{v\in \text{Supp}(V)}-\mathbb{P}(V=v)\log_2 \mathbb{P}(V=v)$, where $\text{Supp}(V)$ is the support of $V$. Since the quantities we are interested in are independent of the choice of mean values, we assume without loss of generality that all variables have zero mean. For a pair of zero-mean random variables ($X, Y$), we use $\mSigma =\left[\begin{matrix}
\mSigma_X & \mSigma_{XY}^\top \\ \mSigma_{XY} & \mSigma_Y
\end{matrix}\right]$ to denote their covariance matrix, where $\mSigma_X =\mathbb{E}(XX^\top)$ and $\mSigma_Y=\mathbb{E}(YY^\top)$ are the marginal covariance matrices, and $\mSigma_{XY} = \mathbb{E}(XY^\top)$ is the cross-covariance matrix. We say that $X$, $Y$ are {\em jointly singular} if \begin{equation}
\text{rank}(\mSigma) < \text{rank}(\mSigma_{X})+\text{rank}(\mSigma_{Y}). \label{eq:singular}
\end{equation}
Our proofs show that we can assume without loss of generality that the marginal covariances are non-singular.

\begin{figure}[t!]
  \centering
  \vspace{-.1in}
  \includegraphics[width=0.7\textwidth]{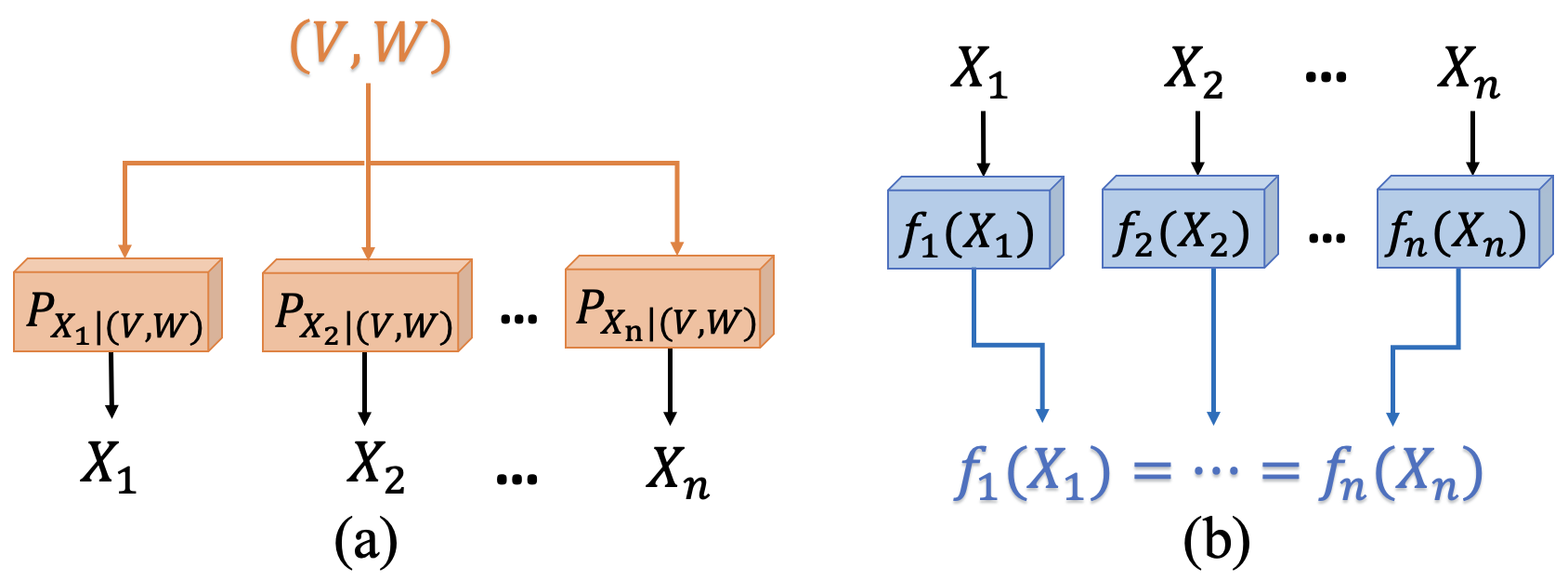}
  \vspace{-.1in}
  \caption{The one-shot exact version of (a) Wyner's distributed simulation problem and
  (b) G{\'a}cs-K{\"o}rner's distributed randomness extraction problem.}
  \vspace{-.15in}
  \label{fig:system}
\end{figure}

\textbf{Common Information Dimension (CID).} 
We consider the one-shot exact version of the distributed simulation problem  in Fig. \ref{fig:system} (a), where $n$ distributed nodes leverage the common randomness $(V,W)$,  in addition to their own local randomness, to simulate random vectors $X_1,\cdots,X_n$ that follow a given joint distribution $\pi_{X_1,\cdots,X_n}$.
We note that the distributed simulation is possible only if $X_1,\cdots,X_n$ are conditionally independent given $(V,W)$. Specifically, each node $i$ generates $X_i$ according to a distribution $P_{X_i|(V,W)}$, and the joint distribution is required to satisfy $\pi_{X_1,\cdots,X_n} = \mathbb{E}_{V,W}\left[\prod_{i=1}^n P_{X_i|(V,W)}\right]$.

We assume that $V$ is a (one-dimensional) random variable with finite entropy $H(V)<\infty$, and $W \in \mathbb{R}^{d_W}$ is a possibly continuous random vector of dimension $d_W$ that can be expressed as $W = g(X_1, \ldots, X_n)$ for some function $g$ in a given class of functions $\mathcal{F}$. 
Our goal is to determine the minimum dimension of $W$ that is necessary to enable the distributed simulation. 
Note that we allow for the finite entropy random variable $V$  to not follow the $\mathcal{F}$ restriction, and thus  the common randomness needs to be expressed using 
a function in $\mathcal{F}$ only up to finite randomness.  This allows the common information dimension to be zero when a finite amount of randomness is sufficient for the simulation, and avoids extra dimensions that may arise when this sufficient finite randomness  cannot be expressed using $\mathcal{F}$.

% The question of interest is determining the minimum dimension of $W$ that is necessary to enable the distributed simulation. Without imposing any structural assumptions on $W$, the minimum dimension does not exceed one. This is due to the existence of measurable bijections between $\mathbb{R}$ and $\mathbb{R}^d$ for any $d\geq 1$. However, these functions are not implementable and unstable, as noted in \cite{wu2010renyi}, and hence, are not useful for applications. Instead, we allow for a restriction that $W=g(X_1,\cdots,X_n)$ for some $g\in \mathcal{F}$, where $\mathcal{F}$ is a given class of functions. For example, $\mathcal{F}$ may be chosen to be the set of linear, or smooth functions. If $\mathcal{F}$ is chosen to contain all possible functions, then the minimum dimension of $W$ will be upper bounded by one, as previously observed.

\begin{definition} \label{def:CID}
    The \textbf{Common Information Dimension (CID)} of random variables $X_1,\cdots,X_n$ with respect to a class of functions $\mathcal{F}$, is defined as
\begin{equation}\label{eq:min-dim}
    d_{\mathcal{F}}(X_1,\cdots,X_n) = \min \{d_W | W\in \mathcal{W}_\mathcal{F}\},
\end{equation}
where $$\mathcal{W}_\mathcal{F}=\{W|~\exists V, g: \mathbb{R}^{\sum_{i}^nd_{X_i}} \to \mathbb{R}^{d_W} \in \mathcal{F},  \mbox{ such that}$$   $$X_1 \indep ... \indep X_n | (V,W),\; H(V)<\infty, \; W=g(X_1,\cdots,X_n)\}.$$
\end{definition}

%\noindent\textbf{R{\'e}nyi Common Information Dimension (RCID).} 
We next define the concept of the \textbf{R{\'e}nyi Common Information Dimension (RCID)} by replacing the dimension of $W$ with the Renyi dimension described in \eqref{eq:renyi}.

\begin{definition}  \label{def:RCID}
The \textbf{R{\'e}nyi Common Information Dimension (RCID)} of random variables $X_1,\cdots,X_n$ with respect to a class of functions $\mathcal{F}$ is defined as
\begin{equation}
    d_{\mathcal{F}}^{R}(X_1,\cdots,X_n) = \inf \{d^R(W) | W\in \mathcal{W}_{\mathcal{F}}\}.
\end{equation}
\end{definition}
\iffalse
\begin{figure}[t!]
  \centering
  \includegraphics[width=0.55\textwidth]{Figures/SystemGK.png}
  \caption{The one-shot exact version of G{\'a}cs-K{\"o}rner's distributed randomness extraction problem.}
  \vspace{-.2in}
  \label{fig:systemGK}
\end{figure}
\fi

 Finally, we define the {\bf G{\'a}cs-K{\"o}rner's Common Information Dimension (GKCID)} (illustrated in Fig. \ref{fig:system} (b)) for continuous random variables by replacing the entropy with the R{\'e}nyi dimension in the G{\'a}cs-K{\"o}rner's common information definition. This measures the maximum dimension of a vector $W$ that can be extracted from each random variable individually, using a potentially different function $f_i\in\mathcal{F}$.
 \begin{definition} \label{def:GKCID}
 The \textbf{G{\'a}cs-K{\"o}rner's Common Information Dimension (GKCID)} of random variables $X_1,\cdots,X_n$ with respect to a class of functions $\mathcal{F}$ is defined as
 \begin{equation}
     d_{\mathcal{F}}^{GK}(X_1,\cdots,X_n) = \sup_{(\forall i\in [n])W=f_i(X_i), f_i\in \mathcal{F}} d^R(W),
 \end{equation}
 where the optimization is over $W,\{f_i\}_{i=1}^n$.
 \end{definition}

\section{CID, RCID and GKCID for Jointly Gaussian Random Variables} \label{sec:calc_CID}

{We note that although CID, RCID and GKCID are well-defined, it is not clear whether and how they can be computed.
In  this section,  we characterize the CID, RCID, and GKCID for an arbitrary number of jointly Gaussian random variables when $\mathcal{F}$ is the class of linear functions. Our results show that CID and RCID are equal and that CID, and RCID can be computed simply from ranks of covariance matrices, while GKCID is not larger than CID. 
%The complete proofs are in the Appendix.
%\textcolor{cyan}{CF: in a journal paper, we need to provide the complete proofs, not refer to ArXiv - I thought you had the complete proof here?}

% Our main goal in this section is to prove the following result that gives the CID for jointly Gaussian random variables in a closed form.

% We start by proving the result for 2 random variables $X,Y$ which will provide properties that can be used to prove the general result. 

\subsection{Main Results} 
We consider a jointly Gaussian random vector $X=[X_1,\cdots, X_n]$, where $X_i \in \mathbb{R}^{d_{X_i}}$, with covariance matrix $\mSigma_{X}$. We use $\mSigma_{I|J}$ for $I,J\subseteq \{1,\cdots,n\}$ to denote the conditional covariance matrix of $X_{I}$ conditioned on $X_{J}$, where $X_I$ denotes the elements of $X$ with indices in the set $I$. We also use $-I$ to denote the complement of set $I$ in $\{1,\cdots,n\}$. To simplify notation, we drop the parentheses when listing the elements of the sets $I,J$. Also, as we only consider $\mathcal{F}$ being the class of linear functions, we omit it in the subscripts. 

% Let $\mSigma_{X_i}$ and $\mSigma_Y$ denote the covariance matrices of $X, Y$ respectively. 
%We assume without loss of generality that the variables $X_i$ are zero-mean. It is known from \cite{xu2015lossy,li2017distributed} that $\mSigma_{X}$ is non-singular if and only if a finite amount of randomness is sufficient to break the dependency between $X_1,\cdots,X_n$. Hence, when $\mSigma_{X}$ is singular, infinite entropy is necessary.

\subsubsection{CID of Jointly Gaussian Random Variables}
Theorem~\ref{thm:main-gauss} derives the CID for two jointly Gaussian random vectors $X,Y$ in a closed form, with $\mathcal{F}$  the class of linear functions.
\begin{theorem}\label{thm:main-gauss}
    Let $[X,Y]$ be a jointly Gaussian random vector. Then, the common information dimension between $X, Y$  with respect to the class of linear functions equals
    $$d(X,Y) = \mbox{rank}(\mSigma_X) + \mbox{rank}(\mSigma_Y) - \mbox{rank}(\mSigma).$$
Moreover, when $\mSigma
_X,\mSigma_Y$ are non-singular\footnote{Note that this can always be achieved by a linear transformation on $X$, and a linear transformation of $Y$.}, an $W$ that satisfies the minimum in \eqref{eq:min-dim} is given by $W=\mN_XX$, where $\mN =\left[\begin{matrix}\mN_X & -\mN_Y\end{matrix}\right]$ is any basis of the row space of the matrix $\mSigma$.
\end{theorem}

Our result enables the simple calculation of CID with only the knowledge of covariance matrices. The proof also provides methods to construct the pair $(V, W)$.  This can be achieved by examining the null space of (linear transformations of) the covariance matrix $\mSigma$. Theorem~\ref{thm:gauss-gen} extends the result to arbitrary number of Gaussian random vectors.
\begin{restatable}{theorem}{thmgaussgen}\label{thm:gauss-gen}
%\begin{theorem}\label{thm:gauss-gen}
    Let $X=[X_1,\cdots,X_n]$ be a jointly Gaussian random vector. The common information dimension between $X_1, \cdots,X_n$ with respect to the class of linear functions is
    $$d(X_1,\cdots,X_n) = \sum_{i=1}^n\mbox{rank}(\mSigma_{-i}) - (n-1)\mbox{rank}(\mSigma),$$
    where $\mSigma$ is the covariance matrix of $X$, and $\mSigma_{-i}$ is the covariance matrix of the random vector $[X_1, \cdots,X_{i-1}, X_{i+1}, \cdots, X_n]$.
    Moreover, a $W$ that satisfies the minimum in \eqref{eq:min-dim} is given by Algorithm~\ref{alg:gauss-gen}.
%\end{theorem}
\end{restatable}

\begin{algorithm}
\caption{Algorithm to find $(V,W)$ }
%with $X_1\indep \cdots \indep X_n | (V,W)$}
\label{alg:gauss-gen}
\begin{algorithmic}[1]
\For{$i=1,\cdots,n$}
    \State Find $\mA_i$, a basis of the row space of $\mSigma_{i|1:i-1}$.
\State Define $U_i=\mA_i X_i$ (remove parts from $X_i$ that can be obtained from previous $X_1,\cdots,X_{i-1}$).
\State Find $\mB_i$, a basis of the row space of $\mSigma_{i+1:n|1:i-1}$.
\State Define $Y_i=\mB_i [X_{i+1},\cdots,X_n]$ (remove parts from $[X_{i+1},\cdots,X_n]$ that can be obtained from previous $X_1,\cdots,X_{i-1}$).
\State Find $\tilde{\mN}_i = [\mN_i\ \bar{\mN}_i]$, the null space of $\mSigma_{U_iY_i}$.
\State Let $Z_i=\mN_i U_i$ (the parts of $X_{i+1},\cdots,X_n$ that can be obtained from $X_i$ but cannot be obtained from $X_1,\cdots,X_{i-1}$).
\EndFor
\item Let $W=[Z_1,\cdots,Z_n]$.
\item Find $\mC_i$: basis for the covariance matrix of $X_i$ conditioned on $W$ (the parts of $X_i$ that cannot be obtained from $W$). Let $T_i=\mC_iX_i$.
\item Use \cite{li2017distributed} to get $V$ that breaks the dependency of $T_1,\cdots,T_n$ conditioned on $W$.
\end{algorithmic}
\end{algorithm}

\subsubsection{RCID of Jointly Gaussian Random Variables}
%\noindent \emph{RCID of Jointly Gaussian Random Variables.}
Our next result shows that for jointly Gaussian random variables RCID and CID are the same. The proof is provided at the end of Appendix~\ref{app:thm-gauss-gen}.
\begin{restatable}{lemma}{lemRCID}\label{lem:RCID}
%\begin{lemma}\label{lem:RCID}
    Let $[X_1,\cdots ,X_n]$ be a jointly Gaussian random vector. Then, the R{\'e}nyi common information dimension between $X_1,\cdots, X_n$ with respect to the class of  $\mathcal{F}$ of linear functions is given by
    $$d^R(X_1,\cdots,X_n) =d(X_1,\cdots,X_n).$$
%\end{lemma}
\end{restatable}
\iffalse
\begin{proof}
    Let $W$ be such that $W\in \sW$ and $d_W=d(X_1,\cdots,X_n)$. As $W\in \sW$, we have that $W$ is a linear function of $[X_1,\cdots,X_n]$, hence, it is jointly Gaussian. Then, we have that $d^R(W)\leq d_W$ \cite{renyi1959dimension}. This shows that
    \begin{equation}\label{eq:ren-1}
        d^R(X_1,\cdots,X_n)\leq d(X_1,\cdots,X_n).
    \end{equation}
    To show that $d^R(X_1,\cdots,X_n)\geq d(X_1,\cdots,X_n)$, we first observe that for every $W\in \mathcal{W}$, we have that $d^R(W)\in \mathbb{N}$, hence, the infimum in the definition of $d^R(X_1,\cdots,X_n)$ can be replaced by a minimum. Then there exists $W$ such that $W\in \sW$ and $d^R(W)=d^R(X_1,\cdots,X_n)$. As $W\in \sW$, we have that $W$ is jointly Gaussian. By Lemma~\ref{lem:1}, we have that there is $I\subseteq \{1,...,d_W\}$ such that $W_I$ has non-degenerate covariance matrix and $W=\mB W_I$ for some matrix $\mB$. By construction of $W_I$ we have that $d^R(W_I)\leq d^R(W)$. Hence, without loss of generality, we assume that $\mSigma_W$ is non-degenerate, otherwise, we can replace $W$ by $W_I$. We also note that since $\mSigma_W^{1/2}W$ has independent entries, we have that $d^R(W)=d^R(\mSigma_W^{1/2}W)=d_W$ \cite{renyi1959dimension}. This shows that
    \begin{equation}\label{eq:ren-2}
        d(X_1,\cdots,X_n)\leq d^R(X_1,\cdots,X_n).
    \end{equation}
    Combining \eqref{eq:ren-1} and \eqref{eq:ren-2} concludes the proof.
\end{proof}
\fi

\subsubsection{GKCID of Jointly Gaussian Random Variables}\label{sec:gk}
%In this section, we provide a closed form solution for GKCID of jointly Gaussian random variables and a class of linear functions, and show that GKCID is not larger than CID. The result is stated in Theorem~\ref{thm:gk}. The proof of this theorem also gives a method to construct $W$, given in \eqref{eq:gk-construct} in \cite{our_arxiv}, with the maximum information dimension. The proof is provided in the Appendix~\ref{app:thm-gk} in \cite{our_arxiv}.
Theorem~\ref{thm:gk} states the closed-form solution for GKCID of jointly Gaussian random variables. The proof of this theorem also gives an efficient method to construct $W$, given in \eqref{eq:gk-construct} in Appendix~\ref{app:thm-gk}, with the maximum information dimension.
\begin{restatable}{theorem}{thmgk}\label{thm:gk}
%\begin{theorem}\label{thm:gk}
    Let $X=[X_1,\cdots,X_n]$ be a jointly Gaussian random vector. The GKCID between $X_1,\cdots,X_n$ with respect to the class of linear functions is given by
    \begin{equation}
        d^{GK}(X_1,\cdots,X_n) = r(\tilde{\mSigma})-\mbox{rank}(\tilde{\mSigma}),
    \end{equation}
    where $r(\tilde{\mSigma})$ is the number of rows of $\tilde{\mSigma}$, with
    \begin{equation}
        \tilde{\mSigma}= \begin{bmatrix}
            \mSigma_{X'_1X'_2}& 0  & \cdots & 0 \\
            0 & \mSigma_{X'_2X'_3} & \cdots & 0\\
            &  \cdots & & \\
            0 & 0 &  \cdots & \mSigma_{X'_{n-1}X'_n}\\
            \vzero_1\ \vone_2 & \vone_2\ \vzero_3 &  \cdots & \vzero_{n-1}\ \vzero_n\\
            \vzero_1\ \vzero_2 & \vzero_2\ \vone_3 &   \cdots & \vzero_{n-1}\ \vzero_n\\
            \vzero_1\ \vzero_2 & \vzero_2\ \vzero_3 & \cdots  & \vone_{n-1}\ \vzero_n
        \end{bmatrix},
    \end{equation}  
    \iffalse
        \begin{equation}
        \tilde{\mSigma}= \begin{bmatrix}
            \mSigma_{X'_1X'_2}& 0 & 0 & \cdots & 0 \\
            0 & \mSigma_{X'_2X'_3}& 0 & \cdots & 0\\
            & & \cdots & & \\
            0 & 0 & 0 & \cdots & \mSigma_{X'_{n-1}X'_n}\\
            \vzero_1\ \vone_2 & \vone_2\ \vzero_3 & \vzero_3\ \vzero_4 &  \cdots & \vzero_{n-1}\ \vzero_n\\
            \vzero_1\ \vzero_2 & \vzero_2\ \vone_3 & \vone_3\ \vzero_4 &  \cdots & \vzero_{n-1}\ \vzero_n\\
            \vzero_1\ \vzero_2 & \vzero_2\ \vzero_3 & \cdots &\vzero_{n-2}\ \vone_{n-1} & \vone_{n-1}\ \vzero_n
        \end{bmatrix},
    \end{equation}  
    \fi
    $X'_i=\mF_iX_i, \forall i\in [n]$, $\mF_i$ is a basis of the row space of $\mSigma_{X_i}$, $\vzero_i \in \mathbb{R}^{1 \times d_{X_i'}}$ and $\vone_i \in \mathbb{R}^{1 \times d_{X_i'}}$ are all zeros (and ones respectively) row vectors with the same dimension as $X_i'$.
%\end{theorem}
\end{restatable}

The following corollary follows from Theorems~\ref{thm:gauss-gen} and~\ref{thm:gk}.
\begin{corollary}
    For two jointly Gaussian random variables $X_1,X_2$ we have that
   % \begin{equation}
     $   d(X_1,X_2) = d^{GK}(X_1,X_2).$
    %\end{equation}
\end{corollary}
This result does not extend to more than two variables: as the following example shows, GKCID can be strictly less than CID. We consider three random vectors $X_1,X_2,X_3$ with non-zero variance, $X_1=X_2$ a.s., and $X_3,[X_1,X_2]$ independent. A $W$ of dimension  equal to $d_{X_1}$ is required to break the dependency, hence, $d(X_1,X_2,X_3)=d_{X_1}$. However, as $X_3$ is independent of $X_1$, all functions with $f_1(X_1)=f_3(X_3)$ have zero entropy \cite{gacs1973common,witsenhausen1975sequences}, and zero information dimension.
Corollary~\ref{cor3} follows from the proof of Theorem~\ref{thm:gk} and Lemma~\ref{lem:main-gauss}.
\begin{corollary}\label{cor3}
    Let $X=[X_1,\cdots,X_n]$ be jointly Gaussian vectors, then
    %\begin{equation}
     $   d^{GK}(X_1,\cdots,X_n)\leq d(X_1,\cdots,X_n).$
    %\end{equation}
\end{corollary}

\subsection{Proofs of Theorem~\ref{thm:main-gauss} and \ref{thm:gauss-gen}}

\begin{table}[t!]
  \centering
  \caption{Table of notation for Theorem~\ref{thm:main-gauss}}
  \label{tab:maim-gauss}
  \begin{IEEEeqnarraybox}[\IEEEeqnarraystrutmode%
    % \IEEEeqnarraystrutsizeadd{2pt}{1pt}% uncomment for more spacing
    ]{l'l}
    \toprule
    \textbf{Notation} & \textbf{Definition} \\
    \midrule
    X, Y & \text{jointly Gaussian random vectors} \\  
    \midrule
    \mSigma,\mSigma_{X}, \mSigma_{Y} & \text{covariance matrices of $[X,Y],X,Y$ respectively} \\
    \midrule
    \mN &  \text{basis of the null space of } \mSigma\\
    \midrule
    \mN_X, \mN_Y & \mN = \begin{bmatrix}\mN_X & -\mN_Y\end{bmatrix}   \hfill \hfill\refstepcounter{equation}\textup{(\theequation)} \label{eq:NX1}\\
    \midrule
    \mN_X', \mN_Y' & \text{basis of the complementary space of } \mN_X, \mN_Y \refstepcounter{equation}\textup{(\theequation)} \label{eq:NX2}\\
    \midrule
    \mM_X, \mM_Y & \mM_X = \begin{bmatrix} \mN_X \\ \mN_X' \end{bmatrix}, \mM_Y = \begin{bmatrix} \mN_Y \\ \mN_Y' \end{bmatrix}
    \hfill\refstepcounter{equation}\textup{(\theequation)} \label{eq:M} \\
    \bottomrule
  \end{IEEEeqnarraybox}
  \vspace{-.25in}
\end{table}

\noindent$\bullet$ \begin{proof}[ of Theorem~\ref{thm:main-gauss}]
We start by stating the following lemma that enables to discover deterministic relations between $X,Y$ by just examining the joint covariance matrix $\mSigma$. The proofs of all lemmas are in Appendix~\ref{app:thm-gauss-2}. 
\begin{restatable}{lemma}{lemONE}\label{lem:1}
%\begin{lemma} \label{lem:1}
Let $X = [X_1, X_2,...,X_n]$ be a {$d_X$}-dimensional random vector, $d_X=\sum_{i=1}^nd_{X_i}$, with zero mean and covariance matrix $\mSigma$. For any vectors $\va, \vb\in \mathbb{R}^{d_X}$, we have that 
$$
\va^\top X =  \vb^\top X \text{ almost surely, if and only if  }\va^\top \mSigma =  \vb^\top \mSigma.
$$
%\end{lemma}
\end{restatable}
\iffalse
\begin{proof}
    If $\va^\top X =  \vb^\top X$ almost surely, the multiplying both sides by $X^\top$ and taking expectation gives $\va^\top \mSigma =  \vb^\top \mSigma$. It remains to show the other direction; namely, if $\va^\top \mSigma =  \vb^\top \mSigma$, then $\va^\top X =  \vb^\top X$ almost surely. We have that the second moment of $(\va^\top-\vb^\top) X$ is given by
    \begin{align}
        \mathbb{E}[((\va-\vb)^\top X)^2] &= \mathbb{E}[(\va-\vb)^\top XX^\top (\va-\vb)]\nonumber\\
        &=(\va-\vb)^\top \mSigma (\va-\vb) = 0.
    \end{align}
    It follows that $(\va-\vb)^\top X=0$ almost surely.
\end{proof}
\fi

\begin{restatable}{corollary}{corone}\label{cor:1}
%\begin{corollary}\label{cor:1}
There is a subset $I \subseteq \{1,...,d_X\}$ such that $|I|=\mbox{rank}(\mSigma_X)$, and $X_I \indep Y | (V,W)$ if and only if $ X \indep Y | (V,W)$.
%\end{corollary}
\end{restatable}
\iffalse
\begin{proof}
    By Lemma \ref{lem:1}, there is a subset $I \subseteq \{1,...,d_X\}$ such that $|I|=\mbox{rank}(\mSigma_X)$, and $X=\mB X_I$ for some $\mB \in \mathbb{R}^{d_x\times |I|}$. Then we have that $X_I \indep Y | (V,W)$ if and only if $ X \indep Y | (V,W)$.
\end{proof}
\fi

Therefore, without loss of generality, we assume that $X$ and $Y$ are non-singular, which implies that $d_X = \mbox{rank}(\mSigma_X)$, and $d_Y = \mbox{rank}(\mSigma_Y)$.

Let $\mN \in \mathbb{R}^{r(\mN) \times (d_X+d_Y)}$ be a basis of the left null space of $\mSigma$. We next show some properties for the matrix $\mN$ that will help us to prove the theorem. By definition of $\mN$, we have the following facts
\begin{align}
    \label{eq:N1}
    & \textbf{Fact 1.} ~ \mN \mSigma = \vzero. \\ 
    \label{eq:N2}
    & 
    \begin{aligned}
    \textbf{Fact 2.} ~ & \mbox{rank}(\mN) = d_X + d_Y - \mbox{rank}(\mSigma)
    = \mbox{rank}(\mSigma_X) + \mbox{rank}(\mSigma_Y) - \mbox{rank}(\mSigma).
    \end{aligned}
\end{align}

Using Lemma \ref{lem:1} and \eqref{eq:N1}, we have that $\mN \begin{bmatrix} X^\top & Y^\top\end{bmatrix}^\top = \vzero$ almost surely. We partition $\mN$ as $\begin{bmatrix}\mN_X & -\mN_Y\end{bmatrix}$, where $\mN_X \in \mathbb{R}^{r(\mN)\times d_X}$ and $\mN_Y \in \mathbb{R}^{r(\mN)\times d_Y}$. Then we have that
\begin{equation}\label{eq:NX}
    \mN_X X = \mN_Y Y.
\end{equation}
$\mN$ is full-row-rank by definition; in the following, we show that $\mN_X$ (and similarly $\mN_Y$) are also full-row-rank.
\begin{restatable}{lemma}{lemrankNx}\label{lem:rankNx} 
%\begin{lemma}\label{lem:rankNx} 
Let $[X,Y]$ be a random vector with covariance matrix $\mSigma$ and $\mN=\begin{bmatrix}
    \mN_X & -\mN_Y\end{bmatrix}$ be a basis for the null space of $\mSigma$, where $\mN\in \mathbb{R}^{r(\mN)\times (d_X+d_Y)}$, $\mN_X\in \mathbb{R}^{r(\mN)\times d_X}$, and $ \mN_Y\in \mathbb{R}^{r(\mN)\times d_Y}$. If $X$ and $Y$ are non-singular (i.e., $\mSigma_X,\mSigma_Y$ are full-rank),  we have that 
    \begin{equation}\label{eq:rankNx}
    \mbox{rank}(\mN_X) = \mbox{rank}(\mN_Y) = \mbox{rank}(\mN).
    \end{equation}
%\end{lemma}
\end{restatable}
\iffalse
\begin{proof}
Suppose towards a contradiction that $\mN_X$ is not full-row-rank. Hence, there exists $\vb \neq \vzero$ such that $\vb^\top \mN_X = \vzero$. As $\mN$ is full-row-rank, 
we have that $\vb^\top \mN_Y \neq \vzero$. From equation \eqref{eq:NX} we have that $\vb^\top \mN_Y Y = \vb^\top \mN_X X$. Hence, we have that $\vb^\top \mN_Y Y = \vzero$, and $\vb^\top \mN_Y \neq \vzero$. As a result, by Lemma~\ref{lem:1}, we have $\vb^\top \mN_Y \mSigma_Y = \vzero, \vb^\top \mN_Y \neq \vzero$, which contradicts our assumption that $\mSigma_Y$ is full-rank. Therefore, it has to hold that $\mbox{rank}(\mN_X) = \mbox{rank}(\mN_Y) = \mbox{rank}(\mN)$.
\end{proof}
\fi
Next, we define two square non-singular matrices $\mM_X$ and $\mM_Y$ as $
    \mM_X = \begin{bmatrix} \mN_X^\top & \mN_X'^\top \end{bmatrix}^\top \in \mathbb{R}^{d_X \times d_X}$ and $\mM_Y = \begin{bmatrix} \mN_Y^\top & \mN_Y'^\top \end{bmatrix}^\top \in \mathbb{R}^{d_Y \times d_Y}$,
where $\mN_X', \mN_Y'$ are a basis for the complementary space of $\mN_X$ and $\mN_Y$, respectively. {Lemma~\ref{lem:M1}, and Lemma~\ref{lem:cond} show two properties of the quantities $\mM_X,\mM_Y,\mN_X', \mN_Y'$ which we prove in the Appendix~\ref{app:thm-gauss-2}.} % in \cite{our_arxiv}. 

\begin{restatable}{lemma}{lemMone}\label{lem:M1}
%\begin{lemma}\label{lem:M1}
Let $\mM_X\in \mathbb{R}^{d_X\times d_X}, \mM_Y\in \mathbb{R}^{d_Y\times d_Y}$ be full-rank matrices, and $X,Y,V,W$ be random vectors of dimension $d_X$, $d_Y$, $d_V$, and $d_W$ respectively.
We have that
    $$
    \mM_X X \indep \mM_Y Y | (V,W) \text{ if and only if } X \indep Y |(V,W).
    $$
%\end{lemma}
\end{restatable}
\iffalse
\begin{restatable}{lemma}{lemMtwo}\label{lem:M2}
%\begin{lemma}\label{lem:M2}
Define $\mN'_X$, $\mM_Y$  as in \eqref{eq:NX2} and \eqref{eq:M}. Then % We have that
    \begin{align*}
    \mbox{ } \nexists~\va, \vb \text{ such that } \va \neq \vzero\text{ and } \va^\top \mN_X' X = \vb^\top \mM_Y Y.
    \end{align*}
%\end{lemma}
\end{restatable}

\begin{proof}
Suppose towards a contradiction that $\exists~\va, \vb$ such that $\va \neq \vzero$ and  $\va^\top \mN_X' X = \vb^\top \mM_Y Y$. Then, we can extend $\mN$ by adding the following row $\begin{bmatrix} \va^\top\mN_X' & -\vb^\top\mM_Y\end{bmatrix}$, which is linearly independent on rows of $\mN$ because $\mN_X'$ is the complementary row space of $\mN_X$ and $\mN = \begin{bmatrix}\mN_X & -\mN_Y\end{bmatrix}$. This contradicts the fact that $\mN$ is a basis for the null space of $\mSigma_{X,Y}$.
\end{proof}
\fi
\begin{restatable}{lemma}{lemcond}\label{lem:cond}
%\begin{lemma}\label{lem:cond}
    Let $\mN_X, \mN_Y, \mN_X', \mN_Y'$ be as defined in \eqref{eq:NX1} and \eqref{eq:NX2}. Conditioned on $\mN_X X$, we have that $\begin{bmatrix}(\mN_X'X)^\top & (\mN_Y'Y)^\top\end{bmatrix}^\top$ has full-rank covariance matrix.
%\end{lemma}
\end{restatable}
\iffalse
\begin{proof}
Suppose towards a contradiction that there is $\va\neq \vzero$ such that $\va^\top \mSigma'_{|\mN_X X}=\vzero$, where $\mSigma'_{|\mN_X X}$ is the covariance matrix of $\begin{bmatrix}\mN_X'X \\ \mN_Y'Y\end{bmatrix}$ conditioned on $\mN_X X$. Since the mean of $\begin{bmatrix}\mN_X'X \\ \mN_Y'Y\end{bmatrix}$ conditioned on $\mN_X X$ is a linear function of $\mN_X X$ (say the conditional mean is $\mB \mN_X X$),  then, by Lemma~\ref{lem:1} we have that $$\va^\top (\begin{bmatrix}\mN_X' X \\ \mN_Y'Y\end{bmatrix} - \mB \mN_X X) = \vzero \text{ almost surely}.$$
Let us partition $\va$ as $\va = \begin{bmatrix} \va_X \\ \va_Y \end{bmatrix}$. Using the fact that $\mN_X X = \mN_Y Y$ we get that
$$\va_X^\top\mN_X'X = \vb ^\top \begin{bmatrix}\mN_Y Y \\ \mN_Y'Y\end{bmatrix} = \vb^\top \mM_Y Y,$$
where $\vb = \begin{bmatrix} \mB^\top \va \\ - \va_Y \end{bmatrix}$, which contradicts Lemma~\ref{lem:M2}.
\end{proof}
\fi

We are now ready to prove Theorem~\ref{thm:main-gauss}. We first show that the common information dimension is upper bounded as $d(X,Y) \leq \mbox{rank}(\mSigma_X) + \mbox{rank}(\mSigma_Y) - \mbox{rank}(\mSigma)$. 
Consider $\mN_XX$ as a possible $W$. Conditioned on $\mN_XX$, both $\mN_XX$ and $\mN_YY$ are deterministic, since $\mN_XX = \mN_YY$ from \eqref{eq:NX}. Hence, conditioned on $\mN_XX$, breaking the dependency between $\mM_XX$ and $\mM_YY$ reduces to breaking the dependency between $\mN_X' X$ and $\mN_Y' Y$ conditioned on $\mN_XX$. 

From Lemma~\ref{lem:cond}, we have that conditioned on $\mN_X X$, $\begin{bmatrix}(\mN_X'X)^\top & (\mN_Y'Y)^\top\end{bmatrix}^\top$ is jointly Gaussian with full-rank covariance matrix. Hence, by the result in \cite{li2017distributed}, there is a $V_W$ with $H(V_W)<\infty$ such that $\mN'_XX \indep \mN'_YY | (W, V_W)$, where $W=\mN_XX$. Since the covariance matrix of $\begin{bmatrix}(\mN_X'X)^\top & (\mN_Y'Y)^\top\end{bmatrix}^\top$ conditioned on $\mN_XX$ does not depend on the value of $\mN_XX$ and is only a function of the covariance matrix of $\mN_XX$,  $V_W$ can be  the same for all $W$, and we can refer to $V_W$ as $V$. This shows that $\mM_XX \indep \mM_YY | (\mN_XX, V)$. By Lemma~\ref{lem:M1}, $X \indep Y | (\mN_XX, V)$. Thus,
\begin{equation}
\begin{aligned}
    d(X,Y) & \leq d_{\mN_XX} \stackrel{(i)}{=} ~ \mbox{rank}(\mN_X) \stackrel{\eqref{eq:rankNx}}{=} ~ \mbox{rank}(\mN) \stackrel{\eqref{eq:N1}}{=} ~ \mbox{rank}(\mSigma_X)+\mbox{rank}(\mSigma_Y)-\mbox{rank}(\mSigma),
    \label{eq:result1}
\end{aligned}
\end{equation}
where $(i)$ follows since $\mN_X$ is full rank by Lemma~\ref{lem:rankNx}.

Next, we prove in Lemma~\ref{lem:main-gauss} the other direction, that $d(X,Y) \geq \mbox{rank}(\mN)$.  At the heart of the lemma, we prove that if there is a common function that can be extracted from both $X,Y$, namely, $f_X(X)=f_Y(Y)$ a.s. for some $f_X$, $f_Y$, then for $(V,W)$ to break the $X,Y$ dependency, $f_X(X)$ (and hence $f_Y(Y)$) is a deterministic function of $(V,W)$. We also show that if $f_X,f_Y$ are linear and $W=\mA\begin{bmatrix} X^\top & Y^\top\end{bmatrix}^\top$ for some $\mA$, then $d_W\geq d_{f_X(X)}$.
\begin{lemma}\label{lem:main-gauss}
    Let $[X,Y]$ be a jointly Gaussian random vector and  $V,W$ be  random variables such that $W=\mA\begin{bmatrix} X^\top & Y^\top\end{bmatrix}^\top$ for some matrix $\mA$, $H(V)\leq \infty$ and $X\indep Y| (V,W)$. Let matrix $\mN_X$ be such that $\mN_XX$ is non-singular. If there exists matrix $\mN_Y$ such that $\mN_XX=\mN_YY$ a.s., then $\mN_XX=\mA' W$ a.s. for some matrix $\mA'$ and $d_W\geq d_{\mN_XX}$.
\end{lemma}

%{\em Proof of Lemma~\ref{lem:main-gauss}: }
%\noindent{\em Proof.}
\begin{proof}[ of Lemma~\ref{lem:main-gauss}]
% As functions of independent random variables are independent, we have that $\mN_XX \indep \mN_YY | (V,W)$. 
First, we show that $\mN_XX$ is a deterministic function of $(V,W)$. Suppose towards a contradiction that there is a set $\sS\subseteq \mathbb{R}^{r(\mN)}$ such that $0<\mathbb{P}[\mN_XX\in\sS |(V,W)]<1$. Since $\mN_XX=\mN_YY$ a.s.,  we have that
\begin{equation}
    \mathbb{P}[\mN_XX\in \sS, \mN_YY\in \sS^C| (V,W)] = 0,
\end{equation}
where $\sS^C$ is the complement of $\sS$. However, from $0<\mathbb{P}[\mN_XX\in \sS|(V,W)]<1$ we get that 
\begin{equation}
\begin{aligned}
    & \mathbb{P}[\mN_XX\in \sS| (V,W)] \mathbb{P}[\mN_YY\in \sS^C| (V,W)] \\
    = \ & \mathbb{P}[\mN_XX\in \sS| (V,W)] \mathbb{P}[\mN_XX\in \sS^C| (V,W)] \neq 0.
\end{aligned}
\end{equation}
This implies that
\begin{equation}\label{eq:indep-contr}
\begin{aligned}
    & \mathbb{P}[\mN_XX\in \sS, \mN_YY\in \sS^C| (V,W)] \\
    \neq \ & \mathbb{P}[\mN_XX\in \sS| (V,W)] \mathbb{P}[\mN_YY\in \sS^C| (V,W)].
\end{aligned}
\end{equation}
However, as functions of independent random variables are independent, we have that $\mN_XX$ and $\mN_YY$ are conditionally independent given $(V,W)$; being projections of $\mM_XX, \mM_YY$. This contradicts \eqref{eq:indep-contr}. 

This implies that for any $\sS\subseteq \mathbb{R}^{r(\mN)}$ we either have $\mathbb{P}[\mN_XX\in \sS| (V,W)]=1$ or $\mathbb{P}[\mN_XX\in \sS| (V,W)]=0$. We show next that this implies that $\mN_XX$ is a deterministic function of $(V,W)$.

As the interval from $(-\infty, \infty)$ can be partitioned into countably many sets of the form $(0+m,1+m]$, by countable additivity of measures we get that there is a cube of the form $\sS=\prod_{i=1}^{r(\mN)}(0+m_i,1+m_i]$ that has $\mathbb{P}[\mN_XX\in \sS| (V,W)]=1$. If we repeatedly halve one of the largest dimensions of the cube we get a sequence of hyper-rectangles $...\subseteq \sR_2 \subseteq \sR_1$ such that $\mathbb{P}[\mN_XX\in \sR_i| (V,W)]=1, \forall i =1,2,...$ and $\cap_{i\in \mathbb{N}}\sR_i$ contains exactly one member. The last fact is proved in the following. We notice that $\cap_{i\in \mathbb{N}}\sR_i$ contains at most one member because for any two points $x_1,x_2\in \mathbb{R}^{r(\mN)}$, there is some $i$ such that the largest dimension of $\sR_i$ is less than $\|x_1-x_2\|_2$ which implies that at most one point of $x_1,x_2$ can be in $R_i$. It is also not possible that $\cap_{i\in \mathbb{N}}\sR_i$ is empty as by the continuity from above of finite measures, we have that
\begin{equation}
    \mathbb{P}[\mN_XX \in \cap_{i\in \mathbb{N}}\sR_i | (V,W)] = 1.
\end{equation}
Therefore, $\cap_{i\in \mathbb{N}}\sR_i$ must contain a single member. Let us denote the unique point in $\cap_{i\in \mathbb{N}}\sR_i$ by $g(V,W)$, where $g$ is a deterministic function. Then, we have that $\mathbb{P}[\mN_XX = g(V,W) | (V,W)] = 1$.
Hence, we have that
\begin{align}\label{eq:H1}
    H(\mN_XX|W) &\leq H(\mN_XX, V|W)\nonumber\\
    & = H(\mN_XX|(V,W)) + H(V|W)\nonumber \\
    & = H(V|W) \leq H(V) < \infty.
\end{align}
Since $W = \mA\begin{bmatrix} X^\top & Y^\top\end{bmatrix}^\top$, we have that $(\mN_XX, W)$ follows a jointly Gaussian distribution. As a result, conditioned on $W$, we have that $\mN_XX$ is also jointly Gaussian, whose entropy is either 0 (for zero variance) or $\infty$. Based on \eqref{eq:H1}, it must be that $H(\mN_XX|W)=0$. Hence, we have that $\mN_XX = \mB W$,
for some $\mB \in \mathbb{R}^{\mbox{rank}(\mN)\times d_W}$. And as a result,
%\begin{equation}\label{eq:cov}
    $\mN_X\mSigma_X = \mB\mathbb{E}(W X^\top)$.
%\end{equation}
Then we have that
\begin{equation}
\begin{aligned}
    \mbox{rank}(\mN_X\mSigma_X) & \stackrel{(i)}{=} \mbox{rank}(\mN_X) \stackrel{\eqref{eq:rankNx}}{=} \mbox{rank}(\mN) \leq \mbox{rank}(\mB) \leq d_W,
\end{aligned}    
\label{eq:result2}
\end{equation}
where $(i)$ follows from  the fact that $\mN_XX$ has full rank covariance matrix.
\end{proof}
Combining \eqref{eq:result1} and \eqref{eq:result2}, we conclude that
\begin{equation}
\begin{aligned}
    d(X,Y) = \mbox{rank} (\mN) \nonumber = \mbox{rank}(\mSigma_X) + \mbox{rank}(\mSigma_Y) - \mbox{rank}(\mSigma). %{\color{blue} \qquad \quad ~~ \hfill{\square}}
\end{aligned}
\end{equation} 
\end{proof}

\noindent $\bullet$ \begin{proof}[ Outline of Theorem~\ref{thm:gauss-gen}]
%Let 
% \begin{align}
%     &A_i \text{ be basis of the row space of } \mSigma_{i|1:i-1},\label{eq:def-Ai} \\
%     &B_i \text{ be basis of the row space of } \mSigma_{i+1:n|1:i-1},\label{eq:def-Bi} \\
%     &U_i = A_i X_i,\\
%     &Y_i = B_i [X_{i+1}^\top\ \cdots\ X_n^\top]^\top, \\
%     &\tilde{N}_i^{r(\tilde{N}_i)\times (r(U_i)+r(Y_i))} \text{ be basis of the null space of }\nonumber \\ 
%     &\qquad \mSigma_{U_i,Y_i|X_1,\cdots,X_{i-1}},\nonumber \\
%     &\text{partition } \tilde{N} \text{ as } [N_i^{r(\tilde{N}_i)\times r_{U_i}}\ \bar{N}_i^{r(\tilde{N}_i)\times r(Y_i)}],\label{eq:def-N}\\
%     &Z_i = N_i U_i,\label{eq:def-Z}\\
%     &C_i \text{ be basis of the row space of } \mSigma_{X_i|Z_1,\cdots,Z_n}\nonumber \\
%     &T_i = C_iX_i.\label{eq:def-T}
% \end{align}
We here give a proof outline, and provide the complete proof in App.~\ref{app:thm-gauss-gen}. 
%~\cite{our_arxiv}.
The main part of the proof, illustrated in Algorithm~\ref{alg:gauss-gen}, constructs variables $Z=[Z_1,\cdots,Z_n]$ that satisfy:\\
\textbf{$(i)$}  Conditioned on $Z$, the dependency between $X_1,\cdots,X_n$ can be broken using finite randomness (see Lemma~\ref{lem:T-full} in App.~\ref{app:thm-gauss-gen}
). This is proved by showing that after eliminating from $X$ the parts that can be almost surely determined by $Z$, the remaining part is jointly non-singular Gaussian. This shows that CID is upper bounded by the total number of dimensions of $Z$.\\
\textbf{$(ii)$} For any $V,W$ that break the dependency between $X_1,\cdots,X_n$, we have that $Z$ is a linear function of $W$ (see Lemma~\ref{lem:Z-full} in App.~\ref{app:thm-gauss-gen}
). By showing that $Z$ is a jointly non-singular Gaussian vector, we prove that the dimension of $W$ is lower bounded by the dimension of $Z$, hence, CID is lower bounded by the number of dimensions of $Z$.

% The remaining part of the proof characterizes the total number of dimensions of $Z$. 
 
We build $Z$ sequentially as follows. $Z_1$ represents the information that $X_1$ contains about $[X_2,\cdots,X_n]$; namely, the linearly independent dimensions of $X_2,\cdots,X_n$ that can be determined from $X_1$. Then, $Z_2$ contains the amount of information that $X_2$ contains about $X_3,\cdots,X_n$ that $X_1$ does not contain. Generally, $Z_i$ contains the information that $X_i$ contains about $X_{i+1},\cdots,X_n$ which is not contained in any of the previous $X_1, \cdots, X_{i-1}$.
\end{proof}
\section{The relation between CID and Wyner's Common Information} \label{sec:approx}
 In this section,  we restrict our attention to the case of two Gaussian random vectors, and formulate three approximation problems,  to investigate the asymptotic behavior of Wyner's common information in the (nearly) infinity regime.  Our results indicate that the growth rates of the approximate common information in all these scenarios are determined by the common information dimension.

 \begin{boldremark}\label{rm:diff}
 We make two observations related to our formulations:\\
 1) We restrict our attention to two  Gaussian random vectors because, as mentioned before, the calculation of Wyner's common information is challenging, and the closed-form solution for continuous random vectors (of arbitrary dimension) is only known for two Gaussian sources, while a closed-form expression for multiple Gaussian sources remains open.\\
 2) We do not consider the GK-version of the approximation problem since the GK common information has an inherent discontinuity. In particular, it is easy to see that if the Gaussian sources are singular, then the GK common information is infinite (as is the case for the Wyner as well), however, if they are approximated by any non-singular Gaussian distribution, then the GK common information of the approximate distribution is zero. Hence, the GK version of the common information is not suitable for such approximations. Thus in this section, we exclusively focus on the Wyner version. For simplicity, we drop the subscripts in $C_{\text{Wyner}}(X,Y)$ and use $C(X,Y)$ instead to denote the Wyner's common information.

\end{boldremark}
\subsection{Problem Statements} \label{sec:ps}
\subsubsection{Common information of nearly singular sources}
This formulation aims to study the growth rate of the common information for a sequence of pairs of random variables that approach joint singularity. In particular, let $X, Y$ be Gaussian random variables with $\text{rank}(\mSigma) < \text{rank}(\mSigma_{X})+\text{rank}(\mSigma_{Y})$, and hence, $d(X,Y)\geq 1$ and  $C(X,Y)=\infty$.  Let $\{(X_\epsilon,Y_\epsilon)\}_{\epsilon>0}$ be a sequence of Gaussian  random variables satisfying
\begin{equation}\label{def: seq}
    \mSigma_{X_\epsilon} = \mSigma_X, \mSigma_{Y_\epsilon} = \mSigma_Y, \text{ and } \forall\; i, \;\; |\rho_i(\epsilon) - \sigma_i| = \epsilon, \
\end{equation}
where $\{\sigma_i\}$ and $\{\rho_i(\epsilon)\}$ are the singular values of $\mSigma_{X}^{-1/2}\mSigma_{X Y}\mSigma_{Y}^{-1/2}$ and $\mSigma_{X_\epsilon}^{-1/2}\mSigma_{X_\epsilon Y_\epsilon}\mSigma_{Y_\epsilon}^{-1/2}$ respectively, in a decreasing order. 
These requirements ensure that $(X_\epsilon,Y_\epsilon)$ remain non-singular (and thus have finite common information), while the joint distribution of $(X_\epsilon,Y_\epsilon)$  converges to that of $X,Y$ as $\epsilon \downarrow 0$.
% We are interested in the behavior of $C(X_\epsilon, Y_\epsilon)$ as $\epsilon \downarrow 0$, which implies that {\color{blue} each of the singular values of $\mSigma_{X_\epsilon}^{-1/2}\mSigma_{X_\epsilon Y_\epsilon}\mSigma_{Y_\epsilon}^{-1/2}$ go to the corresponding singular values of $\mSigma_{X}^{-1/2}\mSigma_{X Y}\mSigma_{Y}^{-1/2}$.} 

\begin{boldremark} 
    The conditions in \eqref{def: seq} force each singular value of $\mSigma_{X_\epsilon}^{-1/2}\mSigma_{X_\epsilon Y_\epsilon}\mSigma_{Y_\epsilon}^{-1/2}$ to go to the corresponding singular value of $\mSigma_{X}^{-1/2}\mSigma_{X Y}\mSigma_{Y}^{-1/2}$ at an identical rate $\epsilon$. This enables us to study how the common information increases as a function of  $\epsilon$. 
    There exist however other cases where the same results apply, and we next give two examples. 

    One case is when we have different convergence rates for each singular value. The second condition in \eqref{def: seq} can be replaced with $\forall i, \; |\rho_i(\epsilon) - \sigma_i| = a_i\epsilon$, where $a_i$'s are constants with different values.
    provided that these rates are of the same order, meaning that they differ only by a multiplicative constant. 

    A second case is when we only approximate the singular values that equal to $1$. Recall that from \eqref{eq:wyner}, the common information is infinite when $\sigma_i=1$ for some $i$. Consider a sequence of covariance matrices that have singular values satisfying \eqref{def: seq} when $\sigma_i=1$ and share the same singular values with the target distribution for all other indices, i.e., $\rho_i=\sigma_i$ when $\sigma_i \neq 1$. It is easy to see that the same results we establish assuming the condition in \eqref{def: seq} holds, also extend for the described sequence as well.
\end{boldremark}

\subsubsection{Approximate simulation}
This formulation looks at approximating a pair of  Gaussian random variables $X, Y$ that are jointly singular ($C(X,Y)=\infty$) with Gaussian random variables $\hat{X}, \hat{Y}$ that (i) are non-singular ($C(\hat{X}, \hat{Y})$ is finite)  and (ii) have a distribution close to the distribution of $X, Y$. In other words, we ask,  if we are restricted to using a finite number of bits as common information, how well can we (approximately) simulate $X, Y$.

We use the Frobenius-norm between covariance matrices to measure how close two Gaussian distributions are.
%$\hat{X}$ and $\hat{Y}$ to follow the same marginal distribution as $X$ and $Y$ (i.e., $\mSigma_{\hat{X}} = \mSigma_{X}$, $\mSigma_{\hat{Y}} = \mSigma_{Y}$), while their cross-covariance matrix is close to $\mSigma_{XY}$. 
For some $\epsilon > 0$, we define the {\em $\epsilon$-approximation common information} as
\begin{equation}\label{def: opt}
        C_\epsilon(X,Y) := \min_{\|\mSigma - \hat{\mSigma}\|_F\le \epsilon} C(\hat{X},\hat{Y}),
\end{equation}
where the optimization is over all pairs $(\hat{X},\hat{Y})$ with covariance matrix $\hat{\mSigma}$ and $\|\cdot\|_F$ is the Frobenius norm of a matrix. 
\begin{boldremark}
    The results on $C_\epsilon(X,Y)$ extend if we replace the Frobenius norm with any distribution distance $\text{dist}(XY,\hat{X}\hat{Y})$ that satisfies $a\|\mSigma - \hat{\mSigma}\|_F \le \text{dist}(XY,\hat{X}\hat{Y}) \le b\|\mSigma - \hat{\mSigma}\|_F$ for all Gaussian variables $XY,\hat{X}\hat{Y}$ and some constants $a$ and $b$.
\end{boldremark} 
%\noindent\textbf{Remark 2.} We note that these two formulations are related but not the sam
\begin{boldremark}\label{rm:diff}
    Note that formulation 1 in \eqref{def: seq} studies a more restricted set of sequences than the sequences included in the feasible set of the optimization problem in \eqref{def: opt}. However, the result we show for formulation 1 is stronger as it holds for {\em all sequences} that satisfy the condition in \eqref{def: seq}. In contrast, the results in formulation 2 only hold for the sequence with the minimum common information (that achieves the optimal value of the minimization problem). It can be easily shown that there exist sequences in the feasible set of formulation~2 that have different asymptotics. For example, if some singular values of the approximation matrix take the value $1$ or approach $1$ at a rate different from $\Theta(\epsilon)$(e.g., $\epsilon^2$ or $2^\epsilon$).
    %It turns out that the sequences with minimum common information in formulation 2 satisfy a similar condition as \eqref{def: seq}. 
\end{boldremark}

\subsubsection{Common information between quantized variables}
In this formulation, we study Wyner's common information between quantized continuous random vectors. Let $X$ and $Y$ be a pair of jointly singular Gaussian random vectors with $C(X,Y) = \infty$. For $m>0$, we use $\langle X \rangle_m = [\langle X_1 \rangle_m, \cdots, \langle X_{d_X} \rangle_m]^\top$ and $\langle Y \rangle_m = [\langle Y_1 \rangle_m, \cdots, \langle Y_{d_Y} \rangle_m]^\top$ to denote the quantized $X$ and $Y$, where the quantization operator on each element is defined as $\langle X_i \rangle_m = \frac{\floor{mX_i}}{m}$ and $\langle Y_i \rangle_m = \frac{\floor{mY_i}}{m}$. Note that $C(\langle X \rangle_m, \langle Y \rangle_m) < \infty$ for all $0 < m < \infty$, and as $m$ approaches $\infty$, $\langle X \rangle_m$ and $\langle Y \rangle_m$ converges to $X$ and $Y$ respectively. We are interested in the growth rate of their common information $C(\langle X \rangle_m, \langle Y \rangle_m)$.

%%%%%%%%%%%%%%%%%%%%%%%%%%%%%%%%%%%%%%%%%%%%%%%%%%%%%%%%%%%%%%%%

\subsection{Asymptotic Behavior of Wyner's Common Information}\label{sec:main}

In this section, we present our main results and proof outlines for the three formulations described in Section~\ref{sec:ps}. The detailed proofs are provided in Appendices~\ref{sec:pf-thm1} and ~\ref{sec:pf-thm2}.

Before stating our main results, we present two properties of covariance matrices and the common information dimension, which are important to the proofs of Theorems \ref{thm:seq} and \ref{thm:appr}. As stated in \eqref{eq:wyner} the common information is determined by the singular values of the normalized cross-covariance matrix $\mSigma_X^{-1/2}\mSigma_{XY}\mSigma_Y^{-1/2}$. Lemma \ref{lm:range} proves a bound on these singular values.

\begin{restatable}{lemma}{lemrange}\label{lm:range}
    Let $X \in \mathbb{R}^{d_X}$ and $Y \in \mathbb{R}^{d_Y}$ be jointly Gaussian variables with covariance matrix $\mSigma = \left[\begin{matrix} \mSigma_{X} & \mSigma_{XY}^\top \\ \mSigma_{XY} & \mSigma_{Y}\end{matrix}\right]$, and $d = \min\{d_X, d_Y\}$. Then the singular values of $\mSigma_{X}^{-1/2}\mSigma_{XY}\mSigma_{Y}^{-1/2}$, denoted as $\{\sigma_i\}_{i=1}^d$, satisfy 
    \begin{equation}
        0 \le \sigma_i \le 1, \forall i \in [d].
    \end{equation}
\end{restatable}

The following lemma shows the relationship between the common information dimension and the singular values of $\mSigma_X^{-1/2}\mSigma_{XY}\mSigma_Y^{-1/2}$, which will enable us to connect the quantities $C(X_\epsilon, Y_\epsilon)$ and $C_\epsilon(X,Y)$ with the common information dimension $d(X,Y)$.

\begin{restatable}{lemma}{lemone}\label{lm:one}
    Assume $X \in \mathbb{R}^{d_X}$, $Y \in \mathbb{R}^{d_Y}$ are jointly Gaussian variables with covariance matrix $\mSigma = \left[\begin{matrix} \mSigma_{X} & \mSigma_{XY}^\top \\ \mSigma_{XY} & \mSigma_{Y}\end{matrix}\right]$, and $\{\sigma_i\}$ are the singular values of $\mSigma_{X}^{-1/2}\mSigma_{XY}\mSigma_{Y}^{-1/2}$. Then the common information dimension between $X$ and $Y$, with respect to linear functions, satisfies
    \begin{equation}
        d(X,Y) = \sum_{i=1}^{\min\{d_X, d_Y\}} \mathbbm{1}\{\sigma_i=1\}.
    \end{equation}
\end{restatable}

\subsubsection{Common information of nearly singular sources}
We consider a sequence of pairs of Gaussian random variables $\{(X_\epsilon,Y_\epsilon)\}_{\epsilon>0}$ satisfying \eqref{def: seq}. The following result shows that the growth rate of the common information $C(X_\epsilon,Y_\epsilon)$ is determined by the common information dimension $d(X, Y)$ with respect to linear functions.
\begin{restatable}{theorem}{thmseq} \label{thm:seq}
    Let $X\in\mathbb{R}^{d_X}$ and $Y\in\mathbb{R}^{d_Y}$ be a pair of jointly singular Gaussian variables, and 
    % with $\text{rank}(\mSigma) < \text{rank}(\mSigma_{X})+\text{rank}(\mSigma_{Y})$
    $\{(X_\epsilon, Y_\epsilon)\}_{\epsilon>0}$ be a sequence as defined in \eqref{def: seq}. Then the common information $C(X_\epsilon, Y_\epsilon)$ satisfies
    \begin{equation}
        \lim_{\epsilon \downarrow 0} \frac{C(X_\epsilon, Y_\epsilon)}{\frac{1}{2}\log(\frac{1}{\epsilon})} = d(X,Y).
    \end{equation}
\end{restatable}
%{\bf Proof Outline.} 
\begin{proof}[ Outline of Theorem~\ref{thm:seq}]
The main technical challenge in proving Theorem~\ref{thm:seq} is the fact that there exist multiple sequences of random variables $X_\epsilon,Y_\epsilon$, with different values of $C(X_\epsilon,Y_\epsilon)$, that satisfy the constraints in \eqref{def: seq}. To address this issue, we prove the result by deriving an upper and a lower bound on $C(X_\epsilon,Y_\epsilon)$ that have the same asymptotic behavior.

The proof focuses on showing that $\lim_{\epsilon \downarrow 0} \frac{C(X_\epsilon, Y_\epsilon)}{\frac{1}{2}\log(\frac{1}{\epsilon})} = \sum_i^{\min\{d_X, d_Y\}} \mathbbm{1}\{\sigma_i=1\}$, where $\{\sigma_i\}$ are the singular values of the matrix $\mSigma_X^{-1/2}\mSigma_{XY}\mSigma_Y^{-1/2}$. We prove this by providing an upper and lower bound on $\frac{C(X_\epsilon, Y_\epsilon)}{\frac{1}{2}\log(\frac{1}{\epsilon})}$ that have the same limit when $\epsilon \downarrow 0$. Then we relate $\sum_i^{\min\{d_X, d_Y\}} \mathbbm{1}\{\sigma_i=1\}$ to the common information dimension $d(X,Y)$ using Lemma \ref{lm:one}. %Then, we relate the singular values $\{\sigma_i\}$ to the CID $d(X,Y)$ by showing that $d(X,Y) = \sum_i^{\min\{d_X, d_Y\}} \mathbbm{1}\{\sigma_i=1\}$. This is proved by observing that the transformation $\mSigma_X^{-1/2}X, \mSigma_Y^{-1/2}Y$ does not change the CID and makes the marginal covariance matrices to be identity. Hence, the rank formula in \eqref{CID-gauss} can be computed by examining the number of ones in the singular values of the cross-covariance matrix after the transformation, i.e., $\mSigma_X^{-1/2}\mSigma_{XY}\mSigma_Y^{-1/2}$.
\end{proof}

\subsubsection{Approximate simulation}

The following result shows that the $\epsilon$-approximation common information $C_\epsilon(X,Y)$, defined in \eqref{def: opt}, for Gaussian variables grows at a rate determined by the common information dimension $d(X,Y)$ with respect to linear functions.

\begin{restatable}{theorem}{thmappr}\label{thm:appr}
Let $X \in \mathbb{R}^{d_X}$ and $Y \in \mathbb{R}^{d_Y}$ be a pair of jointly Gaussian random variables, then
\begin{equation}
    \begin{aligned}
        \lim_{\epsilon \downarrow 0} \frac{C_\epsilon(X,Y)}{\frac{1}{2}\log(\frac{1}{\epsilon})} = d(X,Y).
    \end{aligned}
\end{equation}
\end{restatable}

%{\bf Proof Outline.}
\begin{proof}[ Outline of Theorem~\ref{thm:appr}]
The main technical challenge in proving Theorem~\ref{thm:appr} is the difficulty in finding a closed form solution of the optimization problem defining $C_\epsilon(X,Y)$. To address this issue, we follow a similar approach as in Theorem~\ref{thm:seq} by deriving an upper and a lower bound on $C_\epsilon$ that have the same asymptotics. However, it turns out that finding upper and lower bounds that have the same asymptotics is more involved than in the case of Theorem~\ref{thm:seq}.

The proof uses a pair of upper and lower bounds, derived as described next, to show that $\lim_{\epsilon \downarrow 0} \frac{C(X_\epsilon, Y_\epsilon)}{\frac{1}{2}\log(\frac{1}{\epsilon})} = \sum_i^{\min\{d_X, d_Y\}} \mathbbm{1}\{\sigma_i=1\}$, where $\{\sigma_i\}$ are the singular values of the matrix $\mSigma_X^{-1/2}\mSigma_{XY}\mSigma_Y^{-1/2}$.
From Lemma \ref{lm:one},  this concludes the proof of Theorem~\ref{thm:appr}.

{\bf Upper Bound.} As $C_\epsilon(X,Y)$ is the optimal value of a minimization problem, any feasible solution provides an upper bound. To find a feasible solution we use $\mSigma_{\hat{X}}=\mSigma_{X}, \mSigma_{\hat{Y}}=\mSigma_{Y}$. Then, we design the singular values of $\mSigma_{\hat{X}}^{-1/2}\mSigma_{\hat{X}\hat{Y}}\mSigma_{\hat{Y}}^{-1/2}$, denoted as $\{\rho_i\}$, as follows.  We set $
\rho_i = \sigma_i$ when $\sigma_i \ne 1$. Recall that choosing a singular value to be $1$ results in an infinite value for the common information. Hence, when $\sigma_i=1$ we choose $\rho_i=1-\delta$ where $\delta$ is the largest value that does not violate the constraint $\|\mSigma-\hat{\mSigma}\|_F\leq \epsilon$.

{\bf Lower Bound.} To find a lower bound, we relax the constraints set $\|\mSigma-\hat{\mSigma}\|_F\leq \epsilon$, resulting in a smaller optimal value, to make it possible to find a closed form solution of the problem. The proof of the lower bound hinges on showing that $\|\mSigma-\hat{\mSigma}\|_F\leq \epsilon$ implies
\begin{equation}
    \|\Lambda-\hat{\Lambda}\|_F\leq c\epsilon,
\end{equation}
where {$\Lambda = \text{diag}(\sigma_i), \hat{\Lambda} = \text{diag}(\rho_i)$ }are matrices containing the singular values of $\mSigma_{X}^{-1/2}\mSigma_{XY}\mSigma_{Y}^{-1/2}$, $\mSigma_{\hat{X}}^{-1/2}\mSigma_{\hat{X}\hat{Y}}\mSigma_{\hat{Y}}^{-1/2}$ respectively, and $c$ is a constant that may depend on $\mSigma_X,\mSigma_Y$. To further simplify the problem, we remove from the objective function the terms corresponding to $\sigma_i< 1$, and also remove the value $\log(1+\rho_i)$ from each term (recall the common information in \eqref{eq:wyner}). We note that each term in the objective function is non-negative,  and hence, removing terms will not increase the optimal solution value. Furthermore, we expect the asymptotics of the common information to be influenced by the singular values corresponding to $\sigma_i = 1$. This results in the following optimization problem
\begin{equation*}
    \begin{aligned}
        \min_{\rho} \quad & \frac{1}{2} \sum_{i: \sigma_i=1} \log\frac{1}{1-\rho_i} \\ 
        \text{s.t.} \quad & \sum_{i: \sigma_i=1} (\sigma_i -\rho_i)^2 \le \epsilon^2, 0 \le \rho_i \le 1,
    \end{aligned}
\end{equation*}
which can be solved in a closed form using symmetry and concavity of the $\log$ function.
\end{proof}

\begin{boldremark}
    We note that we can efficiently construct random variables for each $\epsilon$ with common information that has the asymptotic behavior in Theorem 2 (and thus can be used to approximate the target singular distribution with (nearly) the smallest common information). A possible choice is $\mSigma_{\hat{X}}=\mSigma_X, \mSigma_{\hat{Y}}=\mSigma_Y, \mSigma_{\hat{X}}^{-1/2} \mSigma_{\hat{X}\hat{Y}} \mSigma_{\hat{Y}}^{-1/2}= U \hat{\Lambda} V$, where $U,V$ are orthonormal matrices of the singular value decomposition of $\mSigma_{X}^{-1/2} \mSigma_{XY} \mSigma_{Y}^{-1/2}$, and $\hat{\Lambda}$ is given in \eqref{eq:rho_achievable} (Appendix ~\ref{sec:pf-thm2}). 
\end{boldremark}

\begin{boldremark}
    {\bf Why do these two theorems have the same bound?} {It may seem at first surprising that even though $C(X_\epsilon,Y_\epsilon)$ and $C_\epsilon(X,Y)$ have different definitions, they both grow (nearly) as $\frac{1}{2}d(X,Y)\log(1/\epsilon)$. {Indeed, as we observed in Remark~\ref{rm:diff}} the feasible set defining $C_\epsilon(X,Y)$ in \eqref{def: opt} contains different sequences of random variables than those satisfying the conditions in \eqref{def: seq}. However, the proof of Theorem~\ref{thm:appr} shows that the random variables which minimize the common information satisfy a constraint similar to \eqref{def: seq}; namely, the singular values $\rho_i$ corresponding to $\sigma_i=1$ have the same distance to $1$, {where $\sigma_i$ and $\rho_i$ are the singular value of $\mSigma_{X}^{-1/2}\mSigma_{XY}\mSigma_{Y}^{-1/2}$ and $\mSigma_{\hat{X}}^{-1/2}\mSigma_{\hat{X}\hat{Y}}\mSigma_{\hat{Y}}^{-1/2}$ respectively}. Intuitively, to minimize the common information in \eqref{eq:wyner}, we need $\rho_i$ to be as far as possible from the value $1$, however, the distance constraint in \eqref{def: opt} restricts us from choosing $\rho_i$ too far from $1$ whenever $\sigma_i=1$. If one $\rho_i$ is very close to $1$, it will dominate the summation in \eqref{eq:wyner} resulting in large common information. Hence, a good solution to \eqref{def: opt} distributes the distance budget $\epsilon$ evenly across the $\rho_i$'s corresponding to $\sigma_i=1$.}
\end{boldremark}

\begin{boldremark}
%In practical computer systems, each real values is represented in a finite number of binary digits, which requires quantization and is of limited precision. For example, 
Consider a machine that stores each real value in $B$ bits using standard floating-point representation. Such a type of machine can store numbers up to a precision\footnote{Note that for any quantization scheme, there exists at least one input value such that the quantization error is greater or equal to $2^{-B}$.} of $2^{-B}$. Suppose that we are interested in the distributed simulation of two random variables $X,Y$ up to the maximum machine precision. Our result indicates that the minimum amount of shared randomness between the quantized $X,Y$ is proportional to $\frac{1}{2} \log(\frac{1}{2^{-B}})d(X,Y) = \frac{1}{2}Bd(X,Y)$ bits. As each variable stores at most $B$ bits, $\frac{1}{2}d(X,Y)$ variables (dimensions) are required to store the shared randomness on such machines in order to perform the distributed simulation up to the required $2^{-B}$ accuracy.
\end{boldremark}

\subsubsection{Common information between quantized variables} The following result shows that the Wyner's common information between uniformly quantized Gaussian random vectors also grows in proportion to the common information dimension, as the quantization precision increases.
\begin{restatable}{theorem}{thmquant} \label{thm:quant}
    Let $X\in\mathbb{R}^{d_X}$ and $Y\in\mathbb{R}^{d_Y}$ be a pair of jointly singular Gaussian random vectors. Then the common information between the quantized $\langle X\rangle_m$ and $\langle Y \rangle_m$ satisfies
    \begin{equation}
        \lim_{m \to \infty} \frac{C(\langle X\rangle_m, \langle Y \rangle_m)}{\log m} = d(X,Y),
    \end{equation}
    where $d(X,Y)$ is the common information dimension of $X$ and $Y$ with respect to the class of linear functions.
\end{restatable}
\begin{proof}[ Outline of Theorem~\ref{thm:quant}]
Since it is hard in general to directly solve Wyner's common information even for discrete variables \cite{yu2022common}, we prove the result through matching upper and lower bounds. First, we show that $\lim_{m \to \infty} \frac{C(\langle X\rangle_m, \langle Y \rangle_m)}{\log m} \ge d(X,Y)$ using the inequality that mutual information is not larger than the Wyner's common information \cite{wyner1975common}.

Next, we prove $\lim_{m \to \infty} \frac{C(\langle X\rangle_m, \langle Y \rangle_m)}{\log m} \le d(X,Y)$ as follows. Recall (from the proofs of Lemma~\ref{lm:range} and \ref{lm:one}) that there exist invertible linear transformations $\mT_X$ and $\mT_Y$ such that $X = \mT_X X'$ and $Y = \mT_Y Y'$, where $X' \in \mathbb{R}^{d_X}$ and $Y' \in \mathbb{R}^{d_Y}$ are another pair of Gaussian random vectors with independent elements (i.e. their covariance matrices $\mSigma_{X'}$, $\mSigma_{Y'}$, and $\mSigma_{X'Y'}$ are all diagonal). We first show that $\lim_{m' \to \infty}\frac{C(\langle X'\rangle_{m'}, \langle Y' \rangle_{m'})}{\log m'} \le d(X',Y')$. Note that $d(X',Y')=d(X,Y)$ since $\mT_X, \mT_Y$ are invertible transformations.

However, $C(\langle X\rangle_m,\langle Y \rangle_m), C(\langle X'\rangle_{m}, \langle Y' \rangle_{m})$ may not be equal in general (note that\\ $\langle \mT_X \langle X' \rangle_{m} \rangle_m \ne \langle \mT_X X' \rangle_m = \langle X \rangle_m$) (similarly for $Y'$ and $Y$). Our proof proceeds by showing that $\lim_{m \to \infty}\frac{C(\langle X\rangle_m,\langle Y \rangle_m)}{\log m} \le \lim_{m \to \infty}\frac{C(\langle X'\rangle_{m}, \langle Y' \rangle_{m})}{\log m'}$. This is proved using the following ideas: (i) $\lim_{m \to \infty}\frac{C(\langle X'\rangle_m,\langle Y' \rangle_m)}{\log m} = \lim_{m \to \infty}\frac{C(\langle X'\rangle_{\alpha m},\langle Y' \rangle_{\alpha m})}{\log m}$ for any fixed $\alpha \in [0,\infty]$ (ii) $\langle \mT_X \langle X' \rangle_{\alpha m} \rangle_m = \langle X \rangle_m$ and $\langle \mT_Y \langle Y' \rangle_{\alpha m} \rangle_m = \langle Y \rangle_m$ can be achieved with high probability by choosing a large $\alpha$.
\end{proof}

\section{Numerical Evaluation} \label{sec:numeric}
In this section, we numerically verify the asymptotic behaviors of $C_\epsilon(X, Y)$, $C(X_\epsilon, Y_\epsilon)$ and $C(\langle X\rangle_m, \langle Y \rangle_m)$ defined in Section \ref{sec:approx} through two setups. In the first setup, we examine the growth rate of the approximate common information as a function of the approximation error $\epsilon$ (or the quantization size $m$). We use the second setup to illustrate the linear relationship between the approximate common information and common information dimension $d(X,Y)$.

In all simulation results below, we obtain the value of approximate common information as follows:
\begin{itemize}
    \item To measure $C(X_\epsilon,Y_\epsilon)$ (note that there exist multiple sequences\footnote{Note that there are at most $2^{\min\{d_X,d_Y\}}$ $(X_\epsilon, Y_\epsilon)$ that satisfy \eqref{def: seq}, for each $\epsilon$.} $\{(X_\epsilon,Y_\epsilon)\}_{\epsilon>0}$ that satisfy the requirements in \eqref{def: seq}), we choose two representative sequences and plot the results for both: $\{(\underline{X}_\epsilon,\underline{Y}_\epsilon)\}$ which has the minimum common information among such sequences for all $\epsilon > 0$, and $\{(\overline{X}_\epsilon,\overline{Y}_\epsilon)\}$ which has the maximum common information. We calculate the $C((\underline{X}_\epsilon,\underline{Y}_\epsilon))$ and $C(\overline{X}_\epsilon,\overline{Y}_\epsilon)$ using the closed-form solution in \eqref{eq:wyner}\cite{satpathy2015gaussian}.
    \item To calculate the $\epsilon$-approximation common information $C_\epsilon(X,Y)$, we solve the optimization problem in \eqref{def: opt} numerically using SciPy \cite{2020SciPy-NMeth}.
    \item For the $C(\langle X\rangle_m, \langle Y \rangle_m)$, we use a pair of upper and lower bounds from the proof of Theorem~\ref{thm:quant}:
$$I(\langle X \rangle_m ; \langle Y \rangle_m) \le C(\langle X \rangle_m , \langle Y \rangle_m) \le \sum_{X=Y} H(\langle X_i \rangle_m) + \sum_{X\ne Y} C(X_i, Y_i).$$
    In Figure~\ref{fig:numeric} and \ref{fig:error}, the exact value of $C(\langle X\rangle_m, \langle Y \rangle_m)$ lies in the colored area in between.
\end{itemize} 

\subsection{Setup 1} We let $X \in \mathbb{R}^4$ and $Y \in \mathbb{R}^4$ be jointly Gaussian vectors with zero means and covariance matrices $$\mSigma_X,\mSigma_Y = \left[\begin{matrix}1 & 0.5 & 0 & 0\\ 0.5 & 1 & 0 & 0\\ 0 & 0 & 1 & 0 \\ 0 & 0 & 0 & 1  \end{matrix}\right], 
\mSigma_{XY} = \left[\begin{matrix}1 & 0.5 & 0 & 0\\ 0.5 & 1 & 0 & 0\\ 0 & 0 & 1 & 0 \\ 0 & 0 & 0 & 0.3 \end{matrix}\right].$$ It is evident that $X_1 = Y_1, X_2 = Y_2, X_3 = Y_3$ almost surely, and $\text{rank}(\mSigma) = 5 < \text{rank}(\mSigma_X) + \text{rank}(\mSigma_Y)$, thus, $X$ and $Y$ are jointly singular.

\begin{figure}[!th]
  \centering
  \includegraphics[width=0.53\textwidth]{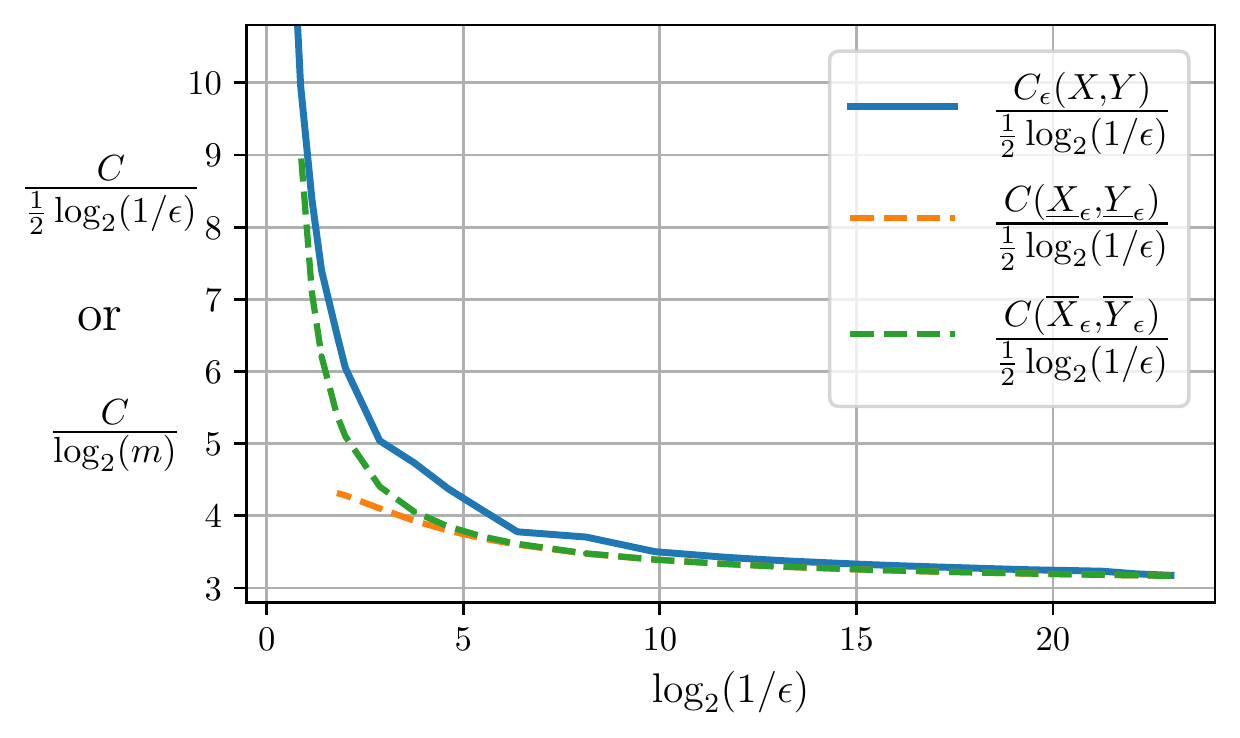}
  \includegraphics[width=0.45\textwidth]{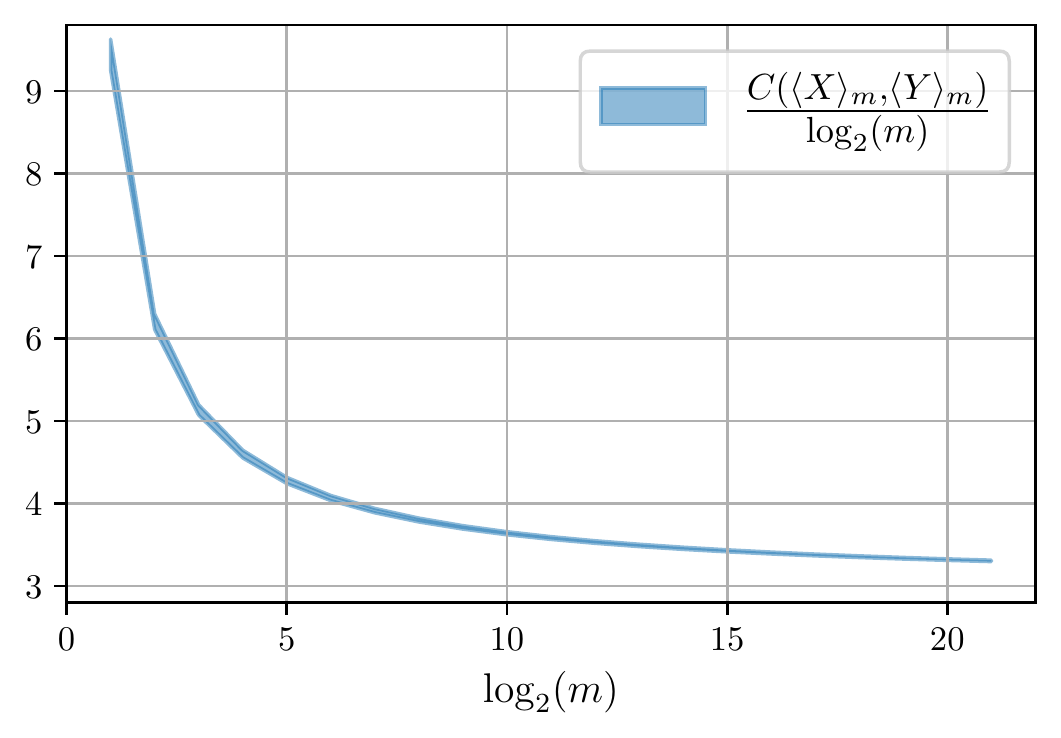}
  \vspace{-.2in}
  \caption{(a) the growth rate of $C_\epsilon(X,Y)$ and $C(X_\epsilon,Y_\epsilon)$; (b) the growth rate of $C(\langle X\rangle_m, \langle Y \rangle_m)$.}
  \vspace{-.2in}
  \label{fig:numeric}
\end{figure}

Figure~\ref{fig:numeric} illustrates the normalized common information $\frac{C_\epsilon(X,Y)}{\frac{1}{2}\log(1/\epsilon)}$, $\frac{C(X_\epsilon,Y_\epsilon)}{\frac{1}{2}\log(1/\epsilon)}$ and $\frac{C(\langle X\rangle_m, \langle Y \rangle_m)}{\log(m)}$ for the three formulations introduced in Section \ref{sec:ps}. They are plotted against the approximation error $\epsilon$ (or the quantization size $m$). We observe that both $\frac{C_\epsilon(X,Y)}{\frac{1}{2}\log(1/\epsilon)}$ and $\frac{C(X_\epsilon,Y_\epsilon)}{\frac{1}{2}\log(1/\epsilon)}$ converge to $d(X,Y) = \text{rank}(\mSigma_X) + \text{rank}(\mSigma_Y) - \text{rank}(\mSigma) = 3$ as $\epsilon$ approaches $0$. Similarly, $\frac{C(\langle X\rangle_m, \langle Y \rangle_m)}{\log(m)}$ converges to $d(X,Y)$ as $m$ approaches $\infty$. These verify our results in Section \ref{sec:main}. Moreover, they reach a value that is close to $d(X,Y)$ (e.g., a value $<4$) quickly, even when $\epsilon$ is relatively large (or $m$ is relatively small). The trade-off between $\epsilon$-approximate common information $C_\epsilon(X,Y)$ and error $\epsilon$ indicates that the common information dimension $d(X,Y)$ provides a theoretical limit of distributed simulation, on the maximum achievable accuracy given a finite number of bits to represent the common randomness; or equivalently, on the minimum number of bits required for the shared randomness to achieve a target simulation accuracy. Similarly, the trade-off between $C(\langle X\rangle_m, \langle Y \rangle_m)$ and the approximation precision $m$ reflects the special case when uniform quantization is applied to the target random variables.

% {\color{blue} [Y: common information bits. fix one specific epsilon. X: axis: vary the common information dimension. Only simulate the upper bound sequence.] [add conclusion] [move proof to appendix]}
% \textcolor{cyan}{From the plot above, can we give some specific examples? Along the lines: this result shows that, if we want to simulate X,Y using the sequence $\underline{X}_\epsilon,\underline{Y}_\epsilon)$  and precision **** then we need *** bits of common randomness, while if we want to achieve precision *** we need **** bits.  Do you think we could have a plot where the y axis is just C (not divided with log(1/e) , and the x-axis takes the discrete values 1,2,3,4,5 and we show two plots, once corresponding to say $e=1/2^5$ and the other to $e=2^10$? Using a constructed example? (we can describe the example in the appendix the exact matrices we used? )}

\subsection{Setup 2} In this example, we use $X \in \mathbb{R}^7$ and $Y \in \mathbb{R}^7$ with $\mSigma_X = \mSigma_Y = \mI_7$, while we choose cross-covariance matrices $\mSigma_{XY}$ to be a diagonal matrix with $d(X,Y)$ number of diagonal elements set to be $1$ and the rest to be $0.5$.

\begin{figure}[!th]
  \centering
  \includegraphics[width=0.46\textwidth]{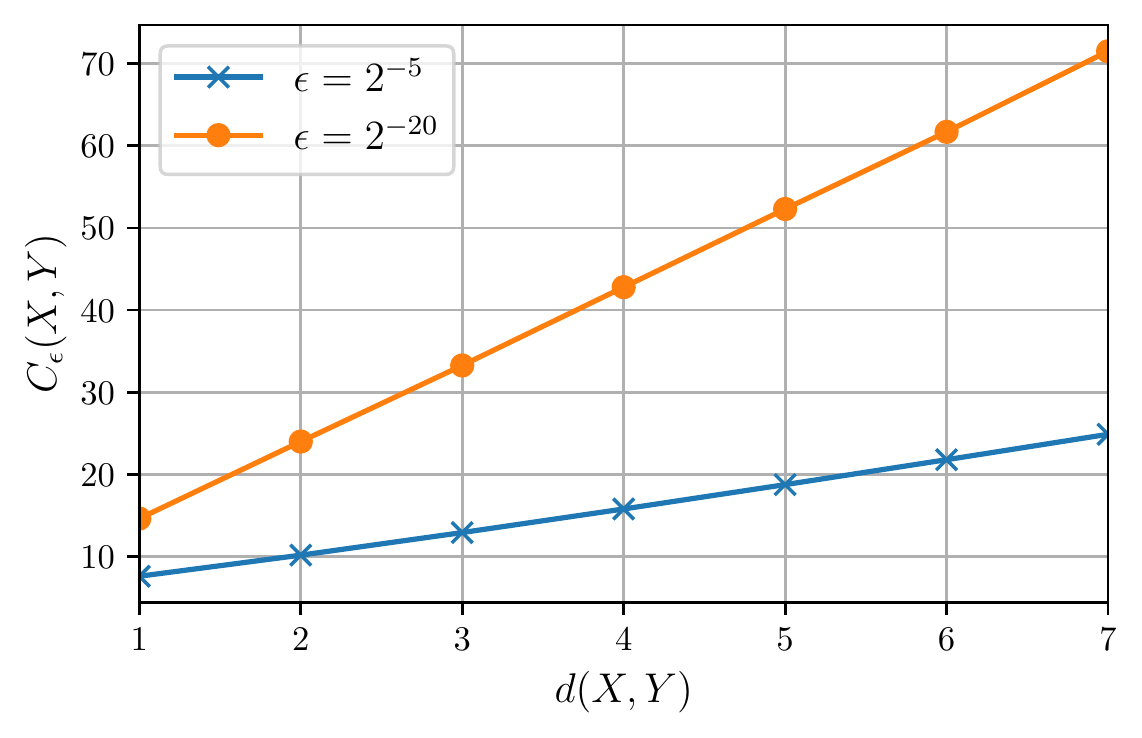}
  \includegraphics[width=0.47\textwidth]{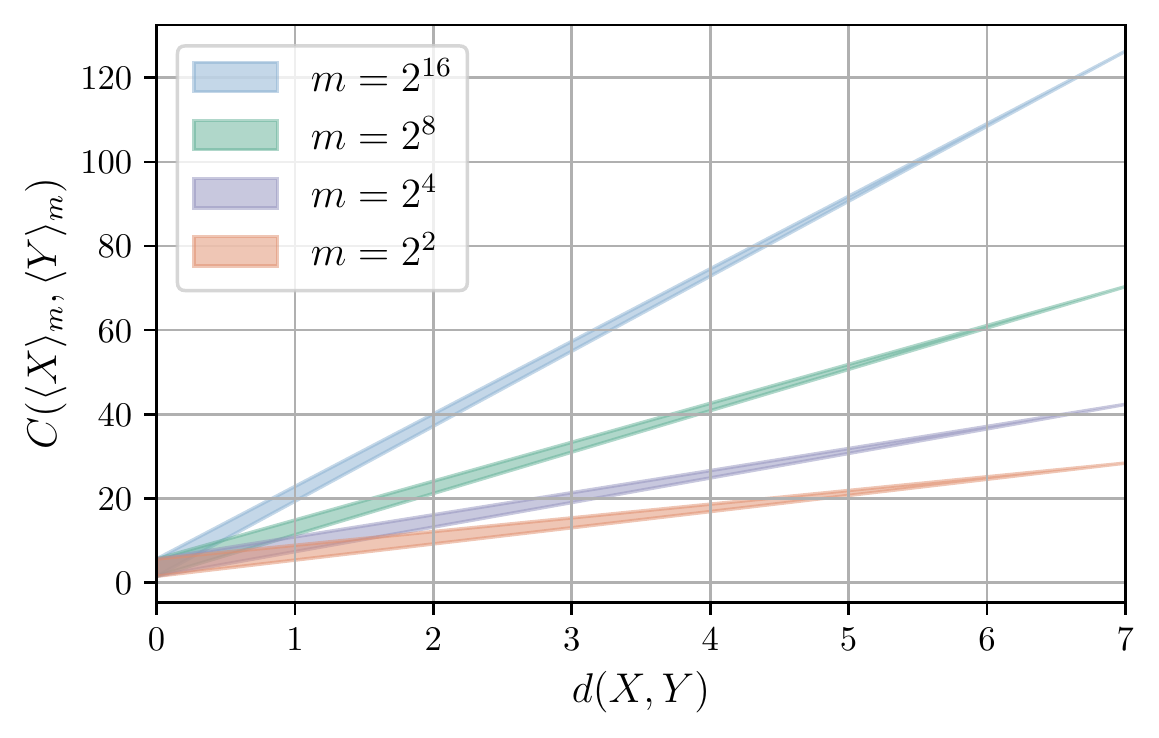}
  \vspace{-.2in}
  \caption{common information dimension $d(X,Y)$ vs (a) the $\epsilon$-approximate common information $C_\epsilon(X,Y)$; (b) the common information of quantized variables $C(\langle X\rangle_m,\langle Y\rangle_m)$.}
  \vspace{-.2in}
  \label{fig:error}
\end{figure}
% \begin{figure}[th!]
%   \centering
%   \includegraphics[width=0.5\textwidth]{Figures/quant_acc.pdf}
%   \vspace{-.2in}
%   \caption{The discretized common information $C(\langle X\rangle_m,\langle Y\rangle_m)$ vs common information dimension $d(X,Y)$.}
%   \label{fig:quant}
% \end{figure}

In Figure \ref{fig:error}(a), we plot $\epsilon$-approximate common information (i.e. the minimum number of bits required to approximate $X,Y$) versus different choices common information dimensions $d(X,Y)$. We use two different levels of accuracy: $\epsilon=2^{-5}$ and $\epsilon=2^{-20}$. We observe that the approximate common information grows linearly with the common information dimension $d(X,Y)$, where the slope is given by $\frac{1}{2}\log(1/\epsilon)$, as we also proved in Theorem~\ref{thm:appr}. In addition, this plot provides guidance on the minimum number of bits that need to be shared to perform the distribution simulation within a given error. For instance, to simulate a target distribution with $d(X,Y)=5$, we need to share $19$ bits to achieve a relatively low $2^{-5}$ accuracy, or $53$ bits to achieve a relatively high $2^{-20}$ accuracy.

In Figure \ref{fig:error}(b), we plot the common information between quantized random variables $\langle X \rangle_m$, $\langle Y \rangle_m$ versus different values of common information dimension $d(X,Y)$. We consider 4 different levels of quantization: $m=2^{k}$, where $k\in \{1,2,4,8\}$. We observe that the approximate common information grows linearly with the common information dimension $d(X,Y)$, and the slope is roughly $\log(m)$ for large enough $m$, as we also proved in Theorem~\ref{thm:quant}. Moreover, this plot specifies the minimum number of bits needed to distributedly simulate target random variables that are uniformed quantized into discrete values, with a known quantization precision.
\section{Conclusions and Open Questions}\label{sec:concl}
%\textcolor{cyan}{In this paper, we introduced the notion of common information dimension and calculated in closed form for arbitrary number of Gaussain random vectors and linear function class. We also examined three instances where the *** becomes important, in particular, ***** Several open questions remain. These include, ****}

In this paper, we introduced (three variants of) the notion of common information dimension, which is suitable for general continuous random variables. We calculated these in closed form for an arbitrary number of Gaussian random vectors and the linear function class, and provided an efficient algorithm to construct a random variable with the minimum dimension that captures the common randomness. Our solutions use the covariance matrices of the target distribution.
We further explored scenarios where the common information is (nearly) infinity in bits, and therefore the asymptotic behavior becomes important. In particular, we studied three different approximation problems in the context of distributed simulation. We proved, for the case of two Gaussian sources, that the approximate common information in all cases grows at a rate proportional to the common information dimension. This gives interpretations to the common information dimension and establishes a link to Wyner’s common information.
Several open questions remain. These include whether our solutions of common information dimension (in terms of both the calculation and construction) can be extended to other function classes and beyond the Gaussian distribution; whether the common information dimension can be approximately calculated in a data-driven approach, without the knowledge of the distribution; and whether the study of common information dimension can provide insights into solving Wyner’s common information involving more than two sources.

%%%% Appendix %%%%%
\newpage
\newpage
\appendix
\subsection{Proofs of lemmas in Theorem~\ref{thm:main-gauss}}\label{app:thm-gauss-2}
\lemONE*
\begin{proof}[ of Lemma~\ref{lem:1}]
    If $\va^\top X =  \vb^\top X$ almost surely, the multiplying both sides by $X^\top$ and taking expectation gives $\va^\top \mSigma =  \vb^\top \mSigma$. It remains to show the other direction; namely, if $\va^\top \mSigma =  \vb^\top \mSigma$, then $\va^\top X =  \vb^\top X$ almost surely. We have that the second moment of $(\va^\top-\vb^\top) X$ is given by
    \begin{align}
        \mathbb{E}[((\va-\vb)^\top X)^2] &= \mathbb{E}[(\va-\vb)^\top XX^\top (\va-\vb)]\nonumber\\
        &=(\va-\vb)^\top \mSigma (\va-\vb) = 0.
    \end{align}
    It follows that $(\va-\vb)^\top X=0$ almost surely.
\end{proof}

\corone*
\begin{proof}[ of Corollary~\ref{cor:1}]
    By Lemma~\ref{lem:1}, there is a subset $I \subseteq \{1,...,d_X\}$ such that $|I|=\mbox{rank}(\mSigma_X)$, and $X=\mB X_I$ for some $\mB \in \mathbb{R}^{d_x\times |I|}$. Then we have that $X_I \indep Y | (V,W)$ if and only if $ X \indep Y | (V,W)$.
\end{proof}

\lemrankNx*
\begin{proof}[ of Lemma~\ref{lem:rankNx}]
Suppose towards a contradiction that $\mN_X$ is not full-row-rank. Hence, there exists $\vb \neq \vzero$ such that $\vb^\top \mN_X = \vzero$. As $\mN$ is full-row-rank, 
we have that $\vb^\top \mN_Y \neq \vzero$. From equation \eqref{eq:NX} we have that $\vb^\top \mN_Y Y = \vb^\top \mN_X X$. Hence, we have that $\vb^\top \mN_Y Y = \vzero$, and $\vb^\top \mN_Y \neq \vzero$. As a result, by Lemma~\ref{lem:1}, we have $\vb^\top \mN_Y \mSigma_Y = \vzero, \vb^\top \mN_Y \neq \vzero$, which contradicts our assumption that $\mSigma_Y$ is full-rank. Therefore, it has to hold that $\mbox{rank}(\mN_X) = \mbox{rank}(\mN_Y) = \mbox{rank}(\mN)$.
\end{proof}

\lemMone*
\begin{proof}[ of Lemma~\ref{lem:M1}]
Since $\mM_X$ and $\mM_Y$ are full-rank, the corresponding linear mappings are one-to-one. If $X,Y$ are conditionally independent given $(V,W)$, then for any sets $\sS_X, \sS_Y$, we have that
\begin{align}
    & \mathbb{P}[\mM_XX\in \sS_X, \mM_YY\in \sS_Y|(V,W)] \\
    = \ & \mathbb{P}[X\in \mM_X^{-1}\sS_X, Y\in \mM_Y^{-1}\sS_Y|(V,W)]\nonumber \\
    = \ & \mathbb{P}[X\in \mM_X^{-1}\sS_X|(V,W)] \mathbb{P}[Y\in \mM_Y^{-1}\sS_Y|(V,W)]\nonumber \\
    = \ & \mathbb{P}[\mM_XX\in \sS_X|(V,W)] \mathbb{P}[\mM_YY\in \sS_Y|(V,W)],
\end{align}
where $\mM_X^{-1}\sS_X = \{\mM_X^{-1}x|x\in \sS_X\}$ and $\mM_Y^{-1}\sS_Y$ is defined similarly. This shows that $\mM_XX$ and  $\mM_YY$ are conditionally independent given $(V,W)$. The proof of the other direction that $\mM_XX \indep \mM_YY|(V,W)$ implies $X\indep Y|(V,W)$ is similar.  
\end{proof}

The following lemma is utilized in the proof of Lemma~\ref{lem:cond}.
% \lemMtwo*
\begin{restatable}{lemma}{lemMtwo}\label{lem:M2}
%\begin{lemma}\label{lem:M2}
Define $\mN'_X$, $\mM_Y$  as in \eqref{eq:NX2} and \eqref{eq:M}. Then % We have that
    \begin{align*}
     \nexists~\va, \vb \text{ such that } \va \neq \vzero\text{ and } \va^\top \mN_X' X = \vb^\top \mM_Y Y.
    \end{align*}
%\end{lemma}
\end{restatable}

\begin{proof}[ of Lemma~\ref{lem:M2}]
Suppose towards a contradiction that $\exists~\va, \vb$ such that $\va \neq \vzero$ and  $\va^\top \mN_X' X = \vb^\top \mM_Y Y$. Then, we can extend $\mN$ by adding the following row $\begin{bmatrix} \va^\top\mN_X' & -\vb^\top\mM_Y\end{bmatrix}$, which is linearly independent on rows of $\mN$ because $\mN_X'$ is the complementary row space of $\mN_X$ and $\mN = \begin{bmatrix}\mN_X & -\mN_Y\end{bmatrix}$. This contradicts the fact that $\mN$ is a basis for the null space of $\mSigma_{X,Y}$.
\end{proof}

\lemcond*
\begin{proof}[ of Lemma~\ref{lem:cond}]
Suppose towards a contradiction that there is $\va\neq \vzero$ such that $\va^\top \mSigma'_{|\mN_X X}=\vzero$, where $\mSigma'_{|\mN_X X}$ is the covariance matrix of $\begin{bmatrix}\mN_X'X \\ \mN_Y'Y\end{bmatrix}$ conditioned on $\mN_X X$. Since the mean of $\begin{bmatrix}\mN_X'X \\ \mN_Y'Y\end{bmatrix}$ conditioned on $\mN_X X$ is a linear function of $\mN_X X$ (say the conditional mean is $\mB \mN_X X$),  then, by Lemma~\ref{lem:1} we have that $$\va^\top (\begin{bmatrix}\mN_X' X \\ \mN_Y'Y\end{bmatrix} - \mB \mN_X X) = \vzero \text{ almost surely}.$$
Let us partition $\va$ as $\va = \begin{bmatrix} \va_X \\ \va_Y \end{bmatrix}$. Using the fact that $\mN_X X = \mN_Y Y$ we get that
$$\va_X^\top\mN_X'X = \vb ^\top \begin{bmatrix}\mN_Y Y \\ \mN_Y'Y\end{bmatrix} = \vb^\top \mM_Y Y,$$
where $\vb = \begin{bmatrix} \mB^\top \va \\ - \va_Y \end{bmatrix}$, which contradicts Lemma~\ref{lem:M2}.
\end{proof}

\subsection{Proof of Theorem~\ref{thm:gauss-gen}}\label{app:thm-gauss-gen}
\thmgaussgen*
\begin{table}[t!]
  \centering
  \caption{Table of notation for proof of Theorem \ref{thm:gauss-gen}}
  \label{tab:gauss-gen}
  \begin{IEEEeqnarraybox}[\IEEEeqnarraystrutmode%
    % \IEEEeqnarraystrutsizeadd{2pt}{1pt}% uncomment for more spacing
    ]{l"l}
    \toprule
    \textbf{Notation} & \textbf{Definition} \\
    \midrule
    \mA_i & \text{basis of the row space of } \mSigma_{i|1:i-1}
    \hfill\refstepcounter{equation}\textup{(\theequation)} \label{eq:def-Ai} \\
    \midrule
    \mB_i & \text{basis of the row space of } \mSigma_{i+1:n|1:i-1}
    \hfill\refstepcounter{equation}\textup{(\theequation)} \label{eq:def-Bi}\label{eq:def-U}\\
    \midrule
    U_i & U_i = \mA_i X_i
    \hfill\refstepcounter{equation}\textup{(\theequation)}\label{eq:def-Y}\\
    \midrule
    Y_i & Y_i = \mB_i [X_{i+1}^\top\ \cdots\ X_n^\top]^\top
    \hfill\refstepcounter{equation}\textup{(\theequation)} \\
    \midrule
    \tilde{\mN}_i & \text{basis of the null space of } \mSigma_{U_i,Y_i|X_1,\cdots,X_{i-1}},
    \hfill\refstepcounter{equation}\textup{(\theequation)} \label{eq:def-Nt}\\
    & \tilde{\mN}_i \in \mathbb{R}^{r(\tilde{\mN}_i)\times (r(U_i)+r(Y_i))} \\
    \midrule
    \mN_i, \bar{\mN}_i & \text{partition } \tilde{\mN}_i \text{ as }[\mN_i\ \bar{\mN}_i] 
    \hfill\refstepcounter{equation}\textup{(\theequation)} \label{eq:def-N}\\
    & \mN_i \in \mathbb{R}^{r(\tilde{\mN}_i)\times r(U_i)}\ \bar{\mN}_i \in \mathbb{R}^{r(\tilde{\mN}_i)\times r(Y_i)}\\
    \midrule
    Z_i & Z_i = \mN_i U_i
    \hfill\refstepcounter{equation}\textup{(\theequation)} \label{eq:def-Z}\\
    \midrule
    \mC_i & \text{basis of the row space of } \mSigma_{X_i|Z_1,\cdots,Z_n} 
    \hfill\refstepcounter{equation}\textup{(\theequation)} \label{eq:def-C}\\
    \midrule
    T_i & T_i = \mC_iX_i
    \hfill\refstepcounter{equation}\textup{(\theequation)} \label{eq:def-T}\\
    \bottomrule
  \end{IEEEeqnarraybox}
  \vspace{-0.2in}
\end{table}

%Let 
% \begin{align}
%     &A_i \text{ be basis of the row space of } \mSigma_{i|1:i-1},\label{eq:def-Ai} \\
%     &B_i \text{ be basis of the row space of } \mSigma_{i+1:n|1:i-1},\label{eq:def-Bi} \\
%     &U_i = A_i X_i,\\
%     &Y_i = B_i [X_{i+1}^\top\ \cdots\ X_n^\top]^\top, \\
%     &\tilde{N}_i^{r(\tilde{N}_i)\times (r(U_i)+r(Y_i))} \text{ be basis of the null space of }\nonumber \\ 
%     &\qquad \mSigma_{U_i,Y_i|X_1,\cdots,X_{i-1}},\nonumber \\
%     &\text{partition } \tilde{N} \text{ as } [N_i^{r(\tilde{N}_i)\times r_{U_i}}\ \bar{N}_i^{r(\tilde{N}_i)\times r(Y_i)}],\label{eq:def-N}\\
%     &Z_i = N_i U_i,\label{eq:def-Z}\\
%     &C_i \text{ be basis of the row space of } \mSigma_{X_i|Z_1,\cdots,Z_n}\nonumber \\
%     &T_i = C_iX_i.\label{eq:def-T}
% \end{align}

\begin{proof}[ of Theorem~\ref{thm:gauss-gen}]
Consider the quantities defined in Table~\ref{tab:gauss-gen}; the proof logic proceeds as follows. We sequentially extract the common information between the $X_i$ by constructing a sequence of variables $Z_i$,  where $Z_i$  contains the amount of information that $X_i$ contains about $X_{i+1},\cdots,X_n$ and that $X_1,\cdots,X_{i-1}$ does not contain. {We recall from the proof of Theorem~\ref{thm:main-gauss} that: $(i)$ to remove the parts of $X_i$ that can be determined by $X_1,\cdots,X_{i-1}$, we can consider $\mA_iX_i$, where $\mA_i$ is a basis for the row space of $\mSigma_{X_i|X_1,\cdots,X_{i-1}}$ (Lemma~\ref{lem:1}); and $(ii)$ to find the common information between two non-singular Gaussian random variables $X,Y$, we need to consider the null space of $\mSigma_{XY}$ (\eqref{eq:result1} and \eqref{eq:result2}). We utilize this in the following to build $Z$:}\\ 
$\bullet$ To get the common information between $X_1$ and $X_2,\cdots,X_n$; we first
remove the singular part of $X_1$
%start by removing the redundancy in $X_1$. \textcolor{cyan}{What redundancy?} This is done 
by linearly transforming $X_1$ using a matrix $\mA_1$ that captures the non-singular part of $X_1$. We denote the resulting vector variable by $U_1$.\\
$\bullet$ Similarly, we remove the singular part of $[X_2,\cdots,X_n]$ using a matrix $\mB_1$ to produce $Y_1$.\\
$\bullet$ We then argue (using $(ii)$) that the parts of $X_2,\cdots,X_n$ that can be determined from $X_1$ (common information), denoted as $Z_1$, can be obtained from the null space of $\mSigma_{U_1Y_1}$.\\
$\bullet$ To obtain $Z_2$ which contains the amount of information that $X_2$ contains about $X_3,\cdots,X_n$ that $X_1$ does not contain: we first eliminate the parts of $X_2$ that can be obtained from $X_1$. We argue (as described in $(i)$) using Lemma~\ref{lem:1}, that the vectors in the null space of $\mSigma_{X_2|X_1}$ are precisely the linear combinations of $X_2$ that can be obtained from $X_1$. Hence, we remove the parts of $X_2$ that can be obtained from $X_1$ using a matrix $\mA_2$ that captures the non-singular part of the matrix $\mSigma_{X_2|X_1}$, to obtain $U_2$.\\
$\bullet$ Similarly, we remove the parts of $[X_3,\cdots,X_n]$ that can be obtained from $X_1$ using a matrix $\mB_2$ to produce $Y_2$.\\
$\bullet$ Similar to $Z_1$, the desired $Z_2$ can then be obtain from the null space of $\mSigma_{U_2Y_2}$.\\

More generally, to obtain $Z_i$ which contains the amount of information that $X_i$ contains about $X_{i+1},\cdots,X_n$ that $X_1,\cdots,X_{i-1}$ does not contain: we first eliminate the parts of $X_i$ that can be obtained from $X_1,\cdots,X_{i-1}$. This is done using a matrix $\mA_i$ that captures the non-singular part of the matrix $\mSigma_{X_i|X_1,\cdots,X_{i-1}}$, to obtain $U_i$. Similarly, we remove the parts of $[X_{i+1},\cdots,X_n]$ that can be obtained from $X_1,\cdots,X_{i-1}$ using a matrix $\mB_i$ to produce $Y_i$. Then, $Z_i$ can be obtained from the null space of $\mSigma_{U_iY_i}$. The method of constructing $(V,W)$ that break the dependency between $X_1,\cdots,X_n$, where $H(V)<\infty$ and $W$ has minimum dimension is summarized in Algorithm~\ref{alg:gauss-gen}. The main part of our proof shows that $Z=[Z_1,\cdots,Z_n]$ is sufficient to break the dependency between $X_1,\cdots,X_n$ up to finite randomness (Lemma~\ref{lem:T-full}), and that $Z$ is also a necessary part of every $W$ that breaks the dependency up to finite randomness and hence $d_W\geq d_Z$ (Lemma~\ref{lem:Z-full}).

We next formally prove our result. We state and prove two properties of the quantities in Table~\ref{tab:gauss-gen}. The first property is that the conditional covariance matrix $\mSigma_{T_1,\cdots,T_n|Z_1,\cdots,Z_n}$ is full-rank, where $T_i$,  defined in \eqref{eq:def-T}, can be thought of as the part of $X_i$ that cannot be expressed as a deterministic function of $Z_1,\cdots,Z_n$.

\begin{lemma}\label{lem:T-full}
    The conditional covariance matrix of $T=[T_1,\cdots,T_n]$ conditioned on $Z=[Z_1,\cdots,Z_n]$, namely, $\mSigma_{T|Z}$ is non-singular, where $T_i$ is defined in \eqref{eq:def-T} and $Z_i$ is defined in \eqref{eq:def-Z}.
\end{lemma}

\begin{proof}[ of Lemma~\ref{lem:T-full}]
%\begin{proof}
    Suppose towards a contradiction that there exists $\va\neq \vzero$ such that $\va^\top \mSigma_{T|Z}= \vzero$. Let us partition $\va$ as $[\va_1,\cdots,\va_n]$, where $\va_i\in \mathbb{R}^{d_{T_i}}$.
    Let $j$ denote the largest index for which $\va_j\neq \vzero$. By Lemma~\ref{lem:1}, we have that
    \begin{equation}\label{eq:lin-T}
        \va^\top T = \va'^\top Z \quad \text{almost surely,}
    \end{equation}
    for some vector $\va'$. We recall that by definition of $T_i,Z_i$ in \eqref{eq:def-T} and \eqref{eq:def-Z}, they are linear functions of $X_i$. Also, by definition of $j$ as the largest index with $\va_j\neq \vzero$, we have that $T_{j+1}, \cdots, T_n$ do not appear in \eqref{eq:lin-T}; only $Z_{j+1},\cdots,Z_n$ appear. By expanding $T_i,Z_i$ for $i<j$ as linear functions of $X_i$, we get that there exist vectors $\ve_1,\ve_2$, $\vc_1,\vc_2,\vc_3$ such that
    \begin{equation}
    \begin{aligned}\label{eq:wts-1}
        & \va_j^\top T_j + \vc_1^\top [Z_{j+1}^\top\ \cdots\ Z_n^\top]^\top\\ = \ & \ve_1^\top X_{j}+ \ve_2^\top [X_{j+1}^\top\ \cdots\ X_{n}^\top]^\top\\
        = \ & \vc_2^\top X_{j-1} + \vc_3^\top [X_1\ \cdots\ X_{j-2}]^\top \text{ almost surely.}
    \end{aligned}
    \end{equation}
    We want to show that $X_{j-1}$ can be replaced by $Z_{j-1}$ in the previous equation (after possibly changing $\vc_2,\vc_3$). In particular, we want to show that there exist vectors $\tilde{\vc}_2,\tilde{\vc}_3$ such that
    \begin{equation}
    \begin{aligned}\label{eq:wts}
        & \va_j^\top T_j + \vc_1^\top [Z_{j+1}^\top\ \cdots\ Z_n^\top]^\top
        =  \tilde{\vc}_2^\top Z_{j-1} + \tilde{\vc}_3^\top [X_1^\top \ \cdots\ X_{j-2}^\top]^\top \text{ almost surely.}
    \end{aligned}
    \end{equation}
    
    We have three possible cases:\\
    $(i)$ $\vc_2$ is in the left null space of $\mSigma_{(j-1)|1:(j-2)}$: It follows from Lemma~\ref{lem:1} that $\vc_2^\top X_{j-1}$ is a linear function of $X_1,\cdots,X_{j-2}$. Hence, it can be replaced by a linear function of $X_1,\cdots,X_{j-2}$ in \eqref{eq:wts-1}, and  \eqref{eq:wts} holds with $\tilde{\vc}_2=\vzero$.\\
    $(ii)$ $[\ve_1\ \ve_2]$ is in the left null space of $\mSigma_{j:n|1:(j-2)}$: It follows from Lemma~\ref{lem:1} that $\ve_1^\top X_{j}+ \ve_2^\top [X_{j+1}^\top\ \cdots\ X_{n}^\top]^\top$ is a linear function of $X_1,\cdots,X_{j-2}$. Hence, \eqref{eq:wts} holds again with $\tilde{\vc}_2=\vzero$.\\
    $(iii)$ $\vc_2$ is in the row space of $\mSigma_{(j-1)|1:(j-2)}$ and $[\ve_1\ \ve_2]$ is in the row space of $\mSigma_{j:n|1:(j-2)}$: by definition of $\mA_{j-1},\mB_{j-1}$ in \eqref{eq:def-Ai} and \eqref{eq:def-Bi}, it follows that $\vc_2^\top =\bar{\vc}_2^\top \mA_{j-1}$ and $[\ve_1^\top\ \ve_2^\top]^\top= \bar{\ve}^\top \mB_{j-1}$ for some $\bar{\vc},\bar{\ve}$. Thus
    \begin{align}\label{eq:nolabel}
        \bar{\ve}^\top Y_{j-1}&=\ve_1^\top X_{j}+ \ve_2^\top [X_{j+1}^\top\ \cdots\ X_{n}^\top]^\top\nonumber \\
        &=\bar{\vc}_2^\top U_{j-1} + \vc_3^\top [X_1^\top\ \cdots\ X_{j-2}^\top]^\top \text{ almost surely.}
    \end{align}
    
    By Lemma~\ref{lem:1} it follows that $[-\bar{\vc}_2^\top\ \bar{\ve}^\top]$ is in the left null space of $\mSigma_{U_{j-1},Y_{j-1}|X_1,\cdots,X_{j-2}}$. By definition of $\mN_{j-1}$ in \eqref{eq:def-N} it follows that $\bar{\vc}_2$ is in the row space of $\mN_{j-1}$. Hence, we have that there is $\bar{\vc}_3$ such that
    \begin{align}
        &\ve_1^\top X_{j}+ \ve_2^\top [X_{j+1}^\top\ \cdots\ X_{n}^\top]^\top\nonumber \\
        = \ & \bar{\vc}_3^\top \mN_{j-1}U_{j-1} + \vc_3^\top [X_1^\top\ \cdots\ X_{j-2}^\top]^\top\nonumber \\
        = \ & \bar{\vc}_3^\top Z_{j-1} + \vc_3^\top [X_1^\top\ \cdots\ X_{j-2}^\top]^\top \text{ almost surely},
    \end{align}
    where the last equality follows by the definition of $Z_{j-1}=\mN_{j-1}U_{j-1}$ in \eqref{eq:def-Z}. Substituting in \eqref{eq:wts-1} it follows that there are vectors $\tilde{\vc}_2,\tilde{\vc}_3$ such that

    \begin{equation}
    \begin{aligned}
        \va_j^\top T_j + \vc_1^\top [Z_{j+1}^\top\ \cdots\ Z_n^\top]^\top = ~ \tilde{\vc}_2^\top Z_{j-1} + \tilde{\vc}_3^\top [X_1^\top\ \cdots\ X_{j-2}^\top]^\top \text{ almost surely.} 
    \end{aligned}
    \end{equation}
Equivalently, there are vectors $\tilde{\vc_2}, \tilde{\vc_3}$ such that
\begin{equation}
\begin{aligned}
        \va_j^\top T_j + \vc_1^\top [Z_{j+1}^\top\ \cdots\ Z_n^\top]^\top - \tilde{\vc}_2^\top Z_{j-1}
        = ~ \tilde{\vc}_3^\top [X_1^\top\ \cdots\ X_{j-2}^\top]^\top \text{ almost surely.} 
\end{aligned}
\end{equation}
Applying the same result $j-2$ times, it follows that $X_1\cdots,X_{j-2}$ can be replaced by $Z_1\cdots,Z_{j-2}$. Equivalently, there is a vector $\vc$ such that
\begin{equation}
    \va_j^\top T_j = \vc^\top Z \quad \text{almost surely}.
\end{equation}
This contradicts the fact that $T_j=\mC_j X_j$ where $\mC_j$ is a basis of the row space of $\mSigma_{X_i|Z_1,\cdots,Z_{n}}$ (recall that $\va_j\neq \vzero$).
This concludes the proof that the conditional covariance matrix of $T$ conditioned on $Z$ is non-singular.
\end{proof}

By definition of $T_i=\mC_iX_i$ in \eqref{eq:def-T}, where $\mC_i$ is a basis for the row space of $\mSigma_{X_i|Z_1,\cdots,Z_n}$, it follows that conditioned on $Z_1,\cdots,Z_n$, $X_i$ can be obtained from $T_i$. Hence, conditioned on $Z_1,\cdots,Z_n$, breaking the dependency between $X_1,\cdots,X_n$ reduces to breaking the dependency between $T_1,\cdots,T_n$. From Lemma~\ref{lem:T-full}, we have that conditioned on $Z$, $T$ is jointly Gaussian with full-rank covariance matrix. Hence, by the result in \cite{li2017distributed}, there is a $V_Z$ with $H(V_Z)<\infty$ such that $T_1 \indep \cdots \indep T_n | (Z, V_Z)$, and hence, $X_1 \indep \cdots \indep X_n | (Z, V_Z)$. Since the covariance matrix of $T$ conditioned on $Z$ does not depend on the value of $Z$ and is only a function of the covariance matrix of $Z$, then $V_Z$ can be chosen the same for all $Z$. Hence, we can refer to $V_Z$ as $V$, which gives that $X_1 \indep \cdots \indep X_n | (Z, V)$. This shows that
\begin{align}\label{eq:gen-upper}
    d(X_1,\cdots,X_n) \leq \sum_{i=1}^nd_{Z_i}.
\end{align}
We next prove a second property that the quantities defined in Table~\ref{tab:gauss-gen} satisfy, that $\mSigma_Z$ is a full-rank covariance matrix.
\begin{lemma}\label{lem:Z-full}
    The covariance matrix $\mSigma_Z$ of $Z = [Z_1,\cdots,Z_n]$, where $Z_i$ is defined in \eqref{eq:def-Z} is full-rank.
\end{lemma}

\begin{proof}[ of Lemma~\ref{lem:Z-full}]
    Suppose towards a contradiction that $\mSigma_Z$ is not full-rank and pick a vector $\va\neq \vzero$ in the null space of $\mSigma_Z$. Partition $\va$ as $[\va_1,\cdots,\va_n]$, where $\va_i\in \mathbb{R}^{d_{Z_i}}$. Let $j$ be the largest index with $\va_j\neq \vzero$. By definition of $Z_i$ in \eqref{eq:def-Z} as linear function of $X_i$, we get that there is a vector $\vb$ such that
    \begin{align}
        \va_{j}^\top Z_{j} = \vb^\top [X_1^\top\ \cdots\ X_{j-1}^\top]^\top \quad \text{almost surely}.
    \end{align}
    Substituting for $Z_j$ using \eqref{eq:def-Z} we get that
    \begin{align} \label{eq:u-null}
        (\va_{j}^\top \mN_j) U_{j} = \vb^\top [X_1^\top\ \cdots\ X_{j-1}^\top]^\top \quad \text{almost surely}.
    \end{align}
    We notice by definition of $U_i,Y_i$ in \eqref{eq:def-U} and \eqref{eq:def-Y} that $\mSigma_{U_i|X_1,\cdots,X_{i-1}}$ and $\mSigma_{Y_i|X_1,\cdots,X_{i-1}}$ are both full-rank. By Lemma~\ref{lem:rankNx}, we also get that $\mN_j$ is full-rank. Hence, the vector $\va_{j}^\top \mN_j\neq \vzero$. By \eqref{eq:u-null} and Lemma~\ref{lem:1} the vector $\va_{j}^\top \mN_j\neq \vzero$ is in the null space of $\mSigma_{U_i|X_1,\cdots,X_{i-1}}$ which contradicts the fact that the matrix is full rank.
    We conclude that $\mSigma_Z$ is full rank.
\end{proof}

{We notice that by definition of $\mN_i$ in \eqref{eq:def-N} and Lemma~\ref{lem:1}, we have that $\mN_iU_i+\hat{\mN}_iY_i=\mA [X_1^\top\ \cdots\ X_{i-1}^\top]^\top$ a.s. for some matrix $\mA$, hence, $Z_i=\mN_iU_i=\mA'X_{-i}$ a.s. for some matrix $\mA'$ (as both $Y_i,[X_1,\cdots,X_{i-1}]$ do not include $X_i$).} By Lemma~\ref{lem:main-gauss}, as $Z_i=\mA' X_{-i}$ a.s., then, if $(V,W)$ breaks the dependency between $X_1,\cdots,X_n$, we must have that $Z_i$ is a deterministic function of $(V,W)$ for all $i$. And hence, similar to \eqref{eq:H1}, we get that
\begin{equation}
    Z = \mB W \quad \text{almost surely},
\end{equation}
for some matrix $\mB$. Next, by multiplying both sides by $Z^\top$ and taking the expectation we get that
\begin{equation}
    \mSigma_Z = \mB \mathbb{E}[WZ^\top].
\end{equation}
Hence, we have that
\begin{align}
    \sum_{i=1}^nd_{Z_i} \stackrel{(i)}{=} \mbox{rank}(\mSigma_Z) \leq \mbox{rank}(\mB) \leq d_W,
\end{align}
where $(i)$ follows from Lemma~\ref{lem:Z-full}. As $V,W$ are arbitrary in $\mathcal{W}$, we have that
\begin{align}\label{eq:gen-lower}
    d(X_1,\cdots,X_n) \geq \sum_{i=1}^nd_{Z_i}.
\end{align}
Combining \eqref{eq:gen-upper} and \eqref{eq:gen-lower}, it follows that
\begin{align}\label{eq:gen-dim}
    d(X_1,\cdots,X_n) = \sum_{i=1}^nd_{Z_i}.
\end{align}
We notice that by definition of $Z_i$ in \eqref{eq:def-Z} we have that
\begin{equation}\label{eq:dz-1}
    d_{Z_i}=\mbox{rank}(\mSigma_{i|1:i-1}) + \mbox{rank}(\mSigma_{i+1:n|1:i-1})- \mbox{rank}(\mSigma_{i:n|1:i-1}).
\end{equation}
To further simplify this expression, we state and prove the following property about the rank of the conditional covariance matrix $\mSigma_{X|Y}$ for two random vectors $X,Y$.
\begin{proposition}\label{prop:cond-rank}
    Let (X,Y) be a pair of jointly Gaussian random variables. We have that
    \begin{equation}
        \mbox{rank}(\mSigma_{X|Y}) = \mbox{rank}(\mSigma) - \mbox{rank}(\mSigma_{Y}),
    \end{equation}
    where $\mSigma, \mSigma_Y, \mSigma_{X|Y}$ are the covariance matrices of $[X,Y], Y$ and $X$ conditioned on $Y$ respectively.
\end{proposition}

\begin{proof}[ of Proposition~\ref{prop:cond-rank}]
    For a vector $\va\neq \vzero$ to be in the null space of $\mSigma_{X|Y}$, we have by Lemma~\ref{lem:1} that
    \begin{equation}\label{eq:xy}
        \va^\top X = \vb^\top Y \text{ almost surely}
    \end{equation}
    for some vector $\vb$. Without loss of generality we assume that $\mSigma_Y$ is full rank, otherwise by Lemma~\ref{lem:1} we can replace $Y$ with $Y'=\mB Y$, where $\mB$ is a matrix containing basis for the row space of $\mSigma_Y$. Hence, we may assume without loss of generality that
    \begin{equation}\label{eq:y-full}
        d_Y = \mbox{rank}(\mSigma_Y)
    \end{equation}
    Let $\mN$ denote the basis of the covariance matrix $\mSigma$,
    \begin{equation}\label{eq:NN2}
        \mN=[\mN_X^{r(\mN)\times d_X}\ \mN_Y^{r(\mN)\times d_X}].
    \end{equation}
    %denote the basis of null space for $\mSigma_{XY}$. 
    Similarly to Lemma~\ref{lem:rankNx}, we get that $\mSigma_Y$  full rank implies that
    \begin{equation}\label{eq:nx-full}
        \mN_X \text{ is full row rank}
    \end{equation}
    We also notice by \eqref{eq:xy} that the rank of $\mN_X$ is the dimension of the null space of $\mSigma_{X|Y}$. This is because from \eqref{eq:xy}, $[\va\ \vb]$ is in the null space of $\mSigma$ (by Lemma~\ref{lem:1}), hence, $\va$ is in the row space of $\mN_X$ by definition of $\mN_X$ in \eqref{eq:NN2}; conversely, every vector in the row space of $\mN_X$ satisfies \eqref{eq:xy} for some $\vb$ (also by definition of $\mN_X$). Hence, by noticing that the dimension of the null space of $\mSigma_{X|Y}$ is $d-\mbox{rank}(\mSigma_{X|Y})$ we get that
    \begin{align}
        d_X - \mbox{rank}(\mSigma_{X|Y}) &=~ \mbox{rank}(\mN_X)\stackrel{\eqref{eq:nx-full}}{=}~\mbox{rank}(\mN)\nonumber =~ d_X+d_Y-\mbox{rank}(\mSigma).
    \end{align}
    It follows that
    \begin{align}
        \mbox{rank}(\mSigma_{X|Y}) =~ & \mbox{rank}(\mSigma) - d_Y \nonumber
        \stackrel{\eqref{eq:y-full}}{=}~ \mbox{rank}(\mSigma) - \mbox{rank}(\mSigma_Y).
    \end{align}
\end{proof}
By substituting \eqref{eq:dz-1} in \eqref{eq:gen-dim} we get that
\begin{equation}
\begin{aligned}
    d(X_1,&\cdots,X_n) =~ \sum_{i=1}^n d_{Z_i}
    =~ \sum_{i=1}^n \left(\mbox{rank}(\mSigma_{i|1:i-1}) + \mbox{rank}(\mSigma_{i+1:n|1:i-1}) - \mbox{rank}(\mSigma_{i:n|1:i-1}) \right) \\
    \stackrel{(i)}{=} ~ & \sum_{i=1}^n \left(\mbox{rank}(\mSigma_{1:i}) - \mbox{rank}(\mSigma_{1:i-1}) + \mbox{rank}(\mSigma_{-i})- \mbox{rank}(\mSigma_{1:i-1}) - \mbox{rank}(\mSigma) + \mbox{rank}(\mSigma_{1:i-1})\right) \\
    \stackrel{(ii)}{=} ~ & \mbox{rank}(\mSigma) + \left(\sum_{i=1}^n \mbox{rank}(\mSigma_{-i})\right) - (n)\ \mbox{rank}(\mSigma) \\
    = ~ & \left(\sum_{i=1}^n\ \mbox{rank}(\mSigma_{-i})\right) - (n-1) \mbox{rank}(\mSigma),
\end{aligned}
\end{equation}
where $(i)$ follows by Proposition~\ref{prop:cond-rank}, and $(ii)$ follows from the fact that the first two terms inside the summation form a telescopic summation.
\end{proof}
\emph{Complexity of Computing $W$.} It is easy to see that every step (finding row space, inverse, matrix multiplication and summations) in the for loop in Algorithm~\ref{alg:gauss-gen} can be performed in $O((\sum_{i=1}^nd_{X_i})^3)$. It follows that the overall time complexity of computing $W$ is $O((\sum_{i=1}^nd_{X_i})^4)$.
\begin{corollary}
The time complexity of computing $W$ in Algorithm~\ref{alg:gauss-gen} is $O((\sum_{i=1}^nd_{X_i})^4)$.
\end{corollary}
\vspace{+.2in}
\lemRCID*
\begin{proof}[ of Lemma~\ref{lem:RCID}]
    Let $W$ be such that $W\in \sW$ and $d_W=d(X_1,\cdots,X_n)$. As $W\in \sW$, we have that $W$ is a linear function of $[X_1,\cdots,X_n]$, hence, it is jointly Gaussian. Then, we have that $d^R(W)\leq d_W$ \cite{renyi1959dimension}. This shows that
    \begin{equation}\label{eq:ren-1}
        d^R(X_1,\cdots,X_n)\leq d(X_1,\cdots,X_n).
    \end{equation}
    To show that $d^R(X_1,\cdots,X_n)\geq d(X_1,\cdots,X_n)$, we first observe that for every $W\in \mathcal{W}$, we have that $d^R(W)\in \mathbb{N}$, hence, the infimum in the definition of $d^R(X_1,\cdots,X_n)$ can be replaced by a minimum. Then there exists $W$ such that $W\in \sW$ and $d^R(W)=d^R(X_1,\cdots,X_n)$. As $W\in \sW$, we have that $W$ is jointly Gaussian. By Lemma~\ref{lem:1}, we have that there is $I\subseteq \{1,...,d_W\}$ such that $W_I$ has a non-singular covariance matrix and $W=\mB W_I$ for some matrix $\mB$. By construction of $W_I$ we have that $d^R(W_I)\leq d^R(W)$. Hence, without loss of generality, we assume that $\mSigma_W$ is non-singular, otherwise, we can replace $W$ by $W_I$. We also note that since $\mSigma_W^{1/2}W$ has independent entries, we have that $d^R(W)=d^R(\mSigma_W^{1/2}W)=d_W$ \cite{renyi1959dimension}. This shows that
    \begin{equation}\label{eq:ren-2}
        d(X_1,\cdots,X_n)\leq d^R(X_1,\cdots,X_n).
    \end{equation}
    Combining \eqref{eq:ren-1} and \eqref{eq:ren-2} concludes the proof.
\end{proof}

\subsection{Proof of Theorem~\ref{thm:gk}}\label{app:thm-gk}
\thmgk*
\begin{proof}[ of Theorem~\ref{thm:gk}]
    Let $\mN_{\tilde{\mSigma}}$ be the right null space of the matrix $\tilde{\mSigma}$. We will show that $\mbox{rank}(\mN_{\tilde{\mSigma}})=d^{GK}(X_1,\cdots,X_n)$.
    As in the proof of Lemma~\ref{lem:RCID}, we observe that for every $W$ in the optimization set, we have that $d^R(W)\in \mathbb{N}$. We also note that $d^{GK}(X_1,\cdots,X_n)\leq d_{X_1}$. Hence, the supremum in the definition of $d^{GK}$ is attainable; namely, there exists $W,f_1,\cdots,f_n$ with $W=f_i(X_i), \forall i\in [n]$ and $d^R(W)=d^{GK}(X_1,\cdots,X_n)$. As $f_i$ is linear, we have that $W$ is jointly Gaussian. By Lemma~\ref{lem:1}, we have that there is $I\subseteq \{1,...,d_W\}$ such that $W_I$ has a non-singular covariance matrix and $W=\mG W_I$ for some matrix $\mG$. By construction of $W_I$ we have that $d^R(W_I)\leq d^R(W)$. Hence, without loss of generality, we assume that $\mSigma_W$ is non-singular, otherwise, we can replace $W$ by $W_I$. We also note that since $\mSigma_W^{1/2}W$ has independent entries, we have that $d^R(W)=d^R(\mSigma_W^{1/2}W)=d_W$ \cite{renyi1959dimension}. 
    This shows that there is $W$ with $d_W=d^{GK}(X_1,\cdots,X_n)$ and matrices $\mK_1,\cdots,\mK_n$ such that $W=\mK_1X_1=\cdots=\mK_nX_
n$ almost surely and $W$ is non-singular, hence, $\mK_i$ is full row rank for each $i$.  Hence, the matrix
\begin{equation}
    \mN_W = \begin{bmatrix}
        \mK_1 \ -\mK_2 \ \mK_2 \ -\mK_3 \ \mK_3 \ \cdots \ \mK_{n-1} \ -\mK_n
    \end{bmatrix}
\end{equation}

is in the null space of $\tilde{\mSigma}$ by Lemma~\ref{lem:1}. As $\mK_1$ is full row rank, we get that $\mN_W$ is full row rank, hence,
\begin{equation}\label{eq:gk-lower}
    \mbox{rank}(\mN_{\tilde{\mSigma}})\geq \mbox{rank}(\mN_W) = r(\mN_W) = d^{GK}(X_1,\cdots,X_n).
\end{equation}
It remains to show that $\mbox{rank}(\mN_{\tilde{\mSigma}})\leq d^{GK}(X_1,\cdots,X_n)$. Let us partition $\mN_{\tilde{\mSigma}}$ (recall that $N$ is the basis for the left null space of $\tilde{\mSigma}$) as
\begin{equation}
    \mN_{\tilde{\mSigma}} = \begin{bmatrix}
        \mN_1^{r(\mN)\times d_{X_1'}} & \mN_2^{r(\mN)\times d_{X_2'}} & \cdots & \mN_n^{r(\mN)\times d_{X_n'}}
    \end{bmatrix}.
\end{equation}
By Lemma~\ref{lem:1} we get that the following is a common function
\begin{equation}\label{eq:gk-equal}
    W=\mN_1 \mF_1 X_1 = -\mN_2 \mF_2 X_2=\cdots = (-1)^{n+1} \mN_n \mF_n X_n.
\end{equation}
Hence, by definition of $d^{GK}$ we get that
\begin{equation}\label{eq:gk-1}
    d^{GK}(X_1,\cdots,X_n)\geq d^R(W)
\end{equation}
where the common function of variable $X_i$ achieving this is
\begin{equation}\label{eq:gk-construct}
    W=(-1)^{i+1}\mN_i \mF_i X_i.
\end{equation}
We observe that if $\va^\top \mN_1=\vzero$ for some vector $\va$, then by \eqref{eq:gk-equal} we have that $\va^\top \mN_iX'_i=\vzero, \forall i\in [n]$. As $X'_i$ has non-singular covariance matrix by construction, we get that $\va^\top \mN_i=\vzero, \forall i\in [n]$, hence, $\va^\top \mN=\vzero$. As $\mN_{\tilde{\mSigma}}$ is full row rank, we conclude that $\va=\vzero$. This shows that $\va^T \mN_1=\vzero$ if and only if $\va=\vzero$, hence, $\mN_1$ is full row rank.

We notice that since $X_1'=\mF_1X_1$ has non-singular covariance matrix and $\mN_1$ is full row rank, then $W$ has a non-singular covariance matrix. Hence, we have that $d_W=d^R(W)$, hence from \eqref{eq:gk-1} we get that
\begin{equation}\label{eq:gk:upper}
    d^{GK}(X_1,\cdots,X_n)\geq d_W=r(\mN_1)=\mbox{rank}(\mN_{\tilde{\mSigma}})
\end{equation}
From \eqref{eq:gk-lower} and \eqref{eq:gk:upper} we get that
\begin{equation}
    d^{GK}(X_1,\cdots,X_n)=\mbox{rank}(\mN_{\tilde{\mSigma}}) = r(\tilde{\mSigma})-\mbox{rank}(\tilde{\mSigma}).
\end{equation}
\end{proof}

\subsection{Proofs of Lemma~\ref{lm:range} and Lemma~\ref{lm:one}}
In this subsection, we give the proofs of the following two properties of covariance matrices and the common information, which are used in both Theorems~\ref{thm:seq} and ~\ref{thm:appr}.
\lemrange*
\begin{proof}[ of Lemma~\ref{lm:range}]
Note that singular values are non-negative. In the following, we show $\sigma_i, \rho_i \le 1$. Let $\mU\mLambda\mV^\top$ be the singular value decomposition of $\mSigma_{X}^{-1/2}\mSigma_{XY}\mSigma_{Y}^{-1/2}$, where $\mLambda=\text{diag}(\rho_1,\cdots,\rho_d)$. We denote $X' = \mU^\top\mSigma_{X}^{-1/2} X$, and $Y' = \mV^\top\mSigma_{Y}^{-1/2} Y$, $r_X = \text{rank}(\mSigma_X)$ and $r_Y =\text{rank}(\mSigma_Y)$. Then the covariance matrix $\mSigma'$ of $X', Y'$ are given by
$$
\mSigma_{X'} = \left[\begin{matrix}
    \mI_{r_X} & \vzero \\ \vzero & \vzero
\end{matrix}\right],
\mSigma_{Y'} = \left[\begin{matrix}
    \mI_{r_Y} & \vzero \\ \vzero & \vzero
\end{matrix}\right],
\mSigma_{X'Y'} = \mLambda.
$$
For $i \le \min\{r_X,r_Y\}$, we have that $\sigma_i=0$. By Cauchy-Schwarz inequality, if $i \le \min\{r_X,r_Y\}$, we have that $\sigma_i^2 = \text{cov}(X'_i,Y'_i)^2 \le \text{var}(X'_i)\text{var}(Y'_i) = 1$.
% The corresponding $\mSigma'$ is a valid (i.e. positive semidefinite) covariance matrix. Let $\mLambda^+$ denote the submatrix where $\sigma_i \ne 0$. If $\forall \sigma_i \in \mLambda^+: \sigma_i \leq 1$, then $\forall \left[ \begin{matrix}\vx \\ \vy \end{matrix}\right] \in \mathbb{R}^{r_X+r_Y}\setminus \{\vzero\}$ we have that
%     \begin{equation} \label{eq:range}
%     \begin{aligned}
%         [\vx^\top \vy^\top] \left[\begin{matrix}
%         \mI_{r_X} & \mLambda^{+\top} \\
%         \mLambda^+ & \mI_{r_Y}
%         \end{matrix} \right] \left[ \begin{matrix}
%         \vx \\ \vy \end{matrix}\right ] &= \|\vx\|_2^2+\|\vy\|_2^2 + 2\vx^\top \mLambda^{+\top} \vy\\
%         & \ge \|\vx\|_2^2+\|\vy\|_2^2 - 2|\vx^\top\vy|\geq  0.
%     \end{aligned}
% \end{equation}
% Since $\mU^\top\mSigma_X^{-1/2}$ and $\mV^\top\mSigma_Y^{-1/2}$ are invertible transformations, the original $\mSigma$ is also a valid covariance matrix.
\end{proof}

\lemone*
\begin{proof}[ of Lemma~\ref{lm:one}]
Similar to the proof of Lemma~\ref{lm:range}, let $\mU\mLambda\mV^\top$ be the singular value decomposition of $\mSigma_{X}^{-1/2}\mSigma_{XY}\mSigma_{Y}^{-1/2}$, where $\mLambda=\text{diag}(\rho_1,\cdots,\rho_d)$. We let $X' = \mU^\top\mSigma_{X}^{-1/2} X$, and $Y' = \mV^\top\mSigma_{Y}^{-1/2} Y$, $r_X = \text{rank}(\mSigma_X)$ and $r_Y =\text{rank}(\mSigma_Y)$. Then the covariance matrix $\mSigma'$ of $X', Y'$ are given by
$$
\mSigma_{X'} = \left[\begin{matrix}
    \mI_{r_X} & \vzero \\ \vzero & \vzero
\end{matrix}\right],
\mSigma_{Y'} = \left[\begin{matrix}
    \mI_{r_Y} & \vzero \\ \vzero & \vzero
\end{matrix}\right],
\mSigma_{X'Y'} = \mLambda.
$$
We notice that the rank of the covariance matrix $\text{rank}(\mSigma') = r_{X}+r_{Y}-\sum_i^{\min\{d_X, d_Y\}} \mathbbm{1}\{\sigma_i=1\}$. By Theorem~\ref{thm:main-gauss}, we have that $d(X',Y') = \text{rank}(\mSigma_{X'})+\text{rank}(\mSigma_{Y'})-\text{rank}(\mSigma')=\sum_i^{\min\{d_X, d_Y\}} \mathbbm{1}\{\sigma_i=1\}$. Since one-to-one transformations preserve the ranks of covariance matrices, $d(X,Y)=d(X',Y')= \sum_i^{\min\{d_X, d_Y\}} \mathbbm{1}\{\sigma_i=1\}$.
\end{proof}

\subsection{Proof of Theorem~\ref{thm:seq}}\label{sec:pf-thm1}
\thmseq*
\begin{proof}[ of Theorem~\ref{thm:seq}]
Let $d=\min\{d_X, d_Y\}$. We use $\{\sigma_i\}_{i=1}^d$ 
and $\{\rho_i(\epsilon)\}_{i=1}^d$ to denote the singular values of $\mSigma_X^{-1/2}\mSigma_{XY}\mSigma_{Y}^{-1/2}$ and $\mSigma_{X_\epsilon}^{-1/2}\mSigma_{X_\epsilon Y_\epsilon}\mSigma_{Y_\epsilon}^{-1/2}$, respectively, in a decreasing order. %To consider all possible values for $\bm{\rho}(\epsilon)$ we first prove the following lemma which shows that $\vzero \le \bm{\rho}(\epsilon) \le \vone$.
As definition \eqref{def: seq} requires that $|\rho_i(\epsilon) - \sigma_i| = \epsilon$, and from Lemma~\ref{lm:range} we have that $0 \le \rho_i(\epsilon) \le 1$, then $\rho_i(\epsilon) = 1-\epsilon$ whenever $\sigma_i = 1$. Let $k=\sum_i^d \mathbbm{1}\{\sigma_i=1\}$ and $l = \argmax\{\sigma_i | \forall i: \sigma_i\ne 1\}$. Since we assume $\{\sigma_i\}$ are in decreasing order, $l = k+1$. Recall that
\begin{equation}
    C(X_\epsilon, Y_\epsilon) = \frac{1}{2} \sum_{i=1}^{d} \log\frac{1+\rho_i(\epsilon)}{1-\rho_i(\epsilon)}
\end{equation}
Since $\frac{1+\rho_i(\epsilon)}{1-\rho_i(\epsilon)} \ge 0$ and is an increasing function of $\rho_i(\epsilon)$, then for $\epsilon \le {(1-\sigma_l)}/{2}$ we have that
\begin{equation}
\begin{aligned}
    \frac{k}{2} \log\frac{2-\epsilon}{\epsilon} & \le C(X_\epsilon, Y_\epsilon) \le \frac{k}{2} \log\frac{2-\epsilon}{\epsilon} + \frac{d-k}{2} \log\frac{1+\sigma_l+\epsilon}{1-\sigma_l-\epsilon}
\end{aligned}
\end{equation}
It follows that
\begin{equation}\label{eq:lim}
    \begin{aligned}
    \lim_{\epsilon \downarrow 0} \frac{C(X_\epsilon, Y_\epsilon)}{\frac{1}{2}\log(1/\epsilon)}  & = \lim_{\epsilon \downarrow 0} \frac{\frac{k}{2} \log\frac{2-\epsilon}{\epsilon}}{\frac{1}{2}\log(1/\epsilon)}
    = \lim_{\epsilon \downarrow 0} \frac{k\log\frac{1}{\epsilon}}{\log(1/\epsilon)}
    = k 
    \end{aligned}
\end{equation}
As Lemma~\ref{lm:one} proves that $k=d(X,Y)$, this concludes the proof of Theorem~\ref{thm:seq}.
\end{proof}

\subsection{Proof of Theorem~\ref{thm:appr}} \label{sec:pf-thm2}
\thmappr*
\begin{proof}[ of Theorem~\ref{thm:appr}]
     We recall that $C_\epsilon(X,Y)$ is defined as the optimal value of the optimization problem
\begin{equation}\label{eq:general}%\tag{\ref{def: opt}}
    \begin{aligned}
        \min_{\hat{\mSigma}} \quad & C(\hat{X},\hat{Y}) \\ 
        \text{s.t.} \quad & \|\mSigma - \hat{\mSigma}\|_F \le \epsilon.
    \end{aligned}
\end{equation}
Due to the difficulty in finding a closed-form solution of the optimization problem, we instead derive an upper and a lower bound on $C_\epsilon(X,Y)$ that have the same asymptotic behaviors.

\textbf{Upper bound.} We derive the upper bound by constructing a feasible solution. We first choose the marginal covariance matrices to be $\mSigma_{\hat{X}}=\mSigma_X, \mSigma_{\hat{Y}} = \mSigma_Y$. In the following, we provide our choice for the cross-covariance matrix $\mSigma_{\hat{X},\hat{Y}}$.

Let $\mU\mLambda\mV^\top$, $\hat{\mU}\hat{\mLambda}\hat{\mV}^\top$ denote the singular value decompositions of
$\mSigma_{X}^{-1/2}\mSigma_{XY}\mSigma_{Y}^{-1/2}$ and $\mSigma_{X}^{-1/2}\mSigma_{\hat{X}\hat{Y}}\mSigma_{Y}^{-1/2}$ respectively, where $\mLambda$, $\hat{\mLambda}$ are diagonal matrices containing the singular values, and $\mU, \hat{\mU}, \mV, \hat{\mV}$ are orthonormal matrices. 
We choose $\hat{\mU} = \mU$ and $\hat{\mV} = \mV$. Then $\|\mU(\mLambda - \hat{\mLambda})\mV^\top\|_F = \|\mLambda - \hat{\mLambda}\|_F$, because the Frobenius norm is unitary invariant. If we let $s=\|\mSigma_X^{1/2}\|_F\|\mSigma_Y^{1/2}\|_F$, by Cauchy–Schwarz inequality of the Frobenius norm, we have that
\begin{equation}
    \begin{aligned}
    \|\mSigma_{\hat{X}\hat{Y}} - \mSigma_{XY}\|_F & = \|\mSigma_X^{1/2}(\mLambda - \hat{\mLambda})\mSigma_{Y}^{1/2}\|_F \\
    & \le s \|(\mLambda - \hat{\mLambda})\|_F.
    \end{aligned}
\end{equation}
Hence, choosing $\hat{\mLambda}$ satisfying $\|\mLambda - \hat{\mLambda}\|_F \le \frac{1}{s}\epsilon$ provides a feasible solution. Among all such values of $\hat{\mLambda}$, we choose the one minimizing the common information $C(\hat{X},\hat{Y})$. In particular, we choose $\hat{\mLambda}$ to be the solution of
\begin{align}\label{eq:simp_up}
    S_{\frac{\epsilon}{s}}(X,Y):= \min_{\|\mLambda - \hat{\mLambda} \|_F \leq \frac{1}{s}\epsilon} C(\hat{X},\hat{Y})
\end{align}
As a result, we get that
\begin{equation}\label{eq:upper_bnd-2}
    C_\epsilon(X,Y)\leq S_{\frac{\epsilon}{s}}(X,Y).
\end{equation}
In the following, we provide a lower bound on $C_\epsilon(X,Y)$ using $S_{\epsilon}(X,Y)$. We analyze the asymptotic behavior of $S_{\epsilon}(X,Y)$ at the end of the proof.

\textbf{Lower bound.} Next, we derive a lower bound on the optimal objective. Note that the marginal distributions of $X$ and $Y$ might also be singular (i.e., $\text{rank}(\mSigma_X)<d_X$ and $\text{rank}(\mSigma_Y)<d_Y$). In this case, there exist one-to-one transformations $X' = \mP_X X$ and $Y' = \mP_Y Y$ (with rows from the Identity Matrix) that select a subset of elements in $X$ and $Y$ such that $\mSigma_{X'}$ and $\mSigma_{Y'}$ are full-rank. If we also apply the same transformation $\hat{X}' = \mP_X \hat{X}$ and $\hat{Y}' = \mP_Y \hat{Y}$, then $C(\hat{X}',\hat{Y}')=C(\hat{X},\hat{Y})$ because common information is preserved under one-to-one linear transformations \cite{satpathy2015gaussian}. Then the objective of \eqref{eq:general} remains the same and the constraints are relaxed, therefore solving the problem \eqref{eq:general} on $X',Y'$ would give a lower bound on $C_\epsilon(X,Y)$. In the remaining part of the proof, we assume without loss of generality that $\mSigma_X$ and $\mSigma_Y$ are full-rank. 

%In the following, we use $c$ to denote a general-purpose constant. With abuse of notation, $c$ may refer to different values in different equations.
The proof relies on relaxing the constraints in \eqref{eq:general} to result in a smaller optimal value. The first relaxation we make is the following  
\begin{align}
    \min_{\hat{\mSigma}} \quad & C(\hat{X},\hat{Y}) \label{eq:general-relax-1} \\ 
    \text{s.t.} \quad & \|\mSigma_X - \mSigma_{\hat{X}}\|_F \le \epsilon \nonumber\\
    & \|\mSigma_Y - \mSigma_{\hat{Y}}\|_F \le \epsilon \nonumber\\
    & \|\mSigma_{XY} - \mSigma_{\hat{X}\hat{Y}}\|_F \le \epsilon.\label{eq:last_cons}
\end{align} 
In the following, we further relax the last constraint \eqref{eq:last_cons} to a bound on the difference between the singular values using the following lemma, which will be proved in the end of this section.

\begin{restatable}{lemma}{lemrelax}\label{lm:relax}
    If covariance matrices satisfy $\|\mSigma_X - \mSigma_{\hat{X}}\|_F \le \epsilon$,$\|\mSigma_Y - \mSigma_{\hat{Y}}\|_F \le \epsilon$, and $\|\mSigma_{XY} - \mSigma_{\hat{X}\hat{Y}}\|_F \le \epsilon$, then $$\|\mLambda - \hat{\mLambda}\|_F \le c\epsilon, \text{for some constant }c.$$
\end{restatable}

%Hence, we have shown that the constraints in \eqref{eq:general} imply \eqref{eq:new_const}. 
As a result, the following optimization problem has a larger feasible set, hence, lower optimal value
\begin{align}
    S_{c\epsilon}(X,Y):= \min_{\|\mLambda - \hat{\mLambda} \|_F \leq c\epsilon} C(\hat{X},\hat{Y}) \label{eq:general-relax-2}
\end{align}
% \begin{align}
%     S_{c\epsilon}(X,Y):= \min_{\hat{\mLambda}} \quad & C(\hat{X},\hat{Y}) \label{eq:general-relax-2} \\ 
%     \text{s.t.} \quad &\|\mLambda - \hat{\mLambda} \|_F \leq c\epsilon. \nonumber
% \end{align}
Combining with \eqref{eq:upper_bnd-2}, we have shown that 
\begin{equation}\label{eq:bounds_general}
        S_{c\epsilon}(X,Y)\le C_{\epsilon}(X,Y) \le S_{\frac{\epsilon}{s}}(X,Y).
\end{equation}

\textbf{Asymptotic behavior of $S_{\epsilon}$.} In the following, we analyze the asymptotic behavior of $S_\epsilon$. Let $d = \min\{d_X,d_Y\}$, and $\mLambda = \text{diag}\{\sigma_1,\sigma_2,\cdots,\sigma_d\}$, $\hat{\mLambda}=\text{diag}\{\rho_1, \rho_2, \cdots,\rho_d\}$  with $\sigma_1 \ge \sigma_2 \ge \cdots \ge \sigma_d$ and $\rho_1 \ge \rho_2 \ge \cdots \ge \rho_d$. By Lemma~\ref{lm:range}, we have that $0 \le \sigma_i, \rho_i \le 1$. Then $S_\epsilon(X,Y)$ can be rewritten as
\begin{equation}\label{eq:opt}
    \begin{aligned}
        \min_{\rho} \quad & \frac{1}{2} \sum_{i=1}^{d} \log\frac{1+\rho_i}{1-\rho_i} \\ 
        \text{s.t.} \quad & \sum_{i=1}^{d} (\sigma_i -\rho_i)^2 \le \epsilon^2; ~0 \le \rho_i \le 1, \forall i \in [d].
    \end{aligned}
\end{equation}

Next, let $k=\sum_i^d \mathbbm{1}\{\sigma_i=1\}$, we will show that $\lim_{\epsilon \downarrow 0} \frac{S_\epsilon(X,Y)}{\frac{1}{2} \log(1/\epsilon)} = k$ using upper and lower bounds. %, using an upper bound $\overline{S}_\epsilon(X,Y) \ge S_\epsilon(X,Y)$ and a lower bound $\underline{S}_\epsilon(X,Y) \le S_\epsilon(X,Y)$. 
Define $\tilde{\rho} \in \mathbb{R}^d$ as
\begin{equation}\label{eq:rho_achievable}
    \tilde{\rho}_{i} = \begin{cases} \sigma_i, &~\text{if }\sigma_i \ne 1 \\ 1-\frac{\epsilon}{\sqrt{k}}, &~\text{if }\sigma_i = 1.\end{cases}
\end{equation}
Since $\tilde{\rho}$ is a feasible solution of problem (\ref{eq:opt}), we have that $S_\epsilon(X,Y) \le \frac{1}{2} \sum_{i=1}^{d} \log\frac{1+\tilde{\rho_i}}{1-\tilde{\rho_i}}$. Then we have
\begin{equation}
    \begin{aligned}\label{eq:k_upper}
    \lim_{\epsilon \downarrow 0} S_\epsilon(X,Y)
    \le &\lim_{\epsilon \downarrow 0} \frac{ \sum_{i=1}^{d} \log\frac{1}{1-\tilde{\rho_i}}}{\log(1/\epsilon)} + \frac{\sum_{i=1}^{d} \log{(1+\tilde{\rho_i})}}{\log(1/\epsilon)} \\
    = &\lim_{\epsilon \downarrow 0} \frac{\sum_{i=1}^{d} \log\frac{1}{1-\tilde{\rho_i}}}{\log(1/\epsilon)}\\
    = &\lim_{\epsilon \downarrow 0} \frac{\sum_{i:\sigma_i=1} \log\frac{\sqrt{k}}{\epsilon} + \sum_{i:\sigma_i\ne 1} \log\frac{1}{1-\sigma_i}}{\log(1/\epsilon)}\\
    = &\lim_{\epsilon \downarrow 0} \frac{k \log\frac{\sqrt{k}}{\epsilon}}{\log(1/\epsilon)} = k
    \end{aligned}
\end{equation}

To derive the lower bound, we consider the following variant optimization problem
\begin{equation}\label{eq:lower}
    \begin{aligned}
        \min_{\rho} \quad & \frac{1}{2} \sum_{i: \sigma_i=1} \log\frac{1}{1-\rho_i} \\ 
        \text{s.t.} \quad & \sum_{i: \sigma_i=1} (\sigma_i -\rho_i)^2 \le \epsilon^2; ~0 \le \rho_i \le 1, \forall i \in [d]
    \end{aligned}
\end{equation}
where the feasible set is relaxed, hence, gives a smaller optimal value, and the objective function is the sum of a subset of the terms in (\ref{eq:opt}). Since $\log\frac{1+\rho_i}{1-\rho_i} \ge 0$ and $\log(1+\rho_i) \ge 0, \forall \rho_i$, the optimal value of problem \eqref{eq:lower} is a lower bound on $S_\epsilon(X,Y)$. 

In the following, we solve problem (\ref{eq:lower}) and show that the optimal solution $\rho_i^* = \tilde{\rho}_i$. First, we observe that by Jensen's inequality, for all ${\rho_i}$ in the feasible set of problem \eqref{eq:lower},
\begin{align}
    -\sum_{i: \sigma_i=1} \log{(1-\rho_i)^2} & \geq -k \log(\frac{\sum_{i: \sigma_i=1}(1-\rho_i)^2}{k}) \stackrel{(a)}{\geq} -k \log (1-\tilde{\rho}_i)^2,
\end{align}
where the inequality $(a)$ follows from the fact that ${\rho_i}$ is in the feasible set of Problem \eqref{eq:lower}. This shows that $\rho_i^* = \tilde{\rho}_i$ is an optimal solution for Problem \eqref{eq:lower}. This shows that $\frac{1}{2} \sum_{i: \sigma_i=1} \log\frac{1}{1-\tilde{\rho_i}} \le S_\epsilon(X,Y)$. Using the similar calculations as in \eqref{eq:k_upper} , we can also obtain that
\begin{equation}\label{eq:k_lower}
    \lim_{\epsilon \downarrow 0} \frac{S_\epsilon(X,Y)}{\frac{1}{2}\log(1/\epsilon)} \ge k
\end{equation}
Combining \eqref{eq:k_upper}, \eqref{eq:k_lower} and \eqref{eq:bounds_general}, we have shown that 
\begin{equation}\label{eq:k}
    \lim_{\epsilon \downarrow 0} \frac{C_\epsilon(X,Y)}{\frac{1}{2}\log(1/\epsilon)} = k.
\end{equation}

Applying Lemma~\ref{lm:one}, we can now conclude that
\begin{equation}\label{eq:I}
    \lim_{\epsilon \downarrow 0} \frac{C_\epsilon(X,Y)}{\frac{1}{2}\log(1/\epsilon)} = d(X,Y)
\end{equation}

\end{proof}

\lemrelax*
\begin{proof}[ of Lemma~\ref{lm:relax}]
In this proof, we use $c$ to denote a general-purpose constant that depends on $\mSigma_X,\mSigma_Y$. With abuse of notation, $c$ may refer to different values in different equations. 
Recall that $\mLambda$ and $\hat{\mLambda}$ are the diagonal matrices containing the singular values of $\mSigma_X^{-1/2}\mSigma_{XY}\mSigma_Y^{-1/2}$ and $\mSigma_{\hat{X}}^{-1/2}\mSigma_{\hat{X}\hat{Y}}\mSigma_{\hat{Y}}^{-1/2}$ respectively. We notice that
\begin{align}
    \|\mLambda - \hat{\mLambda}\|_F & \leq \|\mSigma_X^{-1/2}\mSigma_{XY}\mSigma_Y^{-1/2} - \mSigma_{\hat{X}}^{-1/2}\mSigma_{\hat{X}\hat{Y}}\mSigma_{\hat{Y}}^{-1/2}\|_F \label{leq:1}\\
    &\leq \|(\mSigma_X^{-1/2}-\mSigma_{\hat{X}}^{-1/2})\mSigma_{XY}\mSigma_Y^{-1/2}\|_F  + \|\mSigma_{\hat{X}}^{-1/2}(\mSigma_{XY}\mSigma_Y^{-1/2} - \mSigma_{\hat{X}\hat{Y}}\mSigma_{\hat{Y}}^{-1/2})\|_F \label{leq:2}\\
    & \leq \|\mSigma_X^{-1/2}-\mSigma_{\hat{X}}^{-1/2}\|_F \|\mSigma_{XY}\|_F\|\mSigma_Y^{-1/2}\|_F  + \|\mSigma_{\hat{X}}^{-1/2}\|_F\|\mSigma_{XY}\mSigma_Y^{-1/2} - \mSigma_{\hat{X}\hat{Y}}\mSigma_{\hat{Y}}^{-1/2}\|_F \label{leq:3}\\
    & \leq \|\mSigma_X^{-1/2}-\mSigma_{\hat{X}}^{-1/2}\|_F \|\mSigma_{XY}\|_F\|\mSigma_Y^{-1/2}\|_F + \|\mSigma_{\hat{X}}^{-1/2}\|_F\|(\mSigma_{XY} - \mSigma_{\hat{X}\hat{Y}})\mSigma_{Y}^{-1/2}\|_F \nonumber\\
    & + \|\mSigma_{\hat{X}}^{-1/2}\|_F\|\mSigma_{\hat{X}\hat{Y}}(\mSigma_Y^{-1/2} - \mSigma_{\hat{Y}}^{-1/2})\|_F \label{leq:4}\\
    &\leq \|\mSigma_X^{-1/2}-\mSigma_{\hat{X}}^{-1/2}\|_F \|\mSigma_{XY}\|_F\|\mSigma_Y^{-1/2}\|_F \nonumber + \|\mSigma_{\hat{X}}^{-1/2}\|_F\|\mSigma_{XY} - \mSigma_{\hat{X}\hat{Y}}\|_F \|\mSigma_{Y}^{-1/2}\|_F\nonumber \\
    & + \|\mSigma_{\hat{X}}^{-1/2}\|_F\|\mSigma_{\hat{X}\hat{Y}}\|_F \|\mSigma_Y^{-1/2} - \mSigma_{\hat{Y}}^{-1/2}\|_F \label{leq:5}\\
    &\leq c \|\mSigma_X^{-1/2}-\mSigma_{\hat{X}}^{-1/2}\|_F + c\epsilon  \|\mSigma_{\hat{X}}^{-1/2}\|_F  + \|\mSigma_{\hat{X}}^{-1/2}\|_F\|\mSigma_{\hat{X}\hat{Y}}\|_F \|\mSigma_Y^{-1/2} - \mSigma_{\hat{Y}}^{-1/2}\|_F \label{leq:6},
\end{align}
where $c$ is a constant that may depend on $\mSigma_X,\mSigma_Y,\mSigma_{XY}$, \eqref{leq:1} uses Von Neumann’s trace inequality \cite{von1937some}, \eqref{leq:2} is obtained by subtracting and adding $\mSigma_{\hat{X}}^{-1/2}\mSigma_{XY}\mSigma_{Y}^{-1/2}$ then applying the triangle inequality, similarly, \eqref{leq:4} is obtained by subtracting and adding $\mSigma_{\hat{X}\hat{Y}}\mSigma_{Y}^{-1/2}$, \eqref{leq:3} and \eqref{leq:5} use the Cauchy–Schwarz inequality, and \eqref{leq:6} uses the bound $\|\mSigma_{XY} - \mSigma_{\hat{X}\hat{Y}}\|_F\le \epsilon$.

In the following, we will bound each term in \eqref{leq:6}. We first show that $\|\mSigma_{\hat{X}}^{-1/2}\|_F$ is bounded. We have that
\begin{align}
    \|\mSigma_{\hat{X}}^{-1/2}\|_F \leq \sqrt{d_X} \sigma_{d_X}(\mSigma_{\hat{X}})^{-1/2},
\end{align}
where $\sigma_{d_X}(\mSigma_{\hat{X}})$ denote the smallest singular value of the matrix $\mSigma_{\hat{X}}$. From the bound $\|\mSigma_{XY} - \mSigma_{\hat{X}\hat{Y}}\|_F\le \epsilon$, we have that
\begin{align}
    \sigma_{d_X}(\mSigma_{\hat{X}}) \geq \sigma_{d_X}(\mSigma_X) - \epsilon \geq \sigma_{d_X}(\mSigma_X)/2
\end{align}
for $\epsilon\leq \sigma_{d_X}(\mSigma_X)/2$. Hence, we have that 
\begin{equation}
    \|\mSigma_{\hat{X}}^{-1/2}\|_F\leq c,
\end{equation}
where $c$ is a constant that may depend on $\mSigma_X,\mSigma_Y,\mSigma_{XY}$. Similarly, 
\begin{align}
    \|\mSigma_{\hat{X}}^{-1}\|_F\leq c.
\end{align}
By the triangle inequality, we also have that
\begin{equation}
    \|\mSigma_{\hat{X}\hat{Y}}\|_F\leq \|\mSigma_{XY}\|_F + \|\mSigma_{XY} - \mSigma_{\hat{X}\hat{Y}}\|_F\leq  c.
\end{equation}
Substituting in \eqref{leq:6} we get that
\begin{equation}\label{eq:bnd_2}
    \|\mLambda - \hat{\mLambda} \|_F \leq c(\epsilon + \|\mSigma_X^{-1/2}-\mSigma_{\hat{X}}^{-1/2}\|_F + \|\mSigma_Y^{-1/2}-\mSigma_{\hat{Y}}^{-1/2}\|_F).
\end{equation}
It remains to bound $\|\mSigma_{\hat{X}}^{-1/2}-\mSigma_{X}^{-1/2}\|_F$, and similarly $\|\mSigma_{\hat{Y}}^{-1/2}-\mSigma_{Y}^{-1/2}\|_F$. We have that 
\begin{align}
    \mSigma_X^{-1/2}-\mSigma_{\hat{X}}^{-1/2} &= \mSigma_{X}^{-1/2}(\mSigma_X^{1/2}-\mSigma_{\hat{X}}^{1/2}) \mSigma_{\hat{X}}^{-1/2}\nonumber \\
    &= \mSigma_{X}^{-1/2}(\mSigma_X-\mSigma_{\hat{X}}) (\mSigma_X^{1/2}+\mSigma_{\hat{X}}^{1/2})^{-1} \mSigma_{\hat{X}}^{-1/2}.
\end{align}
Hence, we have that
\begin{align}
    \|\mSigma_X^{-1/2}-\mSigma_{\hat{X}}^{-1/2}\|_F &\leq \|\mSigma_{X}^{-1/2}\|_F \|\mSigma_X-\mSigma_{\hat{X}}\|_F \|(\mSigma_X^{1/2}+\mSigma_{\hat{X}}^{1/2})^{-1}\|_F \|\mSigma_{\hat{X}}^{-1/2}\|_F\nonumber \\
    &\leq c \epsilon / \sqrt{\sigma_{d_X}(\mSigma_X)}.
\end{align}
Similarly, we have that $\|\mSigma_Y^{-1/2}-\mSigma_{\hat{Y}}^{-1/2}\|_F\leq c \epsilon$. Substituting in \eqref{eq:bnd_2} we get that 
\begin{equation}\label{eq:new_const}
    \|\mLambda - \hat{\mLambda} \|_F \leq c\epsilon.
\end{equation}
\end{proof}

\subsection{Proof of Theorem~\ref{thm:quant}}
\thmquant*
\begin{proof}[ of Theorem~\ref{thm:quant}]
Since it is hard in general to directly solve Wyner's common information even for discrete variables, we prove the result by matching upper and lower bounds. First, we show that $\lim_{m \to \infty} \frac{C(\langle X\rangle_m, \langle Y \rangle_m)}{\log m} \ge d(X,Y)$ using the property that mutual information is a lower bound of Wyner's common information \cite{wyner1975common}. Following the definition of mutual information and the R{\'e}nyi dimension in \eqref{eq:renyi}, we can derive that
\begin{equation}\label{eq:quant_low}
    \begin{aligned}
        \lim_{m \to \infty} \frac{C(\langle X\rangle_m, \langle Y \rangle_m)}{\log m} & \ge \lim_{m \to \infty} \frac{I(\langle X\rangle_m; \langle Y \rangle_m)}{\log m} \\
        & = \lim_{m \to \infty} \frac{H(\langle X\rangle_m)}{\log m} + \lim_{m \to \infty}\frac{H(\langle Y \rangle_m)}{\log m} - \lim_{m \to \infty}\frac{H(\langle X\rangle_m, \langle Y \rangle_m)}{\log m} \\
        & = d^R(X) + d^R(Y) - d^R(X,Y)\\
        & \stackrel{\cite{renyi1959dimension}}{=} \mbox{rank}(\mSigma_X) + \mbox{rank}(\mSigma_Y) - \mbox{rank}(\mSigma) \\
        & = d(X,Y),
    \end{aligned}
\end{equation}
where the last equality is proved in our Theorem~\ref{thm:main-gauss}.

% Next we prove that $\lim_{m \to \infty} \frac{C(\langle X\rangle_m, \langle Y \rangle_m)}{\log m} \le d(X,Y)$. We first show that the common information between the quantized $\langle X\rangle_m$ and $\langle Y \rangle_m$ is a lower bound of the common information between the quantized $\langle X\rangle_{m'}$ and $\langle Y \rangle_{m'}$ for $m' > m$. This is because the quantization operator is a non-expanding operator, i.e., $C(\langle X\rangle_m, \langle Y \rangle_m) \le C(\langle X\rangle_{m'}, \langle Y \rangle_{m'})$. Then we can derive that

Next, we show  that $\lim_{m \to \infty} \frac{C(\langle X\rangle_m, \langle Y \rangle_m)}{\log m} \le d(X,Y)$. Recall (from the proofs of Lemma~\ref{lm:range} and \ref{lm:one}) that there exist invertible linear transformations $\mT_X$ and $\mT_Y$ such that $X = \mT_X X'$ and $Y = \mT_Y Y'$, where $X' \in \mathbb{R}^{d_X}$ and $Y' \in \mathbb{R}^{d_Y}$ are another pair of Gaussian random vectors with independent elements (i.e. their covariance matrices $\mSigma_{X'}$, $\mSigma_{Y'}$, and $\mSigma_{X'Y'}$ are all diagonal). Note that $C(X,Y) = C(X', Y')$ and $d(X,Y) = d(X',Y')$ since $\mT_X$ and $\mT_Y$ are invertible transformations. However, $C(\langle X\rangle_m,\langle Y \rangle_m) \ne C(\langle X'\rangle_{m}, \langle Y' \rangle_{m})$, because $\langle \mT_X \langle X' \rangle_{m'} \rangle_m \ne \langle \mT_X X' \rangle_m = \langle X \rangle_m$ can happen (similarly for $Y'$ and $Y$).

In the remaining of the proof, we will first show in Proposition~\ref{prop:quant} that $\lim_{m \to \infty}\frac{C(\langle X'\rangle_{m}, \langle Y' \rangle_{m})}{\log m} \le d(X',Y')$ because of the special structure of the covariance matrices $\mSigma_{X'}$ and $\mSigma_{Y'}$. Then we will show that $\lim_{m \to \infty}\frac{C(\langle X\rangle_m,\langle Y \rangle_m)}{\log m} \le \lim_{m \to \infty}\frac{C(\langle X'\rangle_{m}, \langle Y' \rangle_{m})}{\log m}$. 

\begin{restatable}{proposition}{propQuant}\label{prop:quant}
    Let $\{(X_i,Y_i)\}_n$ be independent pairs of Gaussian random variables with $n < \infty$, then
\begin{equation}
    \lim_{m \to \infty}\frac{C(\langle X\rangle_{m}, \langle Y \rangle_{m})}{\log m} \le d(X,Y)
\end{equation}
\end{restatable}
\begin{proof}[ of Proposition~\ref{prop:quant}]
    \begin{equation}
        \begin{aligned}
            \lim_{m \to \infty} \frac{C(\langle X\rangle_m, \langle Y \rangle_m)}{\log m} 
            & \stackrel{(i)}{=} \lim_{m \to \infty} \frac{\sum_{i=1}^{n}C(\langle X_i\rangle_m, \langle Y_i \rangle_m)}{\log m}\\
            & = \lim_{m \to \infty} \frac{\sum_{i: X_i=Y_i} C(\langle X_i\rangle_m, \langle Y_i \rangle_m) + \sum_{i: X_i \ne Y_i} C(\langle X_i\rangle_m, \langle Y_i \rangle_m)}{\log m} \\
            & \stackrel{(ii)}{\le} \lim_{m \to \infty} \frac{\sum_{i: X_i=Y_i} H(\langle X_i \rangle_m) + \sum_{i: X_i \ne Y_i} c}{\log m}\\
            & \stackrel{(iii)}{=} d(X,Y),
        \end{aligned}
    \end{equation}
where the $c$ is an absolute constant. The equality (i) is because $\{(X_i,Y_i)\}_n$ are independent pairs; the inequality (ii) is due to the property that $C(\langle X_i\rangle_m, \langle Y_i \rangle_m) \le \min\{H(\langle X_i \rangle_m), H(\langle Y_i \rangle_m)\}$ and \eqref{eq:wyner}; and the last equality (iii) follows the definition of R{\'e}nyi dimension \eqref{eq:renyi} and Theorem~\ref{thm:main-gauss}
\end{proof}

Let $m' = \alpha m$. To connect $C(\langle X\rangle_m,\langle Y \rangle_m)$ and $C(\langle X'\rangle_{\alpha m}, \langle Y' \rangle_{\alpha m})$, we focus on two subsets $\sA_m=\{\vx\in\mathbb{R}^{d_X}: \langle \mT_X \langle \vx \rangle_{\alpha m} \rangle_m = \langle \mT_X \vx \rangle_m\}$ and $\sB_m=\{\vy \in \mathbb{R}^{d_Y}: \langle \mT_Y \langle \vy \rangle_{\alpha m} \rangle_m = \langle \mT_Y \vy \rangle_m\}$. We also define an indicator variable $Z_m$ as follows:
\begin{equation}
    Z_m = \begin{cases}
        1, & \text{if } X' \in \sA_m \text{ and } Y' \in \sB_m, \\
        0, & \text{otherwise}.
    \end{cases}
\end{equation}
We consider the conditional common information \cite{lapidoth2016conditional} $C(\langle X\rangle_m,\langle Y \rangle_m|Z_m)$ that is defined as
\begin{equation} \label{def:cond_wyner}
    C(\langle X\rangle_m,\langle Y \rangle_m|Z_m) := \min_{\langle X\rangle_m - (W_{Zm}, Z_m) - \langle Y \rangle_m} \quad I(\langle X\rangle_m\;\langle Y\rangle_m;W_{Zm} | Z_m).
\end{equation}
We can use it to upper bound the common information as follows. Recall from \eqref{def:wyner}, the common information is defined as 
$$C(\langle X\rangle_m,\langle Y \rangle_m) := \min_{\langle X\rangle_m - W_{m} - \langle Y \rangle_m} \quad I(\langle X\rangle_m\;\langle Y\rangle_m;W_{m})$$
Let $W_{Zm}^*$ be the optimal solution of $C(\langle X\rangle_m,\langle Y \rangle_m|Z_m)$. $(W_{Zm}^*,Z_m)$ is a feasible solution of $C(\langle X\rangle_m,\langle Y \rangle_m)$ that satisfies $ \langle X\rangle_m - (W_{Zm}^*, Z_m) - \langle Y \rangle_m$ and 
\begin{equation}\label{eq:cond_ci_1}
\begin{aligned}
    I(\langle X\rangle_m\;\langle Y\rangle_m; W_{Zm}^*, Z_m) & = I(\langle X\rangle_m\;\langle Y\rangle_m; Z_m) + I(\langle X\rangle_m\;\langle Y\rangle_m; W_{Zm}^* \mid Z_m) \\ 
    & \le H(Z_m) + I(\langle X\rangle_m\;\langle Y\rangle_m; W_{Zm}^* \mid Z_m).
\end{aligned}
\end{equation}
Therefore, we have
\begin{equation}\label{eq:cond_ci}
    \begin{aligned}
        C(\langle X\rangle_m,\langle Y \rangle_m) \le & C(\langle X\rangle_m,\langle Y \rangle_m | Z_m) + H(Z_m)\\
        = & C(\langle X\rangle_m,\langle Y \rangle_m |Z_m=1)\mathbb{P}(Z_m=1) + C(\langle X\rangle_m,\langle Y \rangle_m |Z_m=0)\mathbb{P}(Z_m=0) \\
        & + H(Z_m),
    \end{aligned}
\end{equation}
where the equality expands the definition of conditional common information. 

Since we are interested in the asymptotic behavior of $C(\langle X\rangle_m,\langle Y \rangle_m)$, we divide both sides of \eqref{eq:cond_ci} by $\log m$ and take the limit of $m \to \infty$. In the following, we examine the limit of each term on the right-hand side. First, we have that 
\begin{equation}\label{eq:upper1}
    \begin{aligned}
        \lim_{m\to\infty}\frac{C(\langle X\rangle_m,\langle Y \rangle_m |Z_m=1)\mathbb{P}(Z_m=1)}{\log m} & \stackrel{(i)}{\le} \lim_{m\to\infty}\frac{C(\langle X\rangle_m,\langle Y \rangle_m |Z_m=1)}{\log m} \\
        & \stackrel{(ii)}{=} \lim_{m\to\infty}\frac{C(\langle \mT_X \langle X' \rangle_{\alpha m} \rangle_m, \langle \mT_Y \langle Y' \rangle_{\alpha m} \rangle_m|Z_m=1)}{\log m} \\
        & \stackrel{(iii)}{\le} \lim_{m\to\infty}\frac{C(\langle \mT_X \langle X' \rangle_{\alpha m} \rangle_m, \langle \mT_Y \langle Y' \rangle_{\alpha m} \rangle_m)}{\log m} \\
        & \stackrel{(iv)}{\le} \lim_{m\to\infty}\frac{C(\langle X' \rangle_{\alpha m}, \langle Y' \rangle_{\alpha m})}{\log m} \\
        & \stackrel{\text{Proposition~\ref{prop:quant}}}{\le} d(X,Y), 
    \end{aligned}
\end{equation}
where the inequality $(i)$ follows that $\mathbb{P}(Z_m=1) \le 1$; the equality $(ii)$ is due to the fact that $Z_m=1$ implies that $X' \in \sA_m$ and $Y' \in \sB_m$; the inequality $(iii)$ is because if $\langle \mT_X \langle X' \rangle_{\alpha m} \rangle_m - W - \langle \mT_Y \langle Y' \rangle_{\alpha m} \rangle_m$ forms a Markov chain for some random variable $W$, then conditioned on the event that $Z_m=1$, $\langle \mT_X \langle X' \rangle_{\alpha m} \rangle_m - W - \langle \mT_Y \langle Y' \rangle_{\alpha m} \rangle_m$ also forms a Markov chain (one proof can be found in Lemma~12 in \cite{yu2020exact}); the inequality $(iv)$ follows the data processing inequality of common information and the last inequality applies Proposition~\ref{prop:quant}. Using one of the basic properties of Wyner's common information and the definition of the R{\'e}nyi dimension \eqref{eq:renyi}, we also have that
\begin{equation}\label{eq:upper2}
    \begin{aligned}
        \lim_{m\to\infty}\frac{C(\langle X\rangle_m,\langle Y \rangle_m |Z_m=0)}{\log m} & \le \lim_{m\to\infty}\frac{\min\{H(\langle X\rangle_m), H(\langle Y \rangle_m)\}}{\log m} \\
        & = \min\{d^R(X),d^R(Y)\}.
    \end{aligned}
\end{equation}
And we upper bound the limit of the probability $\mathbb{P}(Z_m=0)$ in Lemma~\ref{lm:z0} below, which we will prove after this Theorem.

\begin{restatable}{lemma}{lemZzero}\label{lm:z0}
    For random variables $X,Y,X',Y'$ and set $\sA_m, \sB_m$ defined as above, we denote $t = \max\{ \max \{|(\mT_X)_{i,j}|, \forall i,j \in [d_X]\}, \max \{|(\mT_Y)_{i,j}|, \forall i,j \in [d_Y]\}\}$, and $d=\max\{d_X, d_Y\}$. We have that
    \begin{equation}
        \lim_{m\to\infty} \mathbb{P}(Z_m=0) \le \frac{4t d^2}{\alpha},
    \end{equation}
\end{restatable} 

In addition, since $Z_m$ is a binary variable, $H(Z_m) \le 1$ and $\lim_{m\to\infty}\frac{H(Z_m)}{\log m} = 0$. Combining with \eqref{eq:cond_ci},\eqref{eq:upper1},\eqref{eq:upper2} and Lemma~\ref{lm:z0}, we have that for all $\alpha>0$
\begin{equation}
    \begin{aligned}
        \lim_{m \to \infty}\frac{C(\langle X\rangle_m,\langle Y \rangle_m)}{\log m}
        & \le d(X,Y) + \min\{d^R(X),d^R(Y)\} \cdot \frac{4td^2}{\alpha}
    \end{aligned}
\end{equation}
As this holds for all $\alpha>0$ taking the limit of $\alpha$ to infinity, we get $\lim_{\alpha \to \infty} \frac{4td^2}{\alpha} = 0$, and thus obtain that
\begin{equation}
    \lim_{m \to \infty} \frac{C(\langle X\rangle_m, \langle Y \rangle_m)}{\log m} \le d(X,Y).
\end{equation}

Therefore, combining with \eqref{eq:quant_low}, we can conclude that
\begin{equation}
        \lim_{m \to \infty} \frac{C(\langle X\rangle_m, \langle Y \rangle_m)}{\log m} = d(X,Y).
\end{equation}
\end{proof}

% \textcolor{purple}{
% \begin{restatable}{corollary}{coroQuant}\label{coro:quant}
%     \begin{equation}
%     \lim_{m \to \infty} \frac{C(\langle X\rangle_m, \langle Y \rangle_m)}{\log m} = \lim_{m \to \infty} \frac{C(\langle X'\rangle_m, \langle Y' \rangle_m)}{\log m} 
%     \end{equation}
% \end{restatable}
% }

\lemZzero*
\begin{proof}[ of Lemma~\ref{lm:z0}]
    Since the quantization is element-wise and $\langle X'_i \rangle_{\alpha m} = \frac{\floor{\alpha m X'_i}}{\alpha m}$, we know that $|\langle X'_i\rangle_{\alpha m} - X'_i| \le \frac{1}{\alpha m}$. Now for all $i \in [d_X]$, we can bound the imprecision of transformation caused by quantization as follows:
\begin{equation}\label{eq:bound_abs}
    \begin{aligned}
        |(\mT_X <X'>_{\alpha m})_i - X_i| & =  |(\mT_X)_i (<X'>_{\alpha m} - X')| \\
        & \le \sum_{j=1}^{d_X}  |(\mT_X)_{ij}| |<X'_j>_{\alpha m} - X'_j| \\
        & \le \frac{t d_X}{\alpha m} \le \frac{t d}{\alpha m}.
    \end{aligned}
\end{equation}
Similarly, for all $i\in [d_Y]$, we can bound $|(\mT_Y <Y'>_{\alpha m})_i - Y_i| \le \frac{t d}{\alpha m}$.

Let $\{\frac{k}{m}\}_{k\in \mathbb{Z}}$ denote the sequence of quantization boundaries of $\langle X_i \rangle_{m}$ and $\langle Y_i \rangle_{m}$. We define the set $\sC_m = \{x \in \mathbb{R}: \exists k \in \mathbb{Z}, \text{ s.t. } x \in [\frac{k}{m} + \frac{td}{\alpha m}, \frac{k+1}{m} - \frac{t d}{\alpha m}]\}$.
%\textcolor{cyan}{Here we still assume $\alpha>4td^2$ right?}
We observe that if $X_i \in \sC_m$ for some $k\in\mathbb{Z}$, then 
\begin{equation}
    \begin{aligned}
        (\mT_X <X'>_{\alpha m})_i - \frac{k}{m} \ge X_i - \frac{k}{m} - |(\mT_X <X'>_{\alpha m})_i - X_i| \ge 0,
    \end{aligned}
\end{equation}
where the second inequality is due to the assumption that $X_i \in [\frac{k}{m} + \frac{td}{\alpha m}, \frac{k+1}{m} - \frac{t d}{\alpha m}]$ and \eqref{eq:bound_abs}. Similarly $\frac{k+1}{m} - (\mT_X <X'>_{\alpha m})_i \ge 0$, which implies that $\langle (\mT_X \langle X' \rangle_{\alpha m})_i \rangle_m = \langle (\mT_X X')_i \rangle_m$. Therefore, we have that $\{\forall i \in [d_X], X_i \in \sC_m \text{ and } \forall i \in [d_Y], Y_i \in \sC_m\} \subset \{X' \in \sA_m \text{ and }Y' \in \sB_m \} = \{Z_m=1\}$, and
\begin{equation}\label{eq:z0}
    \begin{aligned}
        \mathbb{P}(Z_m=0) & \le \mathbb{P}(\{\cup_{i \in [d_X]} X_i \in \sC_m^c\} \cup \{\cup_{i \in [d_Y]} Y_i \in \sC_m^c\})\\
        & \le  \sum_{i=1}^{d_X}\mathbb{P}(X_i \in \sC_m^c) +  \sum_{i=1}^{d_Y}\mathbb{P}(Y_i \in \sC_m^c)\\
    \end{aligned}
\end{equation}

Taking $X_1$ as an example, we next show that $\lim_{m\to\infty}\mathbb{P}(X_1 \in \sC_m^c) \le \frac{2td}{\alpha}$. 
For $k \in \mathbb{Z}$, we use $\bar{f_k}(X_1)$ and $\underline{f_k}(X_1)$ denote the maximum and minimum values of the probability density $f(X_1)$ in the range $[\frac{k}{m}, \frac{k+1}{m}]$. We can write each term in the summation as
\begin{equation}
    \begin{aligned}
        \mathbb{P}(X_1 \in \sC^c) & = \sum_{k\in \mathbb{Z}} \mathbb{P}(X_1 \in \sC^c|X_1 \in [\frac{k}{m}, \frac{k+1}{m}]) \mathbb{P}(X_1 \in [\frac{k}{m}, \frac{k+1}{m}])\\
        & = \sum_{k\in \mathbb{Z}} \frac{\mathbb{P}(X_i \in [\frac{k}{m}, \frac{k}{m}+\frac{td}{\alpha m}]\cup[\frac{k+1}{m}-\frac{td}{\alpha m}, \frac{k+1}{m}])}{\mathbb{P}(X_i \in [\frac{k}{m},\frac{k+1}{m}])} \mathbb{P}(X_i \in [\frac{k}{m}, \frac{k+1}{m}]) \\
        & \le \sum_{k\in \mathbb{Z}} \frac{2tdm\cdot\bar{f_k}(X_1)}{\alpha m\cdot \underline{f_k}(X_1)} \mathbb{P}(X_i \in [\frac{k}{m}, \frac{k+1}{m}]).
    \end{aligned}
\end{equation}
Since $f(X_1)$ is a continuous function, $\lim_{m\to\infty}\frac{\bar{f_k}(X_1)}{\underline{f_k}(X_1)}=1$ for all $k \in \mathbb{Z}$. Then we can bound the limit as $\lim_{m\to\infty}\mathbb{P}(X_i \in \sC^c) \le \frac{2td}{\alpha} \sum_{k\in \mathbb{Z}} \mathbb{P}(X_i \in [\frac{k}{m}, \frac{k+1}{m}]) = \frac{2td}{\alpha}.$ Similar computation extends to all elements of $X$ and $Y$. Plugging this result into \eqref{eq:z0}, we can conclude that
\begin{equation}
    \lim_{m\to\infty} \mathbb{P}(Z_m=0) \le (d_X+d_Y)\frac{2td}{\alpha} = \frac{4td^2}{\alpha}.
\end{equation}
\end{proof}

%%%% Bibliography %%%%

\bibliographystyle{IEEEtran}
\bibliography{biblography}

\end{document}